\documentclass[aip,jmp,reprint,superscriptaddress,amsmath,amssymb]{revtex4-2}

\usepackage{amsthm}    
\usepackage{mathtools} 
\usepackage{bm}
\usepackage{graphicx}
\usepackage{booktabs}  
\usepackage{array}
\usepackage{xcolor}
\usepackage{subcaption} 
\makeatletter\DeclareCaptionJustification{flushing}{\flushing}\makeatother 
\usepackage{hyperref}

\newcommand{\sw}[2]{\prescript{}{#1}{#2}}

\newcommand{\ie}{\mbox{\it i.e.}}
\newcommand{\eg}{\mbox{\it e.g.}}

\newcommand{\etal}{\mbox{\it et al.}}

\newcommand{\SOthree}{\mathrm{SO}(3)}
\newcommand{\SOtwo}{\mathrm{SO}(2)}
\newcommand{\sphere}{S^2}
\newcommand{\Lsph}[1]{L^2_{#1}(\sphere)}       
\newcommand{\Lrot}{L^2(\SOthree)}
\newcommand{\Mackey}[1]{\mathcal{M}_{#1}}      
\newcommand{\Ind}{\operatorname{Ind}}

\newcommand{\lift}[1]{\Lambda_{#1}}            
\newcommand{\proj}[1]{\Lambda_{#1}^{\dagger}}        
\newcommand{\conv}{\star}                      
\newcommand{\Cop}[2]{\mathcal{C}_{#1\to #2}}   
\newcommand{\twist}{\hbar}                     
\newcommand{\Rot}{\mathcal{R}}                 
\newcommand{\piG}{\pi_{G}}
\newcommand{\pis}[1]{\pi_{#1}}

\newcommand{\dd}{\,\mathrm{d}}
\newcommand{\dmu}{\dd\mu}
\newcommand{\ii}{\mathrm{i}}
\DeclareMathOperator{\Id}{Id}

\theoremstyle{plain}
\newtheorem{theorem}{Theorem}[section]
\newtheorem{lemma}[theorem]{Lemma}
\newtheorem{proposition}[theorem]{Proposition}
\newtheorem{corollary}[theorem]{Corollary}
\theoremstyle{definition}
\newtheorem{definition}[theorem]{Definition}
\theoremstyle{remark}
\newtheorem{remark}[theorem]{Remark}

\makeatletter
\AtBeginDocument{%
  \bibpunct{[}{]}{,}{n}{}{,}%
  \def\NAT@cmprs{\@ne}\let\NAT@sort\@ne
}
\makeatother

\begin{document}

\title{Unifying spherical convolutions through lifting}

\author{Jason D. McEwen}
\email[Corresponding author: ]{jason.mcewen@ucl.ac.uk}
\affiliation{Mullard Space Science Laboratory, University College London, Surrey RH5 6NT, United Kingdom}
\affiliation{The Alan Turing Institute, Euston Road, London NW1 2DB, United Kingdom}

\date{\today}

\begin{abstract}
We develop a unified description of convolution on the sphere by lifting to the rotation group.  A number of definitions of spherical convolution have been proposed, differing in whether the kernel is axisymmetric or directional, whether the output is on the sphere or the rotation group, whether the signals carry spin, and in their computational cost.  No unified treatment exists.  Regarding the sphere as the homogeneous space $\sphere \cong \SOthree/\SOtwo$, we lift a spin $s$ field to the rotation group, convolve on $\SOthree$, and project back to the sphere through the adjoint of the lift for a potentially different spin $s'$.  We recover every rotation-equivariant convolution considered previously through constraints on the kernel, spin choices and inclusion or omission of the lift and projection.  For maps between fields of fixed input and output spin, we prove that directional kernels, output on the sphere and equivariance cannot be achieved simultaneously.  We make explicit that spin fields are sections of associated line bundles, linking the physics and bundle literatures, and express the spin rotation phase as the character of the twist of the section.  We present harmonic-space algorithms for the exact and efficient computation of lifted spherical convolutions of band-limited fields by appealing to sampling theorems on the sphere and rotation group.  All operators are implemented in the open-source \texttt{s2conv} code, and numerical experiments confirm the reductions and the equivariance of the operators to machine precision.  Finally, we illustrate the operators on fields measured on the Earth and on the celestial sky.
\end{abstract}

\maketitle

\section{Introduction}
\label{sec:intro}

In many fields of science and engineering data are measured on the sphere.  Applications where data are defined inherently on the sphere are found in cosmology (\eg\ the cosmic microwave background and its polarisation \cite{zaldarriaga1997all,kamionkowski1997statistics}, and gravitational lensing \cite{wallis2022mapping}), geophysics and planetary science, meteorology and climate science (\eg\ global weather forecasting \cite{bonev2023spherical}), computer vision \cite{cohen2018spherical,esteves2018learning} and biomedical imaging (\eg\ diffusion tensor imaging \cite{goodwin2022how}), to name only a few.  In many of these applications a harmonic analysis of the data is insightful and convolution is often the fundamental operation: it underlies filtering \cite{mcewen2008optimal}, wavelet analysis \cite{mcewen2007fast,wiaux2008exact,leistedt2013s2let,mcewen2015directional} and, more recently, convolutional neural networks and neural operators defined natively on the sphere \cite{cohen2018spherical,esteves2018learning,kondor2018clebsch,cobb2021efficient,esteves2020spin,ocampo2023scalable,bonev2023spherical}.  Moreover, many of the signals of interest are not only scalar but carry spin: tangent vector fields such as horizontal wind are spin $\pm1$ fields, the polarisation of the cosmic microwave background and cosmic shear of weak gravitational lensing are spin $\pm2$ fields, and the second flexion of weak gravitational lensing is a spin $\pm3$ field \cite{bacon2006weak} (Fig.~\ref{fig:spin-fields}); the values of all of these are defined relative to a local frame and acquire a phase when that frame is rotated.

\begin{figure}[t]
  \centering
  \includegraphics[width=\linewidth]{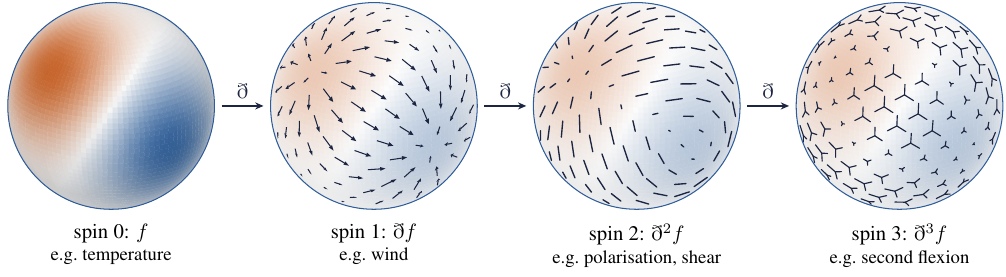}
  \caption{Spin fields on the sphere.  From left to right: a scalar (spin-$0$) field $f$, shown as colour, and the spin-$1$, spin-$2$ and spin-$3$ fields $\eth f$, $\eth^2 f$ and $\eth^3 f$ obtained from it by the spin raising operator, drawn in the tangent plane with glyphs whose symmetry is their spin.  Temperature, wind, polarisation or shear, and the second flexion of weak lensing are physical examples of spin $0$, $1$, $2$ and $3$ respectively.}
  \label{fig:spin-fields}
\end{figure}

A number of definitions of convolution on the sphere have been proposed in the literature, differing in the class of kernel supported, the domain on which the convolved signal is defined, whether spin is accommodated, and the cost of computation.  Axisymmetric convolution (\eg\ \cite{driscoll1994computing,wandelt2001fast,mcewen2007fast,kennedy2011azimuthally,esteves2018learning,bonev2023spherical}) restricts the kernel to be invariant under azimuthal rotation when centred on the North pole, so that the output remains on the sphere and may be computed simply by a product in harmonic space; directional kernels are, however, not supported.  Directional convolution (\eg\ \cite{wandelt2001fast,challinor2000all,mcewen2007fast,wiaux2006fast,wiaux2007complex,mcewen2013directional,cohen2018spherical,ocampo2023scalable}) relaxes this restriction to arbitrary kernels; however, the output is then defined on the rotation group $\SOthree$, with greater computational cost when all azimuthal orders are retained.  Convolution on the rotation group itself \cite{kostelec2008ffts,cohen2018spherical,kondor2018clebsch,cobb2021efficient} takes both field and kernel to be functions on $\SOthree$.  Spin convolutions extend the axisymmetric and directional constructions to spin signals \cite{mcewen2015directional,geller2010spin,geller2011mixed,esteves2020spin,tessore2025euclid,mulder2026spin}, while a further family of anisotropic constructions \cite{kennedy2011azimuthally,sadeghi2012commutative,khalid2012anisotropic,roddy2021sifting} recovers an output on the sphere for directional kernels, but is not in general equivariant under all rotations.  Roddy \& McEwen \cite{roddy2021sifting} compare isotropic, left, directional, commutative anisotropic and sifting constructions by support for directional kernels, spherical output and efficient computation, motivating sifting by the limitations of the preceding constructions in that comparison.  Several families are related by scalar or axisymmetric specialisations, but their differing conventions and representations obscure the full set of connections.  We make those connections explicit and distinguish the original criteria from full rotation equivariance.

We seek a common description in the language of homogeneous spaces.  The sphere is the quotient $\SOthree/\SOtwo$ of the rotation group by the stabiliser of a point, and fields on the sphere carrying spin are sections of the vector bundles associated to this quotient, which may be encoded as functions on the group, the so-called Mackey functions, on which the group acts by the representation induced from the stabiliser \cite{mackey1952induced,folland2016course}.  In this language, equivariant linear maps between fields admit group-convolution representations \cite{kondor2018generalization,cohen2019general}.  Cohen, Geiger \& Weiler \cite{cohen2019general} and Xu \etal\ \cite{xu2022unified} already construct layers by lifting from $G/H$ to $G$, convolving, and returning to $G/H$, with possibly different stabiliser representations at input and output.  Related local-frame formulations are developed in \cite{cohen2019gauge,weiler2021coordinate}.

In this article we show that all of the rotation-equivariant spherical convolutions listed above are special cases of a single construction, which we call the \emph{generalised lifted spherical convolution}: a spin-$s$ field on the sphere is lifted to the rotation group through the irreducible representation $\rho^{(s)}(\gamma) = e^{\ii s\gamma}$ of $\SOtwo$, convolved on $\SOthree$ with an arbitrary kernel $\Psi$, and projected back to the sphere through the adjoint of the lift for a possibly different spin $s'$.  By decomposing the lifted field and the kernel into their Wigner expansions and noting the orthogonality of the Wigner $D$-functions, the construction is found to act in spin spherical harmonic space by multiplication by the conjugate of a single Wigner coefficient $\Psi^{\ell}_{-s',-s}$ of the kernel, up to a sign.  Each existing equivariant convolution is then recovered by constraining the support of the Wigner coefficients of $\Psi$, by the choice of spins $s$ and $s'$, and by whether the lift or the projection is omitted, as we elaborate in Sec.~\ref{sec:reductions}.  The constants in the literature arise from the different ways in which a kernel on the sphere is identified with a kernel on the rotation group.  The adjoint is a bounded equivariant projection on all group functions and selects the prescribed output spin for arbitrary kernels.  Consequently, the axisymmetric, directional, group and spin convolutions and others, and the analysis and synthesis operators of wavelet transforms on the sphere, become instances of one operator whose properties, equivariance in particular, are established once.

The construction also makes explicit what equivariance permits.  We collect the four cases of bounded equivariant maps between spin fields and group functions: scalar multipliers between fixed spins, vectors for spin-to-group and group-to-spin maps, and matrices acting on group multiplicities, with exact boundedness conditions.  Their kernels are generally distributions, with square integrable representatives precisely for Hilbert--Schmidt maps.  Azimuthally band-limited group functions form tuples of spin fields, so a directional convolution with a steerable kernel is a finite tuple of spin-changing convolutions.  A key result of the article is that, for maps between fields of fixed input and output spin, directional kernels, output on the sphere and equivariance cannot be achieved simultaneously.

The geometric description yields two further key results.  We make explicit that the spin-weighted functions of physics \cite{newman1966note,goldberg1967spin}, defined by the phase that they acquire under a rotation of the local frame, are the sections of the line bundles associated to the quotient $\SOthree/\SOtwo$; making this identification explicit links the physics and the bundle literatures.  We also express the phase that accompanies the rotation of a spin field through the geometry of the bundle: the residual-frame phase developed by the author and collaborators \cite{mcewen2015directional} is the character $\rho^{(s)}$ evaluated on the inverse twist of the bundle section, so that the rotation of a spin field is the representation of the rotation group induced from that character.  A forthcoming article \cite{mulder2026spin} applies the construction to spin spherical Fourier neural operators.

The remainder of this article is structured as follows.  Harmonic analysis on the sphere and rotation group is reviewed concisely in Sec.~\ref{sec:prelim}, to present the mathematical preliminaries and conventions required subsequently.  In Sec.~\ref{sec:review} we review the spherical convolutions of the literature, restating each as a numbered definition in a single convention together with its harmonic form.  Fields on homogeneous spaces, the lift and its adjoint, and induced representations are developed in Sec.~\ref{sec:lifting}.  The generalised lifted spherical convolution is defined in Sec.~\ref{sec:convolution}, where its harmonic form, equivariance, kernel constraint and composition rule are derived and the classification of equivariant maps is proved.  In Sec.~\ref{sec:reductions} each rotation-equivariant construction treated here is recovered as a special case, and in Sec.~\ref{sec:beyond} the constructions that fall outside the framework are analysed.  Harmonic-space algorithms, their computational cost, the \texttt{s2conv} code that implements them and the numerical tests of the reductions and of equivariance are presented in Sec.~\ref{sec:computation}, and the operators are illustrated on measured fields of the Earth and of the sky in Sec.~\ref{sec:illustrations}.  Concluding remarks are made in Sec.~\ref{sec:discussion}.  The identities for Wigner functions and spin harmonics used in the proofs are collected in Appendix~\ref{app:conventions}, and the proofs in Appendix~\ref{app:proofs}.

\section{Harmonic analysis on the sphere and rotation group}
\label{sec:prelim}

We first describe the sphere as a quotient of the rotation group, before reviewing scalar and spin spherical harmonics, the Wigner $D$-functions and the relation between them, following the conventions of \cite{mcewen2011novel,mcewen2015directional,price2024differentiable}.  Spin fields and their rotation are then discussed.

\subsection{The sphere and the rotation group}
\label{sec:prelim:sphere}

We consider the space of square integrable functions $L^2(\sphere)$ on the two-sphere $\sphere = \{\omega \in \mathbb{R}^3 : \lvert \omega \rvert = 1\}$, parameterised by spherical coordinates $\omega = (\theta,\varphi)$ with colatitude $\theta \in [0,\pi]$ and longitude $\varphi \in [0,2\pi)$, with the inner product of $f, g \in L^2(\sphere)$ defined by $\langle f, g \rangle = \int_{\sphere} f(\omega)\, g^{*}(\omega) \dmu(\omega)$, where $\dmu(\omega) = \sin\theta \dd\theta \dd\varphi$ is the usual invariant measure on the sphere and $g^{*}$ denotes complex conjugation.  We write $\eta = (0,0)$ for the North pole.

Rotations $\rho \in \SOthree$ are parameterised by $zyz$ Euler angles $\rho = (\alpha,\beta,\gamma)$, with $\alpha \in [0,2\pi)$, $\beta \in [0,\pi]$ and $\gamma \in [0,2\pi)$, corresponding to the rotation matrix
\begin{equation}
  R(\alpha,\beta,\gamma) = R_z(\alpha)\, R_y(\beta)\, R_z(\gamma) .
  \label{eq:euler}
\end{equation}
We adopt the $zyz$ Euler convention corresponding to the rotation of a physical body in a \emph{fixed} coordinate system about the $z$, $y$ and $z$ axes by $\gamma$, $\beta$ and $\alpha$, respectively.  The inner product of $F, G \in \Lrot$ is defined by $\langle F, G\rangle = \int_{\SOthree} F(\rho)\,G^{*}(\rho) \dmu(\rho)$, where $\dmu(\rho) = \sin\beta \dd\alpha \dd\beta \dd\gamma$ is the usual invariant measure on the rotation group, of total mass $8\pi^2$, which we do not normalise.  Note that $\langle \cdot , \cdot \rangle$ is used to denote inner products over both the sphere and the rotation group (the case adopted can be inferred from the context).  The rotation of a function on the sphere is defined by $(\Rot_\rho f)(\omega) = f(R_\rho^{-1}\omega)$, where $R_\rho$ is the rotation matrix corresponding to $\rho$, and for a point $\omega = (\theta,\varphi)$ we adopt the shorthand notation $\Rot_\omega = \Rot_{(\varphi,\theta,0)}$.

The map $p : \SOthree \to \sphere$, $p(\rho) = R_\rho\, \eta$, sends the rotation $R(\alpha,\beta,\gamma)$ to the point with colatitude $\beta$ and longitude $\alpha$, and is invariant under right multiplication by the subgroup $\SOtwo = \{R_z(\gamma)\}$ of rotations about the $z$ axis, which is the stabiliser of the North pole.  Hence $\sphere \cong \SOthree/\SOtwo$, the coset $\rho\,\SOtwo$ corresponding to the point $\rho\eta$.  Throughout, a point of the sphere is written $(\theta,\varphi)$ when it is considered in its own right and $(\beta,\alpha)$ when it arises as the image of the North pole under a rotation with Euler angles $(\alpha,\beta,\gamma)$; the two descriptions denote the same point.

\subsection{Spherical harmonics, spin spherical harmonics and Wigner functions}
\label{sec:prelim:harmonics}

The spherical harmonics $Y_{\ell m} : \sphere \to \mathbb{C}$ form the canonical orthonormal basis of $L^2(\sphere)$ and are defined by
\begin{equation}
  Y_{\ell m}(\theta,\varphi) = \sqrt{\frac{2\ell+1}{4\pi}\frac{(\ell-m)!}{(\ell+m)!}}\; P^{m}_{\ell}(\cos\theta)\, e^{\ii m \varphi},
  \qquad \ell \in \mathbb{N},\ \lvert m \rvert \le \ell ,
  \label{eq:Ylm}
\end{equation}
where $P^m_\ell$ are the associated Legendre functions, for natural degree $\ell$ and integer order $m$.  We adopt the Condon--Shortley phase convention, with the $(-1)^m$ phase factor included in the definition of the associated Legendre functions, to ensure that the conjugate symmetry relation $Y^{*}_{\ell m} = (-1)^m Y_{\ell,-m}$ holds.  Due to the orthogonality and completeness of the spherical harmonics, any $f \in L^2(\sphere)$ may be represented by its spherical harmonic expansion $f = \sum_{\ell m} f_{\ell m} Y_{\ell m}$, where the spherical harmonic coefficients are given by the usual projection onto each basis function, $f_{\ell m} = \langle f, Y_{\ell m}\rangle$.

Spin spherical harmonics $\sw{s}{Y}_{\ell m}$, for integer spin $s$ with $\lvert s\rvert \le \ell$, are defined through the action of the spin raising and lowering operators $\eth$ and $\bar\eth$ on the scalar spherical harmonics \cite{newman1966note,goldberg1967spin},
\begin{equation}
  \eth = -\sin^{s}\theta \Big( \partial_\theta + \frac{\ii\, \partial_\varphi}{\sin\theta} \Big) \sin^{-s}\theta ,
  \qquad
  \bar\eth = -\sin^{-s}\theta \Big( \partial_\theta - \frac{\ii\, \partial_\varphi}{\sin\theta} \Big) \sin^{s}\theta ,
  \label{eq:eth-def}
\end{equation}
which act on a spin-$s$ quantity, by
\begin{equation}
  \sw{s}{Y}_{\ell m} =
  \begin{cases}
    \sqrt{\dfrac{(\ell-s)!}{(\ell+s)!}}\; \eth^{s}\, Y_{\ell m} , & 0 \le s \le \ell , \\[2ex]
    (-1)^{s}\sqrt{\dfrac{(\ell+s)!}{(\ell-s)!}}\; \bar\eth^{-s}\, Y_{\ell m} , & -\ell \le s \le 0 ,
  \end{cases}
  \label{eq:sY-def}
\end{equation}
so that
\begin{equation}
  \eth\, \sw{s}{Y}_{\ell m} = \sqrt{(\ell-s)(\ell+s+1)}\; \sw{s+1}{Y}_{\ell m},
  \qquad
  \bar\eth\, \sw{s}{Y}_{\ell m} = -\sqrt{(\ell+s)(\ell-s+1)}\; \sw{s-1}{Y}_{\ell m} .
  \label{eq:eth}
\end{equation}
For each fixed $s$ the spin spherical harmonics form an orthonormal basis of the space $\Lsph{s}$ of square integrable spin-$s$ fields on the sphere (Sec.~\ref{sec:prelim:spin}), reducing to the scalar spherical harmonics for $s = 0$.  Any $\sw{s}{f} \in \Lsph{s}$ may thus be represented by its spin spherical harmonic expansion $\sw{s}{f} = \sum_{\ell m} \sw{s}{f}_{\ell m}\, \sw{s}{Y}_{\ell m}$, with coefficients $\sw{s}{f}_{\ell m} = \langle \sw{s}{f}, \sw{s}{Y}_{\ell m}\rangle$.  Note that $\sw{s}{Y}_{\ell m}$ exists only for $\ell \ge \lvert s \rvert$, so that a spin-$s$ field has no harmonic content below degree $\lvert s\rvert$.

The Wigner $D$-functions $D^{\ell}_{mn} : \SOthree \to \mathbb{C}$, for $\lvert m\rvert,\lvert n\rvert \le \ell$, are the matrix elements of the irreducible unitary representations of the rotation group \cite{varshalovich1988quantum}, and may be decomposed as
\begin{equation}
  D^{\ell}_{mn}(\alpha,\beta,\gamma) = e^{-\ii m\alpha}\, d^{\ell}_{mn}(\beta)\, e^{-\ii n\gamma},
  \label{eq:wignerD}
\end{equation}
where the Wigner $d$-functions $d^{\ell}_{mn}$ are real.  Since the $D^\ell$ are representation matrices they satisfy $D^{\ell}(\rho\rho') = D^{\ell}(\rho)D^{\ell}(\rho')$ and $D^{\ell}(\rho^{-1}) = D^{\ell}(\rho)^{\dagger}$, and the orthogonality relation of their matrix elements reads $\langle D^{\ell}_{mn}, D^{\ell'}_{m'n'}\rangle = \frac{8\pi^2}{2\ell+1}\delta_{\ell\ell'}\delta_{mm'}\delta_{nn'}$, where $\delta_{ij}$ is the Kronecker delta symbol.  Together with completeness from the Peter--Weyl theorem \cite{folland2016course}, this makes the conjugate $D$-functions an orthogonal basis of $\Lrot$, and any $F \in \Lrot$ admits the Wigner expansion
\begin{equation}
  F(\rho) = \sum_{\ell=0}^{\infty} \frac{2\ell+1}{8\pi^2} \sum_{m,n=-\ell}^{\ell} F^{\ell}_{mn}\, D^{\ell *}_{mn}(\rho),
  \qquad
  F^{\ell}_{mn} = \langle F, D^{\ell *}_{mn} \rangle = \int_{\SOthree} F(\rho)\, D^{\ell}_{mn}(\rho) \dmu(\rho),
  \label{eq:wigner-transform}
\end{equation}
where we adopt the conjugate $D$-functions as basis elements since this convention simplifies connections to convolutions on the sphere.  Parseval's identity reads $\lVert F \rVert^2 = \sum_{\ell} \frac{2\ell+1}{8\pi^2}\sum_{mn} \lvert F^{\ell}_{mn}\rvert^2$.  The Wigner $D$-functions are related to the spin spherical harmonics by \cite{goldberg1967spin,mcewen2011novel}
\begin{equation}
  D^{\ell}_{mn}(\alpha,\beta,\gamma) = (-1)^n \sqrt{\frac{4\pi}{2\ell+1}}\; \sw{-n}{Y}^{*}_{\ell m}(\beta,\alpha)\; e^{-\ii n\gamma},
  \label{eq:D-sY}
\end{equation}
or equivalently
\begin{equation}
  e^{-\ii s\gamma}\, \sw{s}{Y}_{\ell m}(\beta,\alpha) = (-1)^s \sqrt{\frac{2\ell+1}{4\pi}}\; D^{\ell *}_{m,-s}(\alpha,\beta,\gamma) .
  \label{eq:sY-D}
\end{equation}
Consequently, the spin spherical harmonics of spin $s$ are, up to normalisation and conjugation, the column $n = -s$ of the Wigner $D$-functions restricted to $\gamma = 0$.  This relation underlies the lift developed in Sec.~\ref{sec:lifting}: a spin-$s$ field lifted to the rotation group occupies a single column of Wigner space.

Under rotation the spin spherical harmonic coefficients of a field transform with the same Wigner matrices for every spin \cite{mcewen2015directional},
\begin{equation}
  (\Rot_\rho\, \sw{s}{f})_{\ell m} = \sum_{n=-\ell}^{\ell} D^{\ell}_{mn}(\rho)\; \sw{s}{f}_{\ell n},
  \label{eq:rotation-coeffs}
\end{equation}
where the rotation $\Rot_\rho$ of a spin field is defined in \eqref{eq:spin-rotation} below.  Similarly, the Wigner coefficients of a left-translated function $F(\rho^{-1}\,\cdot)$ read $\sum_{m'} D^{\ell}_{mm'}(\rho) F^{\ell}_{m'n}$, \ie\ left translation acts on the first index of the Wigner coefficients and leaves the second untouched.

Throughout, we say that a field on the sphere is band-limited at $L$ if its harmonic coefficients vanish for all $\ell \ge L$, and that a function on $\SOthree$ is in addition azimuthally band-limited at $N$ if $F^{\ell}_{mn} = 0$ for all $\lvert n \rvert \ge N$.  The harmonic band-limit plays no role in the theory developed in Secs.~\ref{sec:lifting}--\ref{sec:beyond}, which is formulated in $L^2$; it enters only in Sec.~\ref{sec:computation}, where it permits all transforms to be computed exactly and efficiently by appealing to sampling theorems on the sphere \cite{mcewen2011novel} and rotation group \cite{mcewen2015novel}, as implemented in \cite{price2024differentiable}.  The azimuthal band-limit, in contrast, is structural: it is what makes a function on the rotation group a finite tuple of spin fields (Corollary~\ref{cor:fibre-decomposition} and Proposition~\ref{prop:band-limited}).

\subsection{Spin fields}
\label{sec:prelim:spin}

A spin-$s$ field $\sw{s}{f}$ on the sphere is a complex function whose value at $\omega$ is defined relative to a choice of orthonormal frame $(e_\theta, e_\varphi)$ in the tangent plane and transforms as
\begin{equation}
  \sw{s}{f}'(\omega) = e^{-\ii s\chi}\, \sw{s}{f}(\omega)
  \label{eq:spin-phase}
\end{equation}
under a right-handed rotation of the frame by $\chi$ \cite{newman1966note,goldberg1967spin,zaldarriaga1997all}.\footnote{The sign convention adopted for the argument of the complex exponential differs from the original definition \cite{newman1966note} but is identical to the convention used in the context of the polarisation of the cosmic microwave background \cite{zaldarriaga1997all,kamionkowski1997statistics}.}  It is important to note that the rotation considered here is \emph{not} a global rotation on the sphere but rather a rotation by $\chi$ in the tangent plane centred on $\omega$.  Scalar functions such as temperature are spin $0$; the complex representation $v_\theta \pm \ii v_\varphi$ of a tangent vector field such as wind has spin $\pm 1$; the Stokes combination $Q \pm \ii U$ of linear polarisation, and likewise cosmic shear, has spin $\pm 2$ \cite{zaldarriaga1997all}; and the second flexion of weak lensing, the spin-raised gradient $\eth\gamma$ of the shear, which bends images into three-fold arcs, has spin $\pm 3$ \cite{bacon2006weak}.  In Fig.~\ref{fig:spin-fields} we plot a scalar field together with the spin-$1$, spin-$2$ and spin-$3$ fields obtained from it by repeated application of the spin raising operator $\eth$.

Figure~\ref{fig:fields-overview} shows seven such fields measured on the Earth and on the sky, which span spins $0$, $1$ and $2$ and to which we return in Sec.~\ref{sec:illustrations}.  The relief of the Earth is a scalar field.  The deflection of the vertical, the angle by which a plumb line departs from the normal to the reference ellipsoid, the horizontal part of the geomagnetic field, which aligns a compass needle, the wind near the surface and the currents at the surface of the ocean are tangent vector fields, of spin $\pm1$, drawn as arrows.  The polarisation of the cosmic microwave background and the shear of weak gravitational lensing, the stretching of the images of distant galaxies, are fields of spin $\pm2$, drawn as headless sticks, since a polarisation angle or a direction of stretching is unchanged by a half turn.  In each case the frame law \eqref{eq:spin-phase} is a physical statement: a bearing, a polarisation angle or a direction of stretching is a number only once a reference direction has been chosen at each point, and it changes by a known phase when that choice changes.  The constructions of this article act on such fields without a global choice of frame, which no choice of coordinates on the sphere can provide.

\begin{figure}[tp]
  \centering
  \begin{subfigure}[t]{0.32\textwidth}
    \centering
    \includegraphics[width=\linewidth]{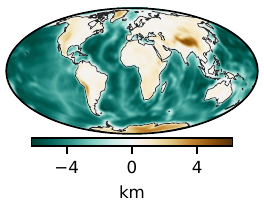}
    \caption{Topography, in km.}
  \end{subfigure}\hfill
  \begin{subfigure}[t]{0.32\textwidth}
    \centering
    \includegraphics[width=\linewidth]{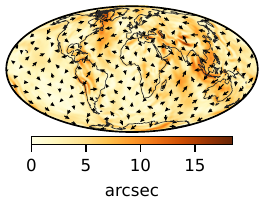}
    \caption{Deflection of the vertical, in arcseconds.}
  \end{subfigure}\hfill
  \begin{subfigure}[t]{0.32\textwidth}
    \centering
    \includegraphics[width=\linewidth]{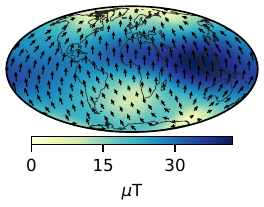}
    \caption{Horizontal geomagnetic field, in $\mu$T.}
  \end{subfigure}\\[1ex]
  \begin{subfigure}[t]{0.32\textwidth}
    \centering
    \includegraphics[width=\linewidth]{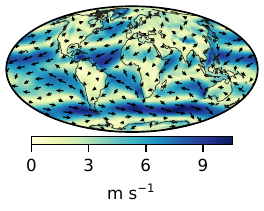}
    \caption{Near-surface wind, in m\,s$^{-1}$.}
  \end{subfigure}\hspace{0.04\textwidth}
  \begin{subfigure}[t]{0.32\textwidth}
    \centering
    \includegraphics[width=\linewidth]{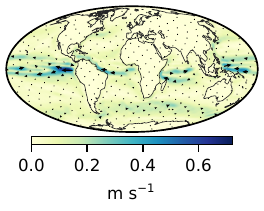}
    \caption{Surface currents of the ocean, in m\,s$^{-1}$.}
  \end{subfigure}\\[1ex]
  \begin{subfigure}[t]{0.32\textwidth}
    \centering
    \includegraphics[width=\linewidth]{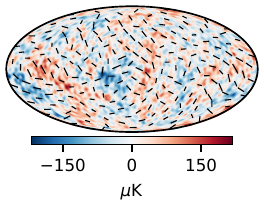}
    \caption{Temperature of the cosmic microwave background, in $\mu$K, with its polarisation.}
  \end{subfigure}\hspace{0.04\textwidth}
  \begin{subfigure}[t]{0.32\textwidth}
    \centering
    \includegraphics[width=\linewidth]{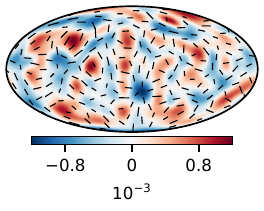}
    \caption{Simulated lensing convergence, with the shear.}
  \end{subfigure}
  \caption{Fields of spin $0$, $1$ and $2$ on the Earth and on the celestial sky (the examples of Sec.~\ref{sec:illustrations}) in the Mollweide projection.  The relief (a) is a scalar field.  The deflection of the vertical (b), the small angle by which the direction of gravity differs from the normal to the reference ellipsoid because of the uneven distribution of mass within the Earth, the horizontal geomagnetic field (c), which aligns a compass needle, the wind (d) and the currents (e) are spin-$1$ fields, with the modulus shown as colour and the direction drawn as arrows.  The polarisation of the cosmic microwave background (f) and the shear of weak gravitational lensing (g) are spin-$2$ fields, drawn as headless sticks over the temperature anisotropy and the convergence, and are smoothed so that the sticks follow their large-scale patterns.}
  \label{fig:fields-overview}
\end{figure}

Because the frame is fixed by the coordinate lines, a global rotation of a spin field must be accompanied by a frame correction.  The rotation of a spin field is defined by \cite{mcewen2015directional}
\begin{equation}
  (\Rot_\rho\, \sw{s}{f})(\omega) = e^{-\ii s \vartheta}\; \sw{s}{f}(R_\rho^{-1}\omega),
  \label{eq:spin-rotation}
\end{equation}
where the angle $\vartheta \in [0,2\pi)$ is the third Euler angle of the rotation $R_\rho^{-1} R_\omega$, with $R_\omega = R(\varphi,\theta,0)$ the rotation matrix corresponding to the Euler angles $(\varphi,\theta,0)$.  For the scalar setting where $s = 0$, \eqref{eq:spin-rotation} reduces to the usual rotation operator defined solely through a rotation of the coordinate system.  The exponential factor compensates the residual rotation of the local frame so that the field retains its spin, as described in \cite{mcewen2015directional,boyle2016how}, and yields the harmonic law \eqref{eq:rotation-coeffs}.  We write $\Rot_\rho$ for the rotation of scalar and spin fields alike, the spin being that of the field on which it acts; in Sec.~\ref{sec:lifting:induced} we show that $\Rot_\rho$ is the representation $\pis{s}(R_\rho)$ of the rotation group induced from the character $\rho^{(s)}$, with the phase evaluated on the inverse twist.  The Euler-angle description holds away from the coordinate singularities, hence almost everywhere.

\section{Spherical convolutions in the literature}
\label{sec:review}

In this section we review the definitions of convolution on the sphere proposed in the literature, restating each as a numbered definition in the single convention of Sec.~\ref{sec:prelim}, together with its harmonic form.  We first recall the Euclidean convolution and the difficulty of translating it to the sphere, before reviewing axisymmetric convolutions, directional convolutions, convolution on the rotation group, spin convolutions, anisotropic convolutions whose output remains on the sphere, and finally the projections from the rotation group back to the sphere that arise in wavelet synthesis.  The definitions are summarised in Table~\ref{tab:review}.

\subsection{Convolution on the plane and translation on the sphere}
\label{sec:review:plane}

On the plane $\mathbb{R}^2$ the convolution $(f \star g)(x) = \int_{\mathbb{R}^2} f(y)\, g(x-y) \dd y$ is a bilinear pairing of $f$ with a translated and reflected $g$, or an inner product when $g$ is real.  It commutes with translations, it is commutative, and every translation-equivariant bounded linear operator is of this form, with a distributional kernel in general.  On the sphere translations are replaced by rotations.  However, as Kennedy \etal\ \cite{kennedy2011azimuthally} emphasise, the rotation group is not the analogue of the plane: a rotation carries three parameters while a point on the sphere carries two, and a kernel translated to a point by a rotation acquires an orientation that has no counterpart in the Euclidean setting.  Every definition of spherical convolution must therefore decide how to treat this third degree of freedom, namely whether to integrate it out, to retain it in the output, or to fix it by a rule.  The definitions below treat this orientation in different ways.  Roddy \& McEwen \cite{roddy2021sifting} compare them by three criteria, namely support for directional kernels, an output that remains on the sphere, and computational cost no greater than that of a fast spherical harmonic transform, to which we add equivariance under $\SOthree$, \ie\ that convolution commute with the rotation of the signal.

\subsection{Axisymmetric convolutions with output on the sphere}
\label{sec:review:axisymmetric}

The simplest approach is to restrict the kernel to be axisymmetric, so that its orientation is immaterial.

\begin{definition}[Axisymmetric convolution]
\label{def:axisymmetric}
For $f,\psi\in L^2(\sphere)$ with axisymmetric kernel $\psi$, \ie\ $\psi_{\ell m}=\psi_{\ell0}\delta_{m0}$, the axisymmetric convolution is
\begin{equation}
  (f \odot \psi)(\omega) = \langle f, \Rot_\omega \psi \rangle = \int_{\sphere} f(\omega')\, (\Rot_\omega \psi)^{*}(\omega') \dmu(\omega'),
  \qquad \Rot_\omega = \Rot_{(\varphi,\theta,0)} .
  \label{eq:axisymmetric-conv}
\end{equation}
\end{definition}
For an axisymmetric kernel, \ie\ a kernel invariant under azimuthal rotation when centred on the North pole so that $\psi_{\ell m} = \psi_{\ell 0}\delta_{m0}$, the harmonic form of \eqref{eq:axisymmetric-conv} reads \cite{mcewen2007fast,roddy2021sifting}
\begin{equation}
  (f \odot \psi)_{\ell m} = \sqrt{\frac{4\pi}{2\ell+1}}\; f_{\ell m}\, \psi^{*}_{\ell 0} ,
  \label{eq:axisymmetric-harmonic}
\end{equation}
which follows directly from \eqref{eq:rotation-coeffs} and \eqref{eq:D-sY}.  The output remains on the sphere and the convolution is simply a product in harmonic space, hence it may be computed exactly and efficiently at the cost of two spherical harmonic transforms.  This operator underlies both the axisymmetric continuous wavelets of \cite{mcewen2007fast}, constructed by stereographic dilation, and the needlet and scale-discretised constructions of \cite{narcowich2006localized,leistedt2013s2let}.
Driscoll \& Healy \cite{driscoll1994computing} instead integrate over the full rotation group.

\begin{definition}[Left convolution of Driscoll \& Healy]
\label{def:driscoll-healy}
For $f, h \in L^2(\sphere)$, the left convolution is
\begin{equation}
  (f \star_{\mathrm{L}} h)(\omega) = \int_{\SOthree} f(\rho\eta)\, h(R_\rho^{-1}\omega) \dmu(\rho) .
  \label{eq:driscoll-healy}
\end{equation}
\end{definition}
Its harmonic form reads \cite{driscoll1994computing,kennedy2011azimuthally,esteves2018learning}
\begin{equation}
  (f \star_{\mathrm{L}} h)_{\ell m} = 2\pi\sqrt{\frac{4\pi}{2\ell+1}}\; f_{\ell m}\, h_{\ell 0} .
  \label{eq:driscoll-healy-harmonic}
\end{equation}
Note that the kernel $h$ may be arbitrary but only its zonal part $h_{\ell 0}$ enters, since the integral over the third Euler angle averages $h$ about its own pole (Theorem~3 of \cite{kennedy2011azimuthally}); note also that the pairing is bilinear, with no complex conjugate, and that Kennedy \etal\ \cite{kennedy2011azimuthally} divide \eqref{eq:driscoll-healy} by $2\pi$.  The left convolution is the convolution of the spherical convolutional networks of Esteves \etal\ \cite{esteves2018learning}, who observe that it ``only makes sense to learn zonal filters'', and of the spherical Fourier neural operator of Bonev \etal\ \cite{bonev2023spherical}.

\subsection{Directional convolutions with output on the rotation group}
\label{sec:review:directional}

The second approach retains the orientation of the kernel in the output.

\begin{definition}[Directional convolution]
\label{def:directional}
For $f, \psi \in L^2(\sphere)$, the directional convolution is
\begin{equation}
  (f \circledast \psi)(\rho) = \langle f, \Rot_\rho \psi \rangle = \int_{\sphere} f(\omega)\, (\Rot_\rho \psi)^{*}(\omega) \dmu(\omega), \qquad \rho \in \SOthree .
  \label{eq:directional-conv}
\end{equation}
\end{definition}
By decomposing $f$ and $\psi$ into their spherical harmonic expansions and noting the orthogonality of the spherical harmonics and their rotation by \eqref{eq:rotation-coeffs}, the directional convolution may be written
\begin{equation}
  (f \circledast \psi)(\rho) = \sum_{\ell m n} f_{\ell m}\, \psi^{*}_{\ell n}\, D^{\ell *}_{mn}(\rho),
  \qquad
  (f \circledast \psi)^{\ell}_{mn} = \frac{8\pi^2}{2\ell+1}\; f_{\ell m}\, \psi^{*}_{\ell n} ,
  \label{eq:directional-harmonic}
\end{equation}
where the second expression follows by comparison with \eqref{eq:wigner-transform}.  Consequently, the output is a function on $\SOthree$ whose Wigner coefficients are the outer product of the harmonic coefficients of signal and kernel \cite{wandelt2001fast,mcewen2007fast,wiaux2007complex,mcewen2013directional,cohen2018spherical}.  Wandelt \& G\'orski \cite{wandelt2001fast} parameterise the rotation by the scan geometry of a cosmic microwave background experiment and evaluate \eqref{eq:directional-harmonic} on a torus, and Challinor \etal\ \cite{challinor2000all} give the polarised, \ie\ spin-2, version.  Fast algorithms by factorisation of rotations \cite{wandelt2001fast,mcewen2007fast,mcewen2013directional} cost $\mathcal{O}(L^4)$ in general, or $\mathcal{O}(NL^3)$ for kernels azimuthally band-limited at $N$, which is the case of steerable kernels \cite{wiaux2006fast}.  Cohen \etal\ \cite{cohen2018spherical} write the same operator for real signals with the Haar measure normalised to unity, and the discrete--continuous convolution of Ocampo \etal\ \cite{ocampo2023scalable} is a quadrature of \eqref{eq:directional-conv} on a sampling of the sphere.  Directional convolution supports directional kernels but yields an output that does not remain on the sphere, at greater computational cost than a spherical harmonic transform.

\subsection{Convolution on the rotation group}
\label{sec:review:group}

Once the output is a function on $\SOthree$, subsequent convolutions take place on the rotation group itself.

\begin{definition}[Group convolution on $\SOthree$]
\label{def:group-conv}
For $\Psi, F \in \Lrot$, the group convolution is
\begin{equation}
  (\Psi \conv F)(\rho) = \langle F, \piG(\rho)\Psi \rangle = \int_{\SOthree} \Psi^{*}(\rho^{-1}\rho')\, F(\rho') \dmu(\rho'),
  \qquad (\piG(\rho)\Psi)(\rho') = \Psi(\rho^{-1}\rho') ,
  \label{eq:group-conv}
\end{equation}
where $\piG$ is the left regular representation of $\SOthree$.
\end{definition}
In Wigner space the group convolution is a matrix product for each degree \cite{kostelec2008ffts,cohen2018spherical,cobb2021efficient},
\begin{equation}
  (\Psi \conv F)^{\ell}_{mn} = \sum_{m'=-\ell}^{\ell} F^{\ell}_{mm'}\, \Psi^{\ell *}_{nm'} ,
  \label{eq:group-conv-harmonic}
\end{equation}
as we show in Appendix~\ref{app:proofs}.  Definition~\ref{def:group-conv} is written in correlation form, with the kernel conjugated and translated by the left regular representation, so that \eqref{eq:group-conv} is manifestly the inner product of $F$ with a rotated copy of $\Psi$; this is the $\SOthree$ layer of Cohen \etal\ \cite{cohen2018spherical}, the correlation of Kostelec \& Rockmore \cite{kostelec2008ffts}, the linear part of the Clebsch--Gordan networks of Kondor \etal\ \cite{kondor2018clebsch} and, with channel mixing, the generalised convolution of Cobb \etal\ \cite{cobb2021efficient}.  Our convention differs from that of Xu \etal\ \cite{xu2022unified} and Kondor \& Trivedi \cite{kondor2018generalization}, who write $\int \kappa(\rho'^{-1}\rho)F(\rho')\dmu(\rho')$, by $\kappa(\rho) = \Psi^{*}(\rho^{-1})$, and from that of Cohen, Geiger \& Weiler \cite{cohen2019general}, who write $\int\kappa(\rho^{-1}\rho')F(\rho')\dmu(\rho')$, by conjugation alone.

\subsection{Spin convolutions}
\label{sec:review:spin}

We follow the directional spin wavelet transform of McEwen \etal\ \cite{mcewen2015directional}, which defines convolutions of spin fields by the same three constructions as in the scalar case: a directional convolution with output on the rotation group, an axisymmetric convolution with output on the sphere, and a synthesis integral back to the sphere.  All three rest on the rotation \eqref{eq:spin-rotation} of a spin field.

\paragraph{Rotation of spin fields.}
A coordinate rotation alone, $\sw{s}{f} \mapsto \sw{s}{f}(R_\rho^{-1}\,\cdot)$, does not map a spin-$s$ field to a spin-$s$ field, because the coordinate frame at $R_\rho^{-1}\omega$ is not the image under $R_\rho^{-1}$ of the coordinate frame at $\omega$.  The operator \eqref{eq:spin-rotation} corrects this with the phase $e^{-\ii s\vartheta}$, where $\vartheta$ is the third Euler angle of $R_\rho^{-1}R_\omega$.  In \cite{mcewen2015directional} the phase is justified in harmonic space: from the composition law $D^{\ell}(\rho_1\rho_2) = D^{\ell}(\rho_1)D^{\ell}(\rho_2)$ and \eqref{eq:sY-D} one finds
\begin{equation}
  (\Rot_\rho\, \sw{s}{Y}_{\ell m})(\omega) = \sum_{n=-\ell}^{\ell} D^{\ell}_{nm}(\rho)\, \sw{s}{Y}_{\ell n}(\omega),
  \label{eq:spin-harmonic-rotation}
\end{equation}
equivalently \eqref{eq:rotation-coeffs}, so that $\Rot_\rho$ maps spin-$s$ harmonics to spin-$s$ harmonics with the same Wigner matrices as in the scalar case, which is what makes the constructions below consistent.  The residual frame rotation is described in \cite{mcewen2015directional,boyle2016how}.  Here we express it through the section twist, connecting the adopted spin convention to the induced representation of Cohen, Geiger \& Weiler \cite{cohen2019general} (Proposition~\ref{prop:intertwining} and Remark~\ref{rem:phase-interpretation}).  The lift then intertwines $\Rot_\rho$ with left translation on the rotation group, giving \eqref{eq:spin-harmonic-rotation}.

\begin{definition}[Directional spin convolution]
\label{def:spin-directional}
For $\sw{s}{f}, \sw{s}{\psi} \in \Lsph{s}$, the directional spin convolution is
\begin{equation}
  (\sw{s}{f} \circledast \sw{s}{\psi})(\rho) = \langle \sw{s}{f}, \Rot_\rho\, \sw{s}{\psi} \rangle
  = \int_{\sphere} \sw{s}{f}(\omega)\, (\Rot_\rho\, \sw{s}{\psi})^{*}(\omega) \dmu(\omega),
  \qquad \rho \in \SOthree .
  \label{eq:spin-directional}
\end{equation}
\end{definition}
These are the wavelet coefficients of \cite{mcewen2015directional} for a spin wavelet $\sw{s}{\psi}$.  By decomposing both fields into their spin spherical harmonic expansions and noting \eqref{eq:spin-harmonic-rotation}, the directional spin convolution may be written
\begin{equation}
  (\sw{s}{f} \circledast \sw{s}{\psi})(\rho) = \sum_{\ell m n} \sw{s}{f}_{\ell m}\, \sw{s}{\psi}^{*}_{\ell n}\, D^{\ell *}_{mn}(\rho),
  \qquad
  (\sw{s}{f} \circledast \sw{s}{\psi})^{\ell}_{mn} = \frac{8\pi^2}{2\ell+1}\; \sw{s}{f}_{\ell m}\, \sw{s}{\psi}^{*}_{\ell n} ,
  \label{eq:spin-directional-harmonic}
\end{equation}
the spin generalisation of \eqref{eq:directional-harmonic}.  From \eqref{eq:spin-directional-harmonic} it is apparent that the output is a scalar function on $\SOthree$ whatever the spin of the input: the third Euler angle records the orientation of the kernel about each point, and the spin phases of field and kernel cancel in the inner product.  The transform may thus be computed via an inverse Wigner transform of the coefficients \eqref{eq:spin-directional-harmonic}, as we elaborate in Sec.~\ref{sec:computation}.

\begin{definition}[Axisymmetric spin convolution]
\label{def:spin-axisymmetric}
For $\sw{s}{f} \in \Lsph{s}$ and an axisymmetric kernel $\sw{s}{\Phi}$, \ie\ $\Rot_{(0,0,\gamma)}\,\sw{s}{\Phi} = \sw{s}{\Phi}$ or equivalently $\sw{s}{\Phi}_{\ell m} = \sw{s}{\Phi}_{\ell 0}\delta_{m0}$, the axisymmetric spin convolution is
\begin{equation}
  (\sw{s}{f} \odot \sw{s}{\Phi})(\omega) = \langle \sw{s}{f}, \Rot_\omega\, \sw{s}{\Phi} \rangle, \qquad \Rot_\omega = \Rot_{(\varphi,\theta,0)} .
  \label{eq:spin-axisymmetric}
\end{equation}
\end{definition}
These are the scaling coefficients of \cite{mcewen2015directional}, whose harmonic form reads
\begin{equation}
  (\sw{s}{f} \odot \sw{s}{\Phi})_{\ell m} = \sqrt{\frac{4\pi}{2\ell+1}}\; \sw{s}{f}_{\ell m}\, \sw{s}{\Phi}^{*}_{\ell 0} ,
  \label{eq:spin-axisymmetric-harmonic}
\end{equation}
where the coefficients on the left are those of the scalar spherical harmonics: the output is a scalar field on the sphere although the input has spin $s$.  
The field is recovered from its wavelet and scaling coefficients $W^{\psi^{(j)}} = \sw{s}{f} \circledast \sw{s}{\psi}^{(j)}$ and $W^{\Phi} = \sw{s}{f} \odot \sw{s}{\Phi}$ by the synthesis \cite{mcewen2015directional}
\begin{equation}
  \sw{s}{f}(\omega) = \int_{\sphere} W^{\Phi}(\omega')\, (\Rot_{\omega'}\, \sw{s}{\Phi})(\omega) \dmu(\omega')
  + \sum_{j} \int_{\SOthree} W^{\psi^{(j)}}(\rho)\, (\Rot_\rho\, \sw{s}{\psi}^{(j)})(\omega) \dmu(\rho),
  \label{eq:spin-synthesis}
\end{equation}
provided the admissibility condition
\begin{equation}
  \frac{4\pi}{2\ell+1}\, \bigl\lvert \sw{s}{\Phi}_{\ell 0} \bigr\rvert^2 + \frac{8\pi^2}{2\ell+1} \sum_{j} \sum_{n=-\ell}^{\ell} \bigl\lvert \sw{s}{\psi}^{(j)}_{\ell n} \bigr\rvert^2 = 1
  \qquad \text{for all } \ell\ge\lvert s\rvert
  \label{eq:spin-admissibility}
\end{equation}
holds.  The second integral in \eqref{eq:spin-synthesis} is the spin version of the synthesis operator of Definition~\ref{def:synthesis} below and the first its axisymmetric counterpart; both project back to a spin-$s$ field.  As we show in Sec.~\ref{sec:reductions}, each is the adjoint of the corresponding analysis operator, which is why admissibility takes the form of a resolution of the identity.  The wavelets are steerable, \ie\ azimuthally band-limited, so that their directional coefficients at all orientations are determined by finitely many \cite{mcewen2015directional,wiaux2006fast}; Proposition~\ref{prop:band-limited} identifies such coefficients with a tuple of spin fields.

Other spin constructions in the literature admit the same description.  The spin-weighted spherical convolutional networks of Esteves \etal\ \cite{esteves2020spin} define, for a set of input spins $W_F$ and filter spins $W_K$, the convolution $\sw{s}{(\widehat{F * K})}{}^{\ell}_{m} = \sum_{i \in W_F} \sw{i}{\hat f}{}^{\ell}_{m}\, \sw{s}{\hat k}{}^{\ell}_{i}$ for $s \in W_K$, in their notation, so that the output carries the spins of the filter, and a cross-correlation $\sw{s}{(\widehat{F \star K})}{}^{\ell}_{m} = \sum_{i} \sw{i}{\hat f}{}^{\ell}_{m}\, \sw{i}{\hat k}{}^{\ell *}_{s}$, which is a sum of spin-changing multipliers.  Their explicit $D$ convention and section harmonics agree with ours; the exact kernel conversion, rather than an unchanged column reading of Definition~\ref{def:spin-directional}, is given in Proposition~\ref{prop:spin-change}.  The spin needlets of Geller \& Marinucci \cite{geller2010spin,geller2011mixed} use isotropic two-point kernels $\sum_{\ell m} b\big(\sqrt{e_{\ell s}}/B^{j}\big)\, \sw{s}{Y}_{\ell m}(x)\,\sw{s}{Y}^{*}_{\ell m}(y)$ with $e_{\ell s} = (\ell-\lvert s\rvert)(\ell+\lvert s\rvert+1)$, mapping spin $s$ to spin $s$, and their mixed needlets replace one of the two harmonics by a scalar harmonic, connecting spin $s$ with spin $0$ only.  The spin-weighted spherical convolution introduced by Tessore \etal\ \cite{tessore2025euclid} for making maps of cosmic shear from discrete galaxy catalogues, $F(x) = \int_{\sphere} f(y)\, e^{\ii S\alpha}\, K(\theta)\, e^{-\ii s\alpha'} \dmu(y)$, pairs a radial profile $K$ of the separation $\theta$ with relative-frame phases and maps spin $s$ to a freely chosen spin $S$.  They state the theorem $F_{\ell m}=K_\ell f_{\ell m}$ for $K(\theta)=\sum_\ell(2\ell+1)K_\ell d^\ell_{Ss}(\theta)/(4\pi)$.  Proposition~\ref{prop:spin-change} translates the angles explicitly and distinguishes the literal printed convention for odd output spins.  Their cap profile $K(\theta)=\cos^{2s}(\theta/2)$ for $\theta\le\beta$ has closed-form coefficients for $S=s\ge0$.  Finally, a number of operators of mathematical physics change the spin of a field while commuting with rotations: the spin raising and lowering operators \eqref{eq:eth}; the $E$/$B$ decomposition of a spin-$\pm2$ polarisation field, $a_{E,\ell m} = -(a_{2,\ell m}+a_{-2,\ell m})/2$ and $a_{B,\ell m} = \ii(a_{2,\ell m}-a_{-2,\ell m})/2$ \cite{zaldarriaga1997all,kamionkowski1997statistics}; the spherical Kaiser--Squires relation $\sw{2}{\gamma}_{\ell m} = \mathcal{D}_\ell\, \kappa_{\ell m}$, with $\mathcal{D}_\ell = -\sqrt{(\ell+2)(\ell-1)/(\ell(\ell+1))}$, between the convergence (spin $0$) and shear (spin $2$) of weak gravitational lensing \cite{wallis2022mapping}; and the all-sky convolution of a polarised beam \cite{challinor2000all}.  All are recovered in Sec.~\ref{sec:reductions}.

\subsection{Anisotropic convolutions with output on the sphere}
\label{sec:review:anisotropic}

The third approach to address the orientation degree of freedom is to fix it by a rule, which yields an output on the sphere for a directional kernel.

\begin{definition}[Standard correlation]
\label{def:standard-correlation}
For $f, \psi \in L^2(\sphere)$ with $\psi$ not necessarily axisymmetric, the standard correlation is the directional convolution restricted to the rotations with vanishing third Euler angle,
\begin{equation}
  (f \circledast \psi)(R_\omega) = \langle f, \Rot_\omega \psi \rangle, \qquad R_\omega = R(\varphi,\theta,0) .
  \label{eq:standard-correlation}
\end{equation}
\end{definition}
This is the ``standard correlation'' of Wiaux \etal\ \cite{wiaux2006fast,wiaux2007complex}, and is formula \eqref{eq:axisymmetric-conv} applied to a directional kernel.  Its harmonic coefficients are not a pointwise product \cite{wiaux2007complex}.

\begin{definition}[Commutative anisotropic convolution]
\label{def:commutative-anisotropic}
For $f, h \in L^2(\sphere)$, the commutative anisotropic convolution of Sadeghi \etal\ \cite{sadeghi2012commutative,khalid2012anisotropic} is
\begin{equation}
  (h \odot_{\mathrm{S}} f)(\theta,\varphi) = \int_{\sphere} \big(\Rot_{(\varphi,\theta,\pi-\varphi)} h\big)(\omega')\, f(\omega') \dmu(\omega') .
  \label{eq:commutative-anisotropic}
\end{equation}
\end{definition}
The pairing is bilinear.  Sadeghi \etal\ \cite{sadeghi2012commutative} prove that their bilinear directional correlation evaluated at rotations $(\varphi,\theta,\chi(\varphi))$ is commutative if and only if $\chi = \pi - \varphi$, and give its harmonic form and an $\mathcal{O}(L^3\log L)$ algorithm; they show commutativity, not equivariance.

\begin{definition}[Harmonic multiplication and sifting convolution]
\label{def:sifting}
For $f, g \in L^2(\sphere)$, harmonic multiplication \cite{kennedy2011azimuthally} and the sifting convolution \cite{roddy2021sifting} are defined in harmonic space by
\begin{equation}
  (g \odot_{\mathrm{K}} f)_{\ell m} = g_{\ell m}\, f_{\ell m},
  \qquad
  (f \circledcirc g)_{\ell m} = f_{\ell m}\, g^{*}_{\ell m} ,
  \label{eq:sifting}
\end{equation}
the latter arising in real space as $(f \circledcirc g)(\omega) = \langle \mathcal{T}_\omega f, g\rangle$ for the translation operator $(\mathcal{T}_{\omega'} f)(\omega) = \sum_{\ell m} f_{\ell m}\, Y_{\ell m}(\omega')\, Y_{\ell m}(\omega)$.
\end{definition}
Both operators are diagonal in the spherical harmonic basis.  Kennedy \etal\ \cite{kennedy2011azimuthally} characterise isotropic operators by independence of the multiplier from $m$.  Harmonic multiplication is commutative; sifting is commutative up to coefficientwise conjugation, not pointwise conjugation of the field.  Sifting has the $\mathcal{O}(L^3)$ cost of spherical harmonic transforms and meets the three criteria of \cite{roddy2021sifting}.  Both families are always polar-equivariant but not fully rotation-equivariant in general; additional symmetries and equivariant special cases are discussed in Sec.~\ref{sec:beyond}.

\subsection{Projections from the rotation group to the sphere}
\label{sec:review:projection}

Finally, the adjoint of directional analysis synthesises a signal from group coefficients by an integral over $\SOthree$.  It becomes an inverse only with appropriate admissibility or dual-frame conditions.

\begin{definition}[Synthesis or transpose convolution]
\label{def:synthesis}
For $W \in \Lrot$ and $\psi \in L^2(\sphere)$, the synthesis is
\begin{equation}
  (W \circledast^{\dagger} \psi)(\omega) = \int_{\SOthree} W(\rho)\, (\Rot_\rho \psi)(\omega) \dmu(\rho),
  \qquad
  (W \circledast^{\dagger} \psi)_{\ell m} = \sum_{n=-\ell}^{\ell} W^{\ell}_{mn}\, \psi_{\ell n} ,
  \label{eq:synthesis}
\end{equation}
where the harmonic form follows by expanding $W$ and $\psi$ and noting \eqref{eq:rotation-coeffs}.
\end{definition}
This is Haar-measure synthesis, distinct from the flat-scan ``transpose convolution'' of Wandelt \& G\'orski \cite{wandelt2001fast} as explained in Sec.~\ref{sec:reductions:directional}, and is the synthesis step of directional wavelet transforms on the sphere \cite{mcewen2007fast,wiaux2008exact,mcewen2013directional,mcewen2015directional}, where the kernel is either the wavelet itself, with the constant $8\pi^2/(2\ell+1)$ carried by the admissibility condition as in \eqref{eq:spin-admissibility} \cite{mcewen2015directional}, or the dual wavelet $\psi_{\ell n} \mapsto \frac{2\ell+1}{8\pi^2}\psi_{\ell n}$, with an admissibility condition free of constants \cite{wiaux2008exact,mcewen2013directional}.  The left convolution of Kennedy \etal\ \cite{kennedy2011azimuthally} and Esteves \etal\ \cite{esteves2018learning} is obtained instead by marginalising the bilinear directional correlation over the third Euler angle; with the sesquilinear convention \eqref{eq:directional-conv}, the spherical kernel is conjugated.

\subsection{Summary}
\label{sec:review:summary}

The convolutions reviewed are summarised in Table~\ref{tab:review}.  In the rotation-equivariant harmonic formulas, the signal order $m$ is a spectator index: the kernel acts on the remaining indices.  Dependence of a diagonal multiplier on $m$, or an unconstrained restriction to selected rotations, instead gives families that are not fully equivariant in general.  Their special cases and symmetry conditions are established in Secs.~\ref{sec:convolution}--\ref{sec:beyond}.

\begin{table}[t]
  \caption{Spherical convolutions in the notation of Sec.~\ref{sec:prelim}.  Symmetry entries describe the family: ``yes'' guarantees $\SOthree$ equivariance, ``polar'' guarantees polar rotations without excluding other symmetries of particular kernels, and ``no'' means full equivariance is not guaranteed.  Harmonic multiplication is commutative; $^{\mathrm{a}}$ denotes coefficientwise conjugation for sifting, with longitude reflection in real space.  ``--'' means no numbered definition or, in the commutativity column, not applicable.  Costs are algorithmic upper bounds for band-limit $L$ and azimuthal band-limit $N$, with standard correlation referring to grid values.}
  \label{tab:review}
  \scriptsize\setlength{\tabcolsep}{2.5pt}
  \centering
  \resizebox{\textwidth}{!}{%
  \begin{tabular}{llllll}
    \toprule
    Convolution & Def. & Input $\to$ output & Kernel & Equiv.\,/\,comm. & Cost \\
    \midrule
    Axisymmetric $\odot$ & \ref{def:axisymmetric} & $\sphere \to \sphere$ & axisymmetric & yes\,/\,-- & $L^3$ \\
    Left $\star_{\mathrm{L}}$ & \ref{def:driscoll-healy} & $\sphere \to \sphere$ & any (zonal part) & yes\,/\,no & $L^3$ \\
    Directional $\circledast$ & \ref{def:directional} & $\sphere \to \SOthree$ & any & yes\,/\,no & $L^4$, $NL^3$ \\
    Group $\conv$ & \ref{def:group-conv} & $\SOthree \to \SOthree$ & any & yes\,/\,no & $L^4$ \\
    Directional spin $\circledast$ & \ref{def:spin-directional} & $\Lsph{s} \to \SOthree$ & spin $s$ & yes\,/\,no & $L^4$, $NL^3$ \\
    Axisymmetric spin $\odot$ & \ref{def:spin-axisymmetric} & $\Lsph{s} \to \Lsph{0}$ & axisymmetric spin $s$ & yes\,/\,-- & $L^3$ \\
    Spin-weighted \cite{esteves2020spin} & -- & $\Lsph{s} \to \Lsph{s'}$ & spin $s'$ & yes\,/\,no & $L^3$ \\
    Spin-weighted, isotropic \cite{tessore2025euclid} & -- & $\Lsph{s} \to \Lsph{S}$ & radial, with frame phases & yes\,/\,-- & $L^3$ \\
    Synthesis $\circledast^{\dagger}$ & \ref{def:synthesis} & $\SOthree \to \sphere$ & any & yes\,/\,-- & $L^4$, $NL^3$ \\
    Standard correlation & \ref{def:standard-correlation} & $\sphere \to \sphere$ & directional & polar\,/\,no & $L^3$ \\
    Commutative anisotropic $\odot_{\mathrm{S}}$ & \ref{def:commutative-anisotropic} & $\sphere \to \sphere$ & directional & no\,/\,yes & $L^3 \log L$ \\
    Harmonic mult.\ / sifting $\circledcirc$ & \ref{def:sifting} & $\sphere \to \sphere$ & any & polar\,/\,yes$^{\mathrm{a}}$ & $L^3$ \\
  \bottomrule
  \end{tabular}}
\end{table}

\section{Fields on homogeneous spaces: lifting and projection}
\label{sec:lifting}

In this section we develop the machinery on which the unification rests.  We first describe the rotation group as a principal bundle over the sphere and introduce the twist that measures the failure of a rotation to preserve a section, before describing spin fields as sections of associated bundles and as Mackey functions on the group.  The lift of a spin field to the rotation group and its adjoint, the projection, are then defined and their structure established, and the induced representation of $\SOthree$ on spin fields is identified with the rotation \eqref{eq:spin-rotation}.  Finally, the decomposition of an arbitrary function on the rotation group into a tuple of spin fields is derived.  The material of the first two subsections is classical \cite{mackey1952induced,folland2016course} and was brought to equivariant networks in \cite{cohen2019general,xu2022unified,weiler2021coordinate}; we present it in the concrete setting of the sphere and in the conventions of Sec.~\ref{sec:prelim}.

\subsection{The principal bundle $\SOthree \to \sphere$ and the twist}
\label{sec:lifting:bundle}

Let $G$ be a compact Lie group and $H$ a closed subgroup, and consider the homogeneous space $X=G/H$ with base point $x_0=H$, so that $H$ is its stabiliser and $gH$ corresponds to $g x_0$.  The quotient map $p : G \to G/H$ exhibits $G$ as a principal $H$-bundle over $G/H$: the fibre $p^{-1}(x) = gH$ over $x = gH$ is a copy of $H$ on which $H$ acts freely by right multiplication.  For $G = \SOthree$, $H = \SOtwo$ and $x_0 = \eta$ the fibre over a point $\omega$ of the sphere is the circle of unit tangent vectors at $\omega$: a rotation $\rho$ with $\rho\eta = \omega$ is determined by where it sends the North pole and a reference tangent vector there, and the third Euler angle $\gamma$ parameterises the fibre.  The bundle is illustrated in Fig.~\ref{fig:fibre-bundle}.

\begin{figure}[t]
  \centering
  \includegraphics[width=0.6\linewidth,page=2]{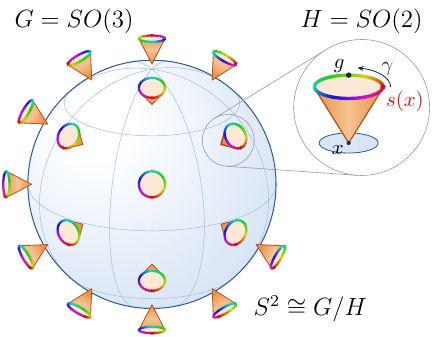}
  \caption{The rotation group as a principal $\SOtwo$-bundle over the sphere $\sphere$.  The fibre over a point $x$ is drawn as the rim of a cone standing on the sphere at $x$; its points are the rotations $R(\alpha,\beta,\gamma)$ with $R(\alpha,\beta,\gamma)\eta = x$, coloured by the fibre coordinate $\gamma$.  The red dot is the section $s(x) = R(\alpha,\beta,0)$, and a general element of the fibre is $g = s(x)R_z(\gamma)$ (inset).}
  \label{fig:fibre-bundle}
\end{figure}

A \emph{section} is a map $s : G/H \to G$ with $p \circ s = \Id$, \ie\ a choice of coset representative for every point.  On the sphere we adopt
\begin{equation}
  s(\theta,\varphi) = R(\varphi,\theta,0)\quad(0<\theta<\pi),\qquad s(\eta)=e,\qquad s(-\eta)=R_y(\pi),
  \label{eq:section}
\end{equation}
which chooses the frame aligned with the coordinate lines away from the poles and is smooth there.  The letter $s$ denotes the section in $s(x)$ and in compositions $F\circ s$, and the integer spin in representation and operator labels; we write $\sigma_0=s(\cdot)$ when an argument-free section symbol is needed for the evaluation operators below.

\begin{remark}[No global section]
\label{rem:no-global-section}
The bundle $\SOthree \to \sphere$ is the unit tangent bundle of the sphere: a section is a choice of unit tangent vector at every point, and by the hairy-ball theorem no such choice is continuous.  The section \eqref{eq:section}, which follows the meridians away from the poles, is measurable but not continuous.  The intrinsic spin-field constructions below are independent of this choice, with their coordinate expressions related by \eqref{eq:section-change}; the mixed-spin section evaluations of Sec.~\ref{sec:beyond:section} instead depend on the chosen section and do not transform by one common spin phase.  For continuous spin fields, the limiting coordinate values at the poles vanish or carry the windings $e^{\mp\ii s\varphi}$ \cite{boyle2016how,mcewen2011novel}; bare $L^2$ classes have no intrinsic point values there.
\end{remark}

The action of $G$ does not preserve a section: in general $g\, s(x) \ne s(gx)$.  The discrepancy lies in $H$.

\begin{definition}[Twist]
\label{def:twist}
The twist $\twist : G/H \times G \to H$ is defined by
\begin{equation}
  g\, s(x) = s(gx)\, \twist(x,g), \qquad \twist(x,g) = s(gx)^{-1}\, g\, s(x) .
  \label{eq:twist}
\end{equation}
We write $\twist(g) = \twist(x_0, g) = s(gx_0)^{-1} g$ for the twist at the base point.
\end{definition}

\begin{lemma}[Composition of twists]
\label{lem:twist-composition}
$\twist(x, g_1 g_2) = \twist(g_2 x, g_1)\, \twist(x, g_2)$ for all $x \in G/H$ and $g_1, g_2 \in G$.  On the sphere, with the section \eqref{eq:section}, $\twist(R(\alpha,\beta,\gamma)) = R_z(\gamma)$ for $0<\beta<\pi$, hence almost everywhere.
\end{lemma}
\begin{proof}
The composition rule follows by inserting $s(g_2x)\,s(g_2x)^{-1}$ into $s(g_1g_2x)^{-1}g_1g_2s(x)$ and regrouping (it is Eq.~(20) of \cite{cohen2019general}).  On the sphere, for $0<\beta<\pi$, $s(R(\alpha,\beta,\gamma)\eta) = R(\alpha,\beta,0)$, hence $\twist(R(\alpha,\beta,\gamma)) = R(\alpha,\beta,0)^{-1}R(\alpha,\beta,\gamma) = R_z(\gamma)$.  The details are given in Appendix~\ref{app:proofs}.
\end{proof}

Geometrically, a rotation $g$ carries the frame chosen by the section at $x$ to a frame at $gx$ that differs from the frame chosen by the section there by a rotation $\twist(x,g)$ about the local normal, as illustrated in Fig.~\ref{fig:twist}.  This residual rotation is what a spin field must compensate for, and it is the origin of the phase in \eqref{eq:spin-rotation}.

\begin{figure}[t]
  \centering
  \includegraphics[page=1,width=0.6\linewidth]{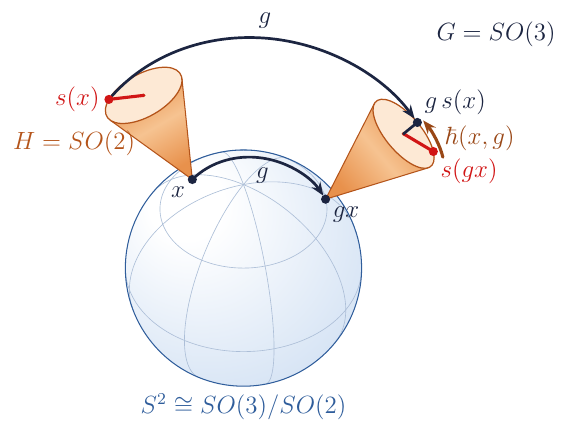}
  \caption{The twist of the section under a rotation, $g\,s(x) = s(gx)\,\twist(x,g)$.  The rotation $g$ carries the section point $s(x)$ in the fibre over $x$ to a point of the fibre over $gx$ that differs from $s(gx)$ by the right action of the twist $\twist(x,g) \in \SOtwo$.}
  \label{fig:twist}
\end{figure}

\subsection{Associated bundles and fields}
\label{sec:lifting:fields}

The irreducible unitary representations of $H = \SOtwo$ are the characters
\begin{equation}
  \rho^{(s)}(R_z(\gamma)) = e^{\ii s\gamma}, \qquad s \in \mathbb{Z},
  \label{eq:so2-irreps}
\end{equation}
and we call $s$ the spin.  Given a representation $\rho$ of $H$ on a vector space $V$, the associated vector bundle is $A = G \times_\rho V = (G \times V)/H$, whose elements are the orbits $[g, v] = \{(gh, \rho(h^{-1})v) : h \in H\}$ of pairs of a pose $g$ and a feature vector $v$, with projection $[g,v] \mapsto gH$.  A \emph{field of type $\rho$} is a section of $A$.  Scalar fields have type $\rho^{(0)}$; a real tangent vector field is the pair of fields $v_\theta \pm \ii v_\varphi$ of types $\rho^{(\pm1)}$; the trace-free symmetric part of a rank-two tensor, such as linear polarisation or shear, is the pair $Q \pm \ii U$ of types $\rho^{(\pm2)}$.  A real field of spin $s$ is thus a pair of complex fields of spins $\pm s$ related by conjugation.

A field admits two equivalent encodings.

\begin{definition}[Mackey functions]
\label{def:mackey}
A Mackey function of type $\rho$ is a map $F : G \to V$ with $F(gh) = \rho(h^{-1})F(g)$ for all $g \in G$, $h \in H$.  For $\rho = \rho^{(s)}$ the condition reads $F(g R_z(\gamma)) = e^{-\ii s\gamma}F(g)$, and we write
\begin{equation}
  \Mackey{s} = \{ F \in \Lrot : F(\,\cdot\,R_z(\gamma)) = e^{-\ii s\gamma}F \text{ in } L^2 \text{ for every fixed } \gamma \} .
  \label{eq:Ms}
\end{equation}
\end{definition}

For $F\in\Mackey{s}$ we use a fibrewise equivariant representative when restricting to a section: its value on the section is the coefficient of $e^{-\ii s\gamma}$ in its fibre Fourier series, not the restriction of an arbitrary $L^2$ representative to a measure-zero set.  The second encoding evaluates $F$ along the section, $f(x)=\sqrt{2\pi}\,F(s(x))$, the factor $\sqrt{2\pi}$ being chosen to match the normalisation of the lift and projection operators introduced in Sec.~\ref{sec:lifting:lift} below, under which $f$ and $F$ determine each other (Lemma~\ref{lem:lift}).  Changing the section to $s'(x)=s(x)R_z(\chi(x))$ changes this function to
\begin{equation}
  f'(x) = \sqrt{2\pi}\,F(s(x)R_z(\chi(x))) = e^{-\ii s\chi(x)} f(x),
  \label{eq:section-change}
\end{equation}
which is precisely the frame dependence \eqref{eq:spin-phase} of a spin-$s$ field.  Thus spin-$s$ fields in the sense of Sec.~\ref{sec:prelim:spin} are the coordinate expressions, with respect to the section \eqref{eq:section}, of type-$\rho^{(s)}$ Mackey functions.

\begin{remark}[Spin fields are sections]
\label{rem:spin-fields-are-sections}
The identification just made is the bridge between the physics and the bundle literatures.  The spin-weighted quantities of Newman \& Penrose \cite{newman1966note} and Goldberg \etal\ \cite{goldberg1967spin}, defined by their behaviour \eqref{eq:spin-phase} under a rotation of the local frame and adopted for the polarisation of the cosmic microwave background in \cite{zaldarriaga1997all,kamionkowski1997statistics}, are the sections of the line bundle associated to $\SOthree \to \sphere$ by the character $\rho^{(s)}$, and their coordinate expression depends on the choice of frame exactly as \eqref{eq:section-change} prescribes.  Boyle \cite{boyle2016how} draws the same conclusion from the impossibility of a global frame, noting that spin-weighted functions are not functions on the sphere alone but functions of a point together with a frame at that point, \ie\ functions on the rotation group; the general theory of such coordinate-independent fields on manifolds is developed by Weiler \etal\ \cite{weiler2021coordinate}.  On the sphere the oriented orthonormal frame bundle is $\SOthree$, which is distinct from a chosen local frame field.  Note that the encoding as a Mackey function is global and involves no choice of section, whereas the encoding as a spin field does; the former is therefore the natural setting in which to establish properties such as equivariance, and the latter the natural setting in which to compute.
\end{remark}

\subsection{Lift and projection}
\label{sec:lifting:lift}

The two encodings are related by the lift and its adjoint.

\begin{definition}[Lift and projection]
\label{def:lift}
For $s \in \mathbb{Z}$ define the lift $\lift{s} : \Lsph{s} \to \Lrot$ and its adjoint, the projection $\proj{s} : \Lrot \to \Lsph{s}$, by
\begin{equation}
  (\lift{s} f)(\alpha,\beta,\gamma) = \frac{1}{\sqrt{2\pi}}\, e^{-\ii s\gamma}\, f(\beta,\alpha),
  \qquad
  (\proj{s} F)(\beta,\alpha) = \frac{1}{\sqrt{2\pi}} \int_0^{2\pi} e^{\ii s\gamma}\, F(\alpha,\beta,\gamma) \dd\gamma .
  \label{eq:lift-proj}
\end{equation}
\end{definition}
In the language of Definition~\ref{def:twist}, $(\lift{s}f)(g) = \rho^{(s)}(\twist(g)^{-1})\, f(gH)/\sqrt{2\pi}$, and the projection is the integral over the fibre against the character.  The normalisation by $\sqrt{2\pi}$, the square root of the volume of the fibre, makes the lift an isometry.  For a general compact Lie group and closed stabiliser $H$, the analogous normalisation is $\sqrt{\operatorname{vol}(H)}$.

\begin{lemma}[Structure of the lift]
\label{lem:lift}
Let $s \in \mathbb{Z}$.
\begin{enumerate}
  \item[(a)] $\lift{s}$ is an isometry of $\Lsph{s}$ onto $\Mackey{s}$; $\proj{s}\lift{s} = \Id$ and $\lift{s}\proj{s}$ is the orthogonal projector onto $\Mackey{s}$.  The subspaces $\Mackey{s}$ are mutually orthogonal and $\Lrot = \widehat{\bigoplus}_{s \in \mathbb{Z}} \Mackey{s}$.
  \item[(b)] On $\Mackey{s}$ the projection is evaluation on the section: $\proj{s} F = \sqrt{2\pi}\, F \circ s$.
  \item[(c)] In harmonic space,
  \begin{equation}
    (\lift{s} f)^{\ell}_{mn} = (-1)^s \sqrt{\frac{8\pi^2}{2\ell+1}}\; \sw{s}{f}_{\ell m}\, \delta_{n,-s},
    \qquad
    (\proj{s} F)_{\ell m} = (-1)^s \sqrt{\frac{2\ell+1}{8\pi^2}}\; F^{\ell}_{m,-s} ,
    \label{eq:lift-harmonic}
  \end{equation}
  where the coefficients on the sphere are taken with respect to the spin-$s$ harmonics.
  \item[(d)] $\lift{s}\, \sw{s}{Y}_{\ell m} = (-1)^s \sqrt{\frac{2\ell+1}{8\pi^2}}\; D^{\ell *}_{m,-s}$: lifted spin harmonics are unit-norm Wigner functions.
  \item[(e)] $\piG(g)\lift{s} = \lift{s}\pis{s}(g)$ and $\proj{s}\piG(g) = \pis{s}(g)\proj{s}$ for all $g \in \SOthree$, where $\piG$ is the left regular representation of $\SOthree$ on $\Lrot$, $(\piG(g)F)(g') = F(g^{-1}g')$, and $\pis{s}$ is the representation of $\SOthree$ on spin-$s$ fields induced from $\rho^{(s)}$; both are developed in Sec.~\ref{sec:lifting:induced}, where (e) is proved as Proposition~\ref{prop:intertwining}.
\end{enumerate}
\end{lemma}
\begin{proof}
We outline the proof; the details are given in Appendix~\ref{app:proofs}.  Part (a) follows from Fubini's theorem and the Fourier series of $F$ in the fibre coordinate $\gamma$, of which $\lift{s}\proj{s}F$ is the term in $e^{-\ii s\gamma}$.  Part (b) is immediate from \eqref{eq:Ms}.  Part (d) is the relation \eqref{eq:sY-D} between spin spherical harmonics and Wigner functions, together with their orthogonality, and (c) follows from (d) by linearity: by \eqref{eq:D-sY} and the orthogonality of the complex exponentials, the Wigner coefficients of $\lift{s}f$ vanish unless $n = -s$, and $\int_0^{2\pi} e^{-\ii(s+n)\gamma}\dd\gamma = 2\pi\delta_{n,-s}$ supplies the constant.  Part (e) is Proposition~\ref{prop:intertwining} below.
\end{proof}

Part (c) shows that the lift of a spin-$s$ field occupies the single column $n = -s$ of Wigner space, so that the Wigner transform restricted to that column is simply a spin-$s$ spherical harmonic transform.  This is classical \cite{goldberg1967spin,boyle2016how} (for scalar fields it is Proposition~1 of Kondor \& Trivedi \cite{kondor2018generalization}), and it underlies Wigner transforms through spin spherical harmonic transforms \cite{mcewen2015novel}, implemented in \cite{price2024differentiable}.

\begin{remark}[Evaluation versus adjoint]
\label{rem:evaluation}
Cohen, Geiger \& Weiler \cite{cohen2019general} and Xu \etal\ \cite{xu2022unified} define the inverse of the unnormalised lift on Mackey functions by evaluation on the section, $F\mapsto F\circ s$, which by (b) equals $\proj{s}/\sqrt{2\pi}$ on $\Mackey{s}$.  Xu \etal\ also define a general convolutional projection $P^\kappa$ on group functions and implement a fixed fibre integral after the nonlinearity; with their unnormalised stabiliser measure this integral is $\sqrt{2\pi}\,\proj{s}$ and is not restricted to Mackey functions.  Kondor \& Trivedi \cite{kondor2018generalization} use the scalar fibre average.  We adopt the adjoint convention throughout: it is bounded on all of $\Lrot$, selects a definite output spin equivariantly for arbitrary kernels, and, with the isometric lift, satisfies $\proj{s}\lift{s}=\Id$.  The degree-dependent factors also cancel under section evaluation on a constrained range; the adjoint provides a consistent normalisation, with the constants of the classical reductions tracked explicitly in Sec.~\ref{sec:reductions}.  The forthcoming article \cite{mulder2026spin} adopts the same convention.
\end{remark}

\subsection{Induced representations}
\label{sec:lifting:induced}

The left regular representation $(\piG(g)F)(g') = F(g^{-1}g')$ of $G$ on $\Lrot$ commutes with right multiplication and therefore preserves every $\Mackey{s}$.  Transporting it to the sphere gives the representation induced from $\rho^{(s)}$.

\begin{definition}[Induced representation]
\label{def:induced}
$\pis{s} = \Ind_H^G \rho^{(s)}$ is the representation of $G = \SOthree$ on $\Lsph{s}$ given by $\pis{s}(g) = \proj{s}\, \piG(g)\, \lift{s}$.
\end{definition}

\begin{proposition}[Intertwining and explicit form]
\label{prop:intertwining}
$\pis{s}$ is unitary, $\piG(g)\lift{s} = \lift{s}\pis{s}(g)$, and
\begin{equation}
  (\pis{s}(g) f)(x) = \rho^{(s)}\big(\twist(x, g^{-1})^{-1}\big)\, f(g^{-1}x) .
  \label{eq:induced-explicit}
\end{equation}
On the sphere, $\pis{s}(R_\rho)$ is the rotation $\Rot_\rho$ of a spin field defined in \eqref{eq:spin-rotation}: $\twist(x, R_\rho^{-1}) = R_z(\vartheta)$ with $\vartheta$ the third Euler angle of $R_\rho^{-1}R_x$ away from the coordinate singularities, hence almost everywhere.
\end{proposition}
\begin{proof}
Writing $(\lift{s}f)(g') = \rho^{(s)}(\twist(g')^{-1})f(g'H)/\sqrt{2\pi}$, we have
\begin{equation}
  (\piG(g)\lift{s}f)(g') = \frac{1}{\sqrt{2\pi}}\,\rho^{(s)}\big(\twist(g^{-1}g')^{-1}\big)\,f(g^{-1}g'H),
\end{equation}
and Lemma~\ref{lem:twist-composition} with $x = x_0$, $g_1 = g^{-1}$ and $g_2 = g'$ splits $\twist(g^{-1}g')$ into $\twist(g'H,g^{-1})\,\twist(g')$, from which \eqref{eq:induced-explicit} and the intertwining relation follow.  Unitarity follows since $\piG$ is unitary and $\lift{s}$ is an isometry onto an invariant subspace.  On the sphere, $R_\rho^{-1}R_x = s(R_\rho^{-1}x)R_z(\vartheta)$ by the definition of $\vartheta$, hence $\twist(x,R_\rho^{-1}) = R_z(\vartheta)$ and \eqref{eq:induced-explicit} reduces to \eqref{eq:spin-rotation}.  The details are given in Appendix~\ref{app:proofs}.
\end{proof}

\begin{remark}[Geometric meaning of the spin rotation phase]
\label{rem:phase-interpretation}
The rotation $R_\rho^{-1}$ carries the reference frame $s(x)$ at $x$ to a frame at $R_\rho^{-1}x$ that differs from the chosen frame there by $\twist(x,R_\rho^{-1})=R_z(\vartheta)$ (Fig.~\ref{fig:twist}).  The phase $e^{-\ii s\vartheta}$ is the character $\rho^{(s)}$ evaluated on the inverse twist.  This residual frame rotation is already described in the spin-wavelet treatment \cite{mcewen2015directional} and by Boyle \cite{boyle2016how}; Proposition~\ref{prop:intertwining} connects those conventions to the induced-action formula of Cohen, Geiger \& Weiler \cite{cohen2019general}.  The common harmonic law \eqref{eq:rotation-coeffs} then follows because $\lift{s}$ intertwines $\Rot_\rho$ with left translation on $\SOthree$, which acts on the first Wigner index by the same $D^\ell$ for every spin.
\end{remark}

\subsection{Fibre decomposition and azimuthal band-limits}
\label{sec:lifting:fibre}

The decomposition in Lemma~\ref{lem:lift}(a) expresses every function on the rotation group as a sum of lifted spin fields.

\begin{corollary}[Functions on $\SOthree$ as tuples of spin fields]
\label{cor:fibre-decomposition}
Every $F \in \Lrot$ decomposes uniquely as $F = \sum_{s \in \mathbb{Z}} \lift{s}\proj{s}F$, its Fourier series in $\gamma$ converging in $L^2$, with $(\lift{s}\proj{s}F)^{\ell}_{mn} = F^{\ell}_{mn}\,\delta_{n,-s}$; the component $\proj{s}F$ is a spin-$s$ field whose harmonic coefficients are the Wigner column $F^{\ell}_{m,-s}$ up to the constant in \eqref{eq:lift-harmonic}.  In particular $F$ is azimuthally band-limited at $N$, $F^{\ell}_{mn} = 0$ for $\lvert n \rvert \ge N$, if and only if $F = \sum_{\lvert s \rvert < N} \lift{s} f_s$ for spin fields $f_s = \proj{s}F$.
\end{corollary}

Consequently, a function on the rotation group azimuthally band-limited at $N$ is represented by a tuple with at most $2N-1$ spin slots, one for each $\lvert s\rvert<N$, some of which may vanish.  At degree $\ell$ there are only $2\min(\ell,N-1)+1$ admissible slots.  The Wigner transform of such a function is therefore a stack of at most $2N-1$ spin spherical harmonic transforms.  This is the identity by which Wigner transforms are computed through spin spherical harmonic transforms in the sampling theorem on the rotation group of McEwen \etal\ \cite{mcewen2015novel}, as implemented in \cite{price2024differentiable}.  The identity is used in Sec.~\ref{sec:reductions} to identify directional convolutions with azimuthally band-limited, \ie\ steerable, kernels with tuples of spin-changing convolutions, and in Sec.~\ref{sec:beyond} to decompose the anisotropic constructions of Sec.~\ref{sec:review:anisotropic} into their spin components.  Figure~\ref{fig:schematic} (bottom) shows the column structure of a lifted field and Fig.~\ref{fig:kernel-support} the part of a kernel read by each of the operators defined in the next section.

\section{The generalised lifted spherical convolution}
\label{sec:convolution}

We now define the construction that unifies the convolutions of Sec.~\ref{sec:review}.  We first define the generalised lifted spherical convolution in its four forms, before deriving its harmonic form and equivariance.  The constraint on the kernel under which the projection discards nothing is then characterised, the closure of the construction under composition is established, and finally all rotation-equivariant linear maps between spin fields and functions on the rotation group are classified, with a corollary on the incompatibility of directional kernels, spherical output and equivariance.

\subsection{Definition}
\label{sec:convolution:definition}

\begin{definition}[Generalised lifted convolution]
\label{def:lifted-conv}
For a kernel $\Psi \in \Lrot$ and spins $s, s' \in \mathbb{Z}$ define the four operators
\begin{alignat}{2}
  \Cop{s}{s'}[\Psi] &= \proj{s'}\,(\Psi \conv \lift{s}\,\cdot\,) &&: \Lsph{s} \to \Lsph{s'}, \notag\\
  \Cop{s}{G}[\Psi]  &= \Psi \conv \lift{s}\,\cdot               &&: \Lsph{s} \to \Lrot, \notag\\
  \Cop{G}{s'}[\Psi] &= \proj{s'}\,(\Psi \conv \cdot\,)          &&: \Lrot \to \Lsph{s'}, \notag\\
  \Cop{G}{G}[\Psi]  &= \Psi \conv \cdot                          &&: \Lrot \to \Lrot ,
  \label{eq:lifted-conv}
\end{alignat}
with $\conv$ the group convolution of Definition~\ref{def:group-conv}.
\end{definition}

The first operator lifts a spin-$s$ field to $\SOthree$, convolves it there and projects the result to a spin-$s'$ field; there is no requirement that $s' = s$.  The other three omit the lift, the projection, or both, and are included so that convolutions with input or output on the rotation group are covered by the same statements.  Figure~\ref{fig:schematic} shows the construction and its harmonic-space form, which we now derive.

\begin{figure}[t]
  \centering
  \includegraphics[page=1]{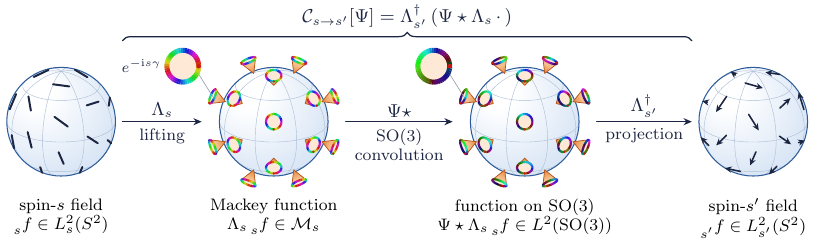}\\[1.5ex]
  \includegraphics[page=2]{figures/lifted_convolution.pdf}
  \caption{The generalised lifted spherical convolution $\Cop{s}{s'}[\Psi] = \proj{s'}(\Psi\conv\lift{s}\,\cdot\,)$, drawn for $s = 2$ and $s' = 1$.  Top: a spin-$s$ field on the sphere is lifted to a type-$s$ Mackey function on $\SOthree$, convolved on $\SOthree$ and projected to a spin-$s'$ field.  Cones are the $\SOtwo$ fibres, with the red dot at $\gamma = 0$; rims are coloured by the phase of the function on the fibre, which winds as $e^{-\ii s\gamma}$ after the lift and mixes all windings after the convolution.  Bottom: the same at one degree $\ell$.  The lift fills the Wigner column $n = -s$, the convolution multiplies by the conjugate transpose of the kernel block, and the projection reads the column $n = -s'$, so that only the entry $\Psi^{\ell}_{-s',-s}$ reaches the output, \eqref{eq:harmonic-ss}.}
  \label{fig:schematic}
\end{figure}

\subsection{Harmonic form and equivariance}
\label{sec:convolution:harmonic}

\begin{theorem}[Harmonic form]
\label{thm:harmonic}
Let $\Psi \in \Lrot$, $s, s' \in \mathbb{Z}$, and let $\widetilde{\Psi}(g) = \Psi^{*}(g^{-1})$.  For $\lvert m\rvert, \lvert n\rvert \le \ell$, with $\ell \ge \lvert s\rvert$ whenever the input is a spin-$s$ field and $\ell \ge \lvert s'\rvert$ whenever the output is a spin-$s'$ field,
\begin{align}
  \big(\Cop{s}{G}[\Psi]\, \sw{s}{f}\big)^{\ell}_{mn} &= (-1)^s \sqrt{\frac{8\pi^2}{2\ell+1}}\; \sw{s}{f}_{\ell m}\; \Psi^{\ell *}_{n,-s} ,
  \label{eq:harmonic-sG}\\
  \big(\Cop{s}{s'}[\Psi]\, \sw{s}{f}\big)_{\ell m} &= (-1)^{s+s'}\; \sw{s}{f}_{\ell m}\; \Psi^{\ell *}_{-s',-s} ,
  \label{eq:harmonic-ss}\\
  \big(\Cop{G}{s'}[\Psi]\, F\big)_{\ell m} &= (-1)^{s'} \sqrt{\frac{2\ell+1}{8\pi^2}}\; \sum_{m'=-\ell}^{\ell} F^{\ell}_{mm'}\; \Psi^{\ell *}_{-s',m'} ,
  \label{eq:harmonic-Gs}\\
  \big(\Cop{G}{G}[\Psi]\, F\big)^{\ell}_{mn} &= \sum_{m'=-\ell}^{\ell} F^{\ell}_{mm'}\; \Psi^{\ell *}_{nm'} ,
  \label{eq:harmonic-GG}
\end{align}
and the output of $\Cop{s}{s'}[\Psi]$ has no component of degree $\ell < \max(\lvert s\rvert,\lvert s'\rvert)$.  All four operators are bounded and $\SOthree$-equivariant for every $\Psi$: $\Cop{s}{s'}[\Psi]\,\pis{s}(g) = \pis{s'}(g)\,\Cop{s}{s'}[\Psi]$ and correspondingly with $\piG$ in place of $\pis{s}$ or $\pis{s'}$.  Their adjoints are $\big(\Cop{s}{G}[\Psi]\big)^{\dagger} = \Cop{G}{s}[\widetilde\Psi]$, $\big(\Cop{s}{s'}[\Psi]\big)^{\dagger} = \Cop{s'}{s}[\widetilde\Psi]$ and $\big(\Cop{G}{G}[\Psi]\big)^{\dagger} = \Cop{G}{G}[\widetilde\Psi]$.
\end{theorem}
\begin{proof}
By Lemma~\ref{lem:lift}(c) the Wigner coefficients of the lifted field are supported on the single column $m' = -s$, so that the sum over $m'$ in \eqref{eq:group-conv-harmonic} collapses to a single term and \eqref{eq:harmonic-sG} follows; applying $\proj{s'}$ by \eqref{eq:lift-harmonic} reads the column $n = -s'$ and multiplies by $(-1)^{s'}\sqrt{(2\ell+1)/8\pi^2}$, so that the degree-dependent constants cancel and \eqref{eq:harmonic-ss} follows, while \eqref{eq:harmonic-Gs} is $\proj{s'}$ applied to \eqref{eq:group-conv-harmonic}.  Equivariance follows since $\lift{s}$, $\proj{s'}$ and $\Psi\conv$ each intertwine the relevant representations, by Lemma~\ref{lem:lift}(e) and the left invariance of the Haar measure, and the adjoint identities follow from Fubini's theorem, the adjoint of $\Psi\conv$ being $\widetilde\Psi\conv$.  The details are given in Appendix~\ref{app:proofs}.
\end{proof}

Equation \eqref{eq:harmonic-ss} is the central formula of the article: a spin-$s$ to spin-$s'$ lifted convolution multiplies each spin-$s$ harmonic coefficient by a single Wigner coefficient of the kernel, conjugated and up to a sign, $(-1)^{s+s'}\Psi^{\ell *}_{-s',-s}$, with no degree-dependent constant.  The equivariant families listed in Table~\ref{tab:review} arise from the different ways in which the literature identifies spherical and group kernels, as we show in Sec.~\ref{sec:reductions}.  Figure~\ref{fig:kernel-support} in Sec.~\ref{sec:reductions} shows the part of the kernel block read by each of the four operators.

\begin{proposition}[Composition]
\label{prop:composition}
The four families of Definition~\ref{def:lifted-conv} are closed under composition, and composing multiplies kernel blocks: for kernels $\Psi_1, \Psi_2 \in \Lrot$,
\begin{equation}
  \Cop{G}{G}[\Psi_2]\,\Cop{G}{G}[\Psi_1] = \Cop{G}{G}[\Psi_3], \qquad \Psi_3^{\ell} = \Psi_2^{\ell}\,\Psi_1^{\ell},
  \label{eq:composition-GG}
\end{equation}
and likewise $\Cop{G}{s''}[\Psi_2]\,\Cop{s}{G}[\Psi_1] = \Cop{s}{s''}[\Psi_3]$, $\Cop{G}{G}[\Psi_2]\,\Cop{s}{G}[\Psi_1] = \Cop{s}{G}[\Psi_3]$ and $\Cop{G}{s''}[\Psi_2]\,\Cop{G}{G}[\Psi_1] = \Cop{G}{s''}[\Psi_3]$ with the same $\Psi_3$.  When the intermediate stage is a spin-$s'$ field only the intermediate index $-s'$ is retained: $\Cop{s'}{s''}[\Psi_2]\,\Cop{s}{s'}[\Psi_1] = \Cop{s}{s''}[\Psi_3]$ with $(\Psi_3)^{\ell}_{-s'',-s} = (\Psi_2)^{\ell}_{-s'',-s'}\,(\Psi_1)^{\ell}_{-s',-s}$, so that the multipliers of \eqref{eq:harmonic-ss} multiply, $c^{(3)}_\ell = c^{(2)}_\ell c^{(1)}_\ell$; correspondingly $\Cop{s'}{G}[\Psi_2]\,\Cop{s}{s'}[\Psi_1] = \Cop{s}{G}[\Psi_3]$ with $(\Psi_3)^{\ell}_{n,-s} = (\Psi_2)^{\ell}_{n,-s'}(\Psi_1)^{\ell}_{-s',-s}$.
\end{proposition}
\begin{proof}
By \eqref{eq:group-conv-harmonic}, $(\Psi_2\conv(\Psi_1\conv F))^{\ell} = F^{\ell}\,\Psi_1^{\ell\dagger}\,\Psi_2^{\ell\dagger} = F^{\ell}(\Psi_2^{\ell}\Psi_1^{\ell})^{\dagger}$, which is \eqref{eq:composition-GG}; inserting $\lift{s'}\proj{s'}$ between the two convolutions, whose action in Wigner space is the projector onto the column $n = -s'$ by Corollary~\ref{cor:fibre-decomposition}, retains only the intermediate index $-s'$, and the remaining cases follow by omitting the outer lift or projection.  See Appendix~\ref{app:proofs}.
\end{proof}

Consequently, consecutive sphere-to-sphere convolutions require no intermediate transforms, since their multipliers simply multiply, whereas an intermediate stage on the rotation group between two spin fields sums over all intermediate spins, which is the content of Proposition~\ref{prop:band-limited} below.  The forthcoming article \cite{mulder2026spin} relies on this rule to stack layers.

\subsection{Lossless projection and the kernel constraint}
\label{sec:convolution:constraint}

Only one column of $\Psi$ is read by $\Cop{s}{G}[\Psi]$ and only one entry by $\Cop{s}{s'}[\Psi]$.  The projection discards nothing precisely when the intermediate function on $\SOthree$ is already a Mackey function of the target type.

\begin{corollary}[Lossless projection and kernel constraint]
\label{cor:constraint}
The following are equivalent:
\begin{enumerate}
  \item[(i)] $\Psi \conv \lift{s} f \in \Mackey{s'}$ for all $f \in \Lsph{s}$;
  \item[(ii)] $\Psi^{\ell}_{n,-s} = 0$ for all $n \ne -s'$ and all $\ell \ge \lvert s\rvert$;
  \item[(iii)] the type-$s$ right-isotypic part $\Psi^{(s)}(g) = \frac{1}{2\pi}\int_0^{2\pi}\Psi(gR_z(\gamma))\,e^{\ii s\gamma}\dd\gamma$ of the kernel satisfies the two-sided constraint
  \begin{equation}
    \Psi^{(s)}\big(R_z(\chi_1)\, g\, R_z(\chi_2)\big) = e^{-\ii s'\chi_1}\; \Psi^{(s)}(g)\; e^{-\ii s\chi_2}
    \qquad\text{for all } \chi_1, \chi_2 .
    \label{eq:two-sided}
  \end{equation}
\end{enumerate}
In that case $\Cop{s}{s'}[\Psi] f = \sqrt{2\pi}\,(\Psi \conv \lift{s} f)\circ s$, the constrained part of the kernel is
\begin{equation}
  \Psi^{(s)}(\alpha,\beta,\gamma) = \sum_{\ell \ge \max(\lvert s\rvert,\lvert s'\rvert)} \frac{2\ell+1}{8\pi^2}\; \Psi^{\ell}_{-s',-s}\; e^{-\ii s'\alpha}\, d^{\ell}_{-s',-s}(\beta)\, e^{-\ii s\gamma} ,
  \label{eq:constrained-kernel}
\end{equation}
where the series converges in $L^2$.  For every kernel lifted from a spin-$s$ field, $\Psi=\lift{s}\psi$, whether or not projection is lossless,
\begin{equation}
  \big(\Cop{s}{s'}[\lift{s}\psi]\, \sw{s}{f}\big)_{\ell m} = (-1)^{s'} \sqrt{\frac{8\pi^2}{2\ell+1}}\; \sw{s}{f}_{\ell m}\; \sw{s}{\psi}^{*}_{\ell,-s'} .
  \label{eq:harmonic-ss-lifted}
\end{equation}
For such a lifted kernel, losslessness is equivalent to $\sw{s}{\psi}_{\ell m}=\sw{s}{\psi}_{\ell,-s'}\delta_{m,-s'}$.
\end{corollary}
\begin{proof}
The equivalence of (i) and (ii) follows from \eqref{eq:harmonic-sG} and the characterisation of $\Mackey{s'}$ as the column $n = -s'$ of Wigner space (Lemma~\ref{lem:lift}); that of (ii) and (iii) by expanding $\Psi^{(s)}$ in Wigner functions and noting that $D^{\ell *}_{mn}(R_z(\chi_1)gR_z(\chi_2)) = e^{\ii m\chi_1}D^{\ell *}_{mn}(g)e^{\ii n\chi_2}$.  The evaluation formula is Lemma~\ref{lem:lift}(b), \eqref{eq:constrained-kernel} is the Wigner expansion restricted to the entry $(-s',-s)$, and \eqref{eq:harmonic-ss-lifted} is \eqref{eq:harmonic-ss} with Lemma~\ref{lem:lift}(c).  See Appendix~\ref{app:proofs}.
\end{proof}

\subsection{Classification of equivariant maps}
\label{sec:convolution:classification}

Theorem~\ref{thm:harmonic} shows that every generalised lifted spherical convolution between spin fields is a multiplier in spin spherical harmonic space.  We now prove the converse.

\begin{theorem}[Classification]
\label{thm:classification}
Let $T : \Lsph{s} \to \Lsph{s'}$ be a bounded linear map with $T\pis{s}(g) = \pis{s'}(g)T$ for all $g \in \SOthree$.  Then there is a unique bounded sequence $(c_\ell)_{\ell \ge \max(\lvert s\rvert,\lvert s'\rvert)}$ such that
\begin{equation}
  (T\, \sw{s}{f})_{\ell m} = c_\ell\; \sw{s}{f}_{\ell m}
  \qquad \text{for all } \ell \ge \max(\lvert s\rvert,\lvert s'\rvert),\ \lvert m \rvert \le \ell ,
  \label{eq:classification}
\end{equation}
and $T$ annihilates the input components with degrees absent from the target; conversely every bounded sequence defines such a $T$, with $\lVert T\rVert = \sup_\ell\lvert c_\ell\rvert$.  Moreover:
\begin{enumerate}
  \item[(a)] $T = \Cop{s}{s'}[\Psi]$ for some $\Psi \in \Lrot$ if and only if $\sum_\ell (2\ell+1)\lvert c_\ell\rvert^2 < \infty$, \ie\ if and only if $T$ is Hilbert--Schmidt.  The kernel is not unique, but the representative whose entire kernel satisfies the two-sided law \eqref{eq:two-sided}, with $\Psi$ in place of $\Psi^{(s)}$, is unique: $\Psi_T{}^{\ell}_{-s',-s} = (-1)^{s+s'}c^{*}_{\ell}$ and all other Wigner coefficients are zero.
  \item[(b)] Every bounded $T$ is $\Cop{s}{s'}[\Psi_T]$ in the sense of Fourier multipliers, with $\Psi_T$ the distribution on $\SOthree$ with the Wigner coefficients of (a).
  \item[(c)] Set $d_\ell = 2\ell+1$ and $a_\ell = d_\ell/(8\pi^2)$.  The equivariant maps of the other three types act on coefficients as
  \begin{equation}
    (T_{s\to G}f)^\ell_{mn} = u^\ell_n\,\sw{s}{f}_{\ell m},\qquad
    (T_{G\to s'}F)_{\ell m} = \sum_k v^\ell_k F^\ell_{mk},\qquad
    (T_{G\to G}F)^\ell = F^\ell B^\ell .
    \label{eq:classification-blocks}
  \end{equation}
  They are bounded if and only if the respective bounds
  \begin{align}
    \lVert T_{s\to G}\rVert &= \sup_{\ell\ge\lvert s\rvert}\sqrt{a_\ell}\,\lVert u^\ell\rVert_2 < \infty, \notag\\
    \lVert T_{G\to s'}\rVert &= \sup_{\ell\ge\lvert s'\rvert}\frac{\lVert v^\ell\rVert_2}{\sqrt{a_\ell}} < \infty,\qquad
    \lVert T_{G\to G}\rVert = \sup_{\ell\ge0}\lVert B^\ell\rVert_{\mathrm{op}} < \infty
    \label{eq:classification-norms}
  \end{align}
  hold.  Degrees absent from the input or output spin space contribute zero.  These maps are realised by the column, row and full block of a kernel, respectively, with
  \begin{equation}
    u^\ell_n = (-1)^s a_\ell^{-1/2}\Psi^{\ell *}_{n,-s},\qquad
    v^\ell_k = (-1)^{s'}a_\ell^{1/2}\Psi^{\ell *}_{-s',k},\qquad
    B^\ell = \Psi^{\ell\dagger} .
    \label{eq:classification-kernels}
  \end{equation}
  Setting all unread entries to zero gives a canonical distributional kernel, which is square integrable if and only if the map is Hilbert--Schmidt.  The corresponding summability conditions are given in Appendix~\ref{app:proofs:classification}.
\end{enumerate}
\end{theorem}
\begin{proof}
By Frobenius reciprocity the multiplicity of the degree-$\ell$ irreducible representation of $\SOthree$ in $\Ind_H^G\rho^{(s)}$ equals the multiplicity of $\rho^{(s)}$ in its restriction to $\SOtwo$, which is one if $\ell \ge \lvert s\rvert$ and zero otherwise, so that $\Lsph{s}$ decomposes without multiplicity into the subspaces spanned by the $\sw{s}{Y}_{\ell m}$, $\lvert m\rvert \le \ell$, on which $\SOthree$ acts by the same matrices $D^{\ell}$ for every spin by \eqref{eq:rotation-coeffs}.  Schur's lemma then forces $T$ to map each such subspace to its counterpart of the same degree by a scalar $c_\ell$, and to annihilate degrees absent from the target; boundedness of $T$ is equivalent to $\sup_\ell \lvert c_\ell\rvert < \infty$.  Part (a) follows by comparing \eqref{eq:harmonic-ss} with \eqref{eq:classification} and computing $\lVert\Psi_T\rVert^2 = \sum_\ell\frac{2\ell+1}{8\pi^2}\lvert c_\ell\rvert^2$, part (b) by extending $\Psi\conv$ to distributional kernels with bounded Wigner coefficients, and part (c) by the same argument applied to the left regular representation, in which the degree-$\ell$ representation occurs with multiplicity $2\ell+1$, one copy per column.  Parseval's identity weights each Wigner block by $a_\ell$, whereas the spin harmonic coefficients have unit weight, giving the norms \eqref{eq:classification-norms}.  The details are given in Appendix~\ref{app:proofs}.
\end{proof}

For scalar fields Theorem~\ref{thm:classification} is Theorem~1 and Lemma~10 of Kondor \& Trivedi \cite{kondor2018generalization} and, for $s = s' = 0$, the characterisation of isotropic operators of Kennedy \etal\ \cite{kennedy2011azimuthally}; for general types it is contained in Theorems~3.2--3.4 of Cohen, Geiger \& Weiler \cite{cohen2019general}, which for the sphere describe the kernels as functions on $[0,\pi]$ unconstrained except at the poles, and in Proposition~1 and Corollary~1 of Xu \etal\ \cite{xu2022unified}.  Xu \etal\ give the character lift, the single-column harmonic representation, the single-entry input/output kernel and a completeness argument explicitly, while Esteves \etal\ \cite{esteves2020spin} parameterise spin-changing multipliers directly.  Here we collect the four cases in the conventions of Sec.~\ref{sec:prelim}, including the degree ranges, the exact norm bounds and the distinction between square integrable and distributional kernels.  The identity already requires the distributional interpretation of part (b).

The differential operators of Sec.~\ref{sec:reductions:physics} require a separate extension beyond bounded maps.  For a possibly unbounded sequence $(c_\ell)$, set $c_\ell=0$ on input degrees absent from the target and define
\begin{equation}
  \mathcal{D}(T) = \left\{ f\in\Lsph{s} : \sum_{\ell m}\lvert c_\ell\rvert^2\lvert\sw{s}{f}_{\ell m}\rvert^2<\infty\right\} .
  \label{eq:multiplier-domain}
\end{equation}
The coefficient formula \eqref{eq:classification} defines a closed, densely defined operator on this rotation-invariant domain and commutes there with rotations.  Polynomial growth of $c_\ell$ gives a distributional kernel; an arbitrary sequence still acts on finite harmonic sums but need not define a distribution.  For $\eth$ and $\bar\eth$ the domain is the first-order spin Sobolev space, and for the second-order operators below it is the second-order space, with the finitely many zero low-degree multipliers imposing no additional restriction.

\begin{corollary}[Equivariant projections and lifts]
\label{cor:projection-general}
Every bounded $\SOthree$-equivariant linear map $\Lrot \to \Lsph{s'}$ is $\Cop{G}{s'}[\Psi]$ for a kernel (in general a distribution) whose only nonzero row is $n=-s'$ in each degree $\ell\ge\lvert s'\rvert$.  If the map acts by $(TF)_{\ell m}=\sum_{m'}c^\ell_{m'}F^\ell_{mm'}$, then $c^\ell_{m'}=(-1)^{s'}\sqrt{(2\ell+1)/(8\pi^2)}\,\Psi^{\ell *}_{-s',m'}$, subject to \eqref{eq:classification-norms}.  The adjoint projection $\proj{s'}$ is the case $\Psi^{\ell *}_{-s',m'}=\delta_{m',-s'}$, which reads one Wigner column.  Dually, every bounded equivariant map $\Lsph{s} \to \Lrot$ is $\Cop{s}{G}[\Psi]$ for a kernel with a single column $m'=-s$, and $\lift{s}$ is the case $\Psi^{\ell *}_{n,-s}=\delta_{n,-s}$.
\end{corollary}
\begin{proof}
This is Theorem~\ref{thm:classification}(c) with the normalisations of Lemma~\ref{lem:lift}(c).
\end{proof}

The learned projecting convolution of the forthcoming article \cite{mulder2026spin} uses the family in Corollary~\ref{cor:projection-general}.

\begin{corollary}[Incompatibility of directional kernels, spherical output and equivariance]
\label{cor:nogo}
The kernel, in general a distribution, of a bounded equivariant map from one fixed input spin to one fixed output spin, $T : \Lsph{s} \to \Lsph{s'}$, is
\begin{equation}
  K(x,y) = (-1)^{s+s'} \sum_{\ell} c_\ell\, \frac{2\ell+1}{4\pi}\; D^{\ell}_{-s',-s}\big(s(x)^{-1}s(y)\big) ,
  \label{eq:nogo-kernel}
\end{equation}
a radial kernel with the frame phases forced by \eqref{eq:section-change}; for $s=s'=0$ it is $\sum_\ell c_\ell(2\ell+1)P_\ell(x\cdot y)/(4\pi)$.  Apart from the angular characters prescribed by the input and output spins, such a map cannot retain independent directional kernel content.  Thus, beyond the angular characters prescribed by the spins, directional kernels, a fixed-spin output on the sphere and equivariance cannot be achieved simultaneously for a fixed input spin.
\end{corollary}
\begin{proof}
By \eqref{eq:classification}, $K(x,y) = \sum_\ell c_\ell\sum_m\sw{s'}{Y}_{\ell m}(x)\sw{s}{Y}^{*}_{\ell m}(y)$, and the sum over $m$ is evaluated by the addition theorem for spin spherical harmonics, which follows from \eqref{eq:sY-D} and $D^{\ell}(s(x)^{-1}s(y)) = D^{\ell}(s(x))^{\dagger}D^{\ell}(s(y))$; the second Euler angle of $s(x)^{-1}s(y)$ is the geodesic distance between $x$ and $y$, and the first and third are the phases forced by the frames at $x$ and $y$.  See Appendix~\ref{app:proofs}.
\end{proof}

Corollary~\ref{cor:nogo} adds an equivariance qualification to the criteria of Roddy \& McEwen \cite{roddy2021sifting}: between two fixed spins the kernel is radial with prescribed frame phases, not necessarily an axisymmetric spin filter in the sense of Definition~\ref{def:spin-axisymmetric}.  For example, a scalar filter $Y_{\ell,-1}$ can produce an equivariant spin-$1$ output; its prescribed angular character is not independent directionality.  Retaining additional directional content from a fixed input spin requires a tuple of output spins, equivalently a function on the rotation group.  The restriction does not apply to a group-to-spin synthesis, whose input already contains several spin components.

\section{Reduction to existing spherical convolutions}
\label{sec:reductions}

In this section we recover the rotation-equivariant convolutions treated in Sec.~\ref{sec:review} as generalised lifted spherical convolutions.  For each family we state the input, the part of the Wigner block $\Psi^{\ell}_{nm'}$ of the kernel that is read, and the output, and we recover the harmonic form together with its constant.  The results are collected in Table~\ref{tab:taxonomy}, and Fig.~\ref{fig:kernel-support} shows the part of the kernel block that each operator reads.  Throughout, $\lift{s}\psi$ denotes a kernel lifted from a spin-$s$ field $\psi$ on the sphere; by Lemma~\ref{lem:lift}(c) such a kernel has Wigner support in the single column $m'=-s$.  The reductions follow by direct calculation from Theorem~\ref{thm:harmonic} and Lemma~\ref{lem:lift}; proofs not immediate from the harmonic formulas are supplied in Appendix~\ref{app:proofs}.

\begin{figure}[t]
  \centering
  \includegraphics[page=3]{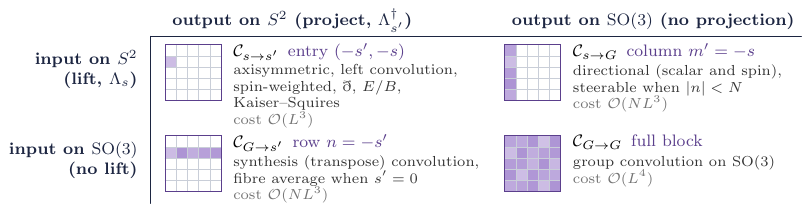}
  \caption{The part of the kernel block $\Psi^{\ell}_{nm'}$ read by each of the four operators of Definition~\ref{def:lifted-conv} at one degree $\ell$: a single entry $(-s',-s)$ for sphere to sphere, the column $m' = -s$ for sphere to $\SOthree$, the row $n = -s'$ for $\SOthree$ to sphere, and the full block for $\SOthree$ to $\SOthree$, together with the convolutions of this section that each case covers and the costs of Sec.~\ref{sec:computation}.}
  \label{fig:kernel-support}
\end{figure}

\begin{table}[t]
  \caption{Taxonomy of spherical convolutions as generalised lifted spherical convolutions.  Kernel support is the set of Wigner coefficients $\Psi^{\ell}_{nm'}$ of the kernel that are read; $\lift{s}\psi$ denotes a kernel lifted from a spin-$s$ field $\psi$ on the sphere.}
  \label{tab:taxonomy}
  \scriptsize\setlength{\tabcolsep}{2.5pt}
  \centering
  \resizebox{\textwidth}{!}{%
  \begin{tabular}{llll}
    \toprule
    Convolution (Prop.) & Operator & Input $\to$ output & Kernel support \\
    \midrule
    Directional, scalar (\ref{prop:directional}) & $\Cop{0}{G}[\lift{0}\psi]$ & $\Lsph{0}$ $\to$ $\SOthree$ & col.\ $m'=0$ \\
    Directional, spin (\ref{prop:directional}) & $\Cop{s}{G}[\lift{s}\psi]$ & $\Lsph{s}$ $\to$ $\SOthree$ & col.\ $m'=-s$ \\
    Synthesis, transpose (\ref{prop:synthesis}) & $\Cop{G}{s}[\widetilde{\lift{s}\psi}]$ & $\Lrot$ $\to$ $\Lsph{s}$ & row $n=-s$ \\
    Axisymmetric (\ref{prop:axisymmetric}) & $(2\pi)^{-1/2}\Cop{0}{0}[\lift{0}\psi]$ & $\Lsph{0}$ $\to$ $\Lsph{0}$ & entry $(0,0)$ \\
    Left (\ref{prop:driscoll-healy}) & $\Cop{0}{0}[\Psi_h]$ & $\Lsph{0}$ $\to$ $\Lsph{0}$ & entry $(0,0)$ \\
    Spin $s \to s'$ (\ref{prop:spin-change}) & $\Cop{s}{s'}[\Psi]$ & $\Lsph{s}$ $\to$ $\Lsph{s'}$ & entry $(-s',-s)$ \\
    Band-limited directional (\ref{prop:band-limited}) & $\bigoplus_{\lvert s'\rvert<N}\Cop{s}{s'}[\Psi]$ & $\Lsph{s}$ $\to$ $\bigoplus_{\lvert s'\rvert<N}\Lsph{s'}$ & col.\ $-s$, $\lvert n\rvert<N$ \\
    Group convolution (\ref{prop:group}) & $\Cop{G}{G}[\Psi]$ & $\Lrot$ $\to$ $\Lrot$ & all \\
    $\eth$, $E/B$, Kaiser--Squires (\ref{prop:physics}) & $\Cop{s}{s'}[\Psi_T]$ & $\Lsph{s}$ $\to$ $\Lsph{s'}$ & entry $(-s',-s)$, distrib. \\
  \bottomrule
  \end{tabular}}
\end{table}

\subsection{Directional convolutions and synthesis}
\label{sec:reductions:directional}

We begin with the directional convolution, which is the lifted convolution with no projection and a kernel lifted from the sphere.

\begin{proposition}[Directional convolutions]
\label{prop:directional}
For $\sw{s}{f}, \sw{s}{\psi} \in \Lsph{s}$ (including $s = 0$),
\begin{equation}
  \sw{s}{f} \circledast \sw{s}{\psi} = \Cop{s}{G}[\lift{s}\psi]\; \sw{s}{f},
  \qquad
  (\sw{s}{f} \circledast \sw{s}{\psi})^{\ell}_{mn} = \frac{8\pi^2}{2\ell+1}\; \sw{s}{f}_{\ell m}\, \sw{s}{\psi}^{*}_{\ell n} ,
\end{equation}
exactly, with no residual constant.  The kernel $\lift{s}\psi$ has Wigner support in the column $m' = -s$.
\end{proposition}
\begin{proof}
By the intertwining relation of Lemma~\ref{lem:lift}(e) and the isometry of the lift, $(\lift{s}\psi\conv\lift{s}f)(g) = \langle\lift{s}f, \piG(g)\lift{s}\psi\rangle = \langle\lift{s}f,\lift{s}\pis{s}(g)\psi\rangle = \langle f, \pis{s}(g)\psi\rangle$, which is Definition~\ref{def:spin-directional} since $\pis{s}(g) = \Rot_g$ by Proposition~\ref{prop:intertwining}; the harmonic form is \eqref{eq:harmonic-sG} with Lemma~\ref{lem:lift}(c).
\end{proof}

Proposition~\ref{prop:directional} covers the directional convolutions of Wandelt \& G\'orski \cite{wandelt2001fast}, Challinor \etal\ \cite{challinor2000all}, McEwen \etal\ \cite{mcewen2007fast,mcewen2013directional,mcewen2015directional} and Wiaux \etal\ \cite{wiaux2006fast,wiaux2007complex}, and the $\sphere$ correlation of Cohen \etal\ \cite{cohen2018spherical} for real signals with the measure normalised to unity.  Note that the absence of a residual constant is a consequence of the isometric normalisation of the lift: the directional convolution on the sphere \emph{is} the group convolution of the lifts.

\begin{proposition}[Synthesis and transpose convolution]
\label{prop:synthesis}
For $W \in \Lrot$ and $\sw{s}{\psi} \in \Lsph{s}$,
\begin{gather}
  W \circledast^{\dagger} \sw{s}{\psi} = \int_{\SOthree} W(\rho)\, \big(\pis{s}(\rho)\,\sw{s}{\psi}\big) \dmu(\rho) = \Cop{G}{s}\big[\widetilde{\lift{s}\psi}\big]\, W = \big(\Cop{s}{G}[\lift{s}\psi]\big)^{\dagger} W , \\
  (W \circledast^{\dagger} \sw{s}{\psi})_{\ell m} = \sum_n W^{\ell}_{mn}\, \sw{s}{\psi}_{\ell n} .
\end{gather}
Synthesis is the adjoint of analysis, and the flipped kernel $\widetilde{\lift{s}\psi}$ has Wigner support in the row $n = -s$.
\end{proposition}
\begin{proof}
Since $\pis{s}(\rho)\psi = \proj{s}\piG(\rho)\lift{s}\psi$ and $\proj{s}$ is bounded, it may be taken outside the integral, and $\int W(\rho)(\piG(\rho)\Phi)(g)\dmu(\rho) = (\widetilde\Phi\conv W)(g)$ since $\widetilde\Phi^{*}(g^{-1}\rho)=\Phi(\rho^{-1}g)$; the harmonic form follows from \eqref{eq:harmonic-Gs} and the adjoint identity from Theorem~\ref{thm:harmonic}.
\end{proof}

Proposition~\ref{prop:synthesis} recovers the Haar-measure synthesis step of the directional wavelet transforms of \cite{mcewen2007fast,wiaux2008exact,mcewen2013directional,mcewen2015directional}, whose synthesis kernels compensate the constant $8\pi^2/(2\ell+1)$ of \eqref{eq:spin-directional-harmonic} either through the dual kernel $\frac{2\ell+1}{8\pi^2}\psi_{\ell n}$ \cite{wiaux2008exact,mcewen2013directional} or through the admissibility condition \cite{mcewen2015directional}, so that the sum over scales of analysis followed by synthesis is the identity.  Consequently, the admissibility condition \eqref{eq:spin-admissibility} is the statement that the sum of the operators $A_j^{\dagger}A_j$ over the scales, with $A_j$ the analysis operators of Proposition~\ref{prop:directional} and Proposition~\ref{prop:axisymmetric} below, is the identity, \ie\ that the wavelets form a Parseval frame.

\subsection{Axisymmetric convolutions}
\label{sec:reductions:axisymmetric}

\begin{proposition}[Axisymmetric convolution]
\label{prop:axisymmetric}
For axisymmetric $\psi$, $\psi_{\ell m} = \psi_{\ell 0}\delta_{m0}$,
\begin{equation}
  f \odot \psi = \frac{1}{\sqrt{2\pi}}\, \Cop{0}{0}[\lift{0}\psi]\, f ,
  \qquad
  (f \odot \psi)_{\ell m} = \sqrt{\frac{4\pi}{2\ell+1}}\; f_{\ell m}\, \psi^{*}_{\ell 0} ,
\end{equation}
and $\lift{0}\psi$ satisfies the constraint \eqref{eq:two-sided} with $s = s' = 0$.  The same holds for the axisymmetric spin convolution of Definition~\ref{def:spin-axisymmetric}, which is $(2\pi)^{-1/2}\Cop{s}{0}[\lift{s}\Phi]$ with harmonic form \eqref{eq:spin-axisymmetric-harmonic}.  For a non-axisymmetric $\psi$ the left-hand side is the standard correlation of Definition~\ref{def:standard-correlation}, which is not of the form $\Cop{0}{0}$ (Sec.~\ref{sec:beyond}).
\end{proposition}
\begin{proof}
By Proposition~\ref{prop:directional} evaluated at the section, $\langle f,\Rot_\omega\psi\rangle = (\Cop{0}{G}[\lift{0}\psi]f)(s(\omega))$, and since the output lies in $\Mackey{0}$ when $\psi$ is axisymmetric (Corollary~\ref{cor:constraint}), Lemma~\ref{lem:lift}(b) gives $(\Cop{0}{G}[\lift{0}\psi]f)\circ s = (2\pi)^{-1/2}\proj{0}\Cop{0}{G}[\lift{0}\psi]f$; the constant follows from \eqref{eq:harmonic-ss-lifted}.  The spin case is identical with $s' = 0$.
\end{proof}

Proposition~\ref{prop:axisymmetric} recovers the axisymmetric convolution of \cite{mcewen2007fast,kennedy2011azimuthally,roddy2021sifting}, the axisymmetric wavelets and needlets of \cite{narcowich2006localized,leistedt2013s2let}, and the scaling coefficients of the spin wavelet transform of \cite{mcewen2015directional}, which are a spin $s \to 0$ convolution.  The factor $(2\pi)^{-1/2}$ is the difference between evaluating a Mackey function on the section and projecting it with the adjoint.

\begin{proposition}[Left convolution]
\label{prop:driscoll-healy}
For any $h \in L^2(\sphere)$ let $\Psi_h(g) = h^{*}(g^{-1}\eta)$.  Then
\begin{equation}
  f \star_{\mathrm{L}} h = \Cop{0}{0}[\Psi_h]\, f ,
  \qquad
  (f \star_{\mathrm{L}} h)_{\ell m} = 2\pi \sqrt{\frac{4\pi}{2\ell+1}}\; f_{\ell m}\, h_{\ell 0} .
\end{equation}
Only the entry $(\Psi_h)^{\ell}_{00}$ is read, so that only the zonal part of $h$ enters; the factor $2\pi$ is the fibre volume that appears because \eqref{eq:driscoll-healy} lifts $f$ and $h$ without the normalisation of Definition~\ref{def:lift}.  The marginalisation of the bilinear directional correlation over the third Euler angle used in \cite{kennedy2011azimuthally,esteves2018learning} is $\sqrt{2\pi}\,\proj{0}$.  Applied instead to the sesquilinear correlation of Definition~\ref{def:directional}, it conjugates the spherical kernel.
\end{proposition}
\begin{proof}
By direct substitution, $(\Psi_h\conv\lift{0}f)(g) = (2\pi)^{-1/2}\int h(g'^{-1}g\eta)f(g'\eta)\dmu(g')$; applying $\proj{0}$ and noting that $R_z(\gamma)\eta = \eta$ gives \eqref{eq:driscoll-healy}, and the constant follows by evaluating $(\Psi_h)^{\ell}_{00}$ with \eqref{eq:D-sY}.  See Appendix~\ref{app:proofs}.
\end{proof}

Proposition~\ref{prop:driscoll-healy} recovers Theorem~3 of Kennedy \etal\ \cite{kennedy2011azimuthally}, with their factor of $1/2\pi$, the convolution of Esteves \etal\ \cite{esteves2018learning} and the convolution theorem of Bonev \etal\ \cite{bonev2023spherical}.  The two convolutions in Definitions~\ref{def:axisymmetric} and \ref{def:driscoll-healy} read the same kernel entry but differ in both normalisation and kernel conjugation: the axisymmetric correlation is sesquilinear, whereas the left convolution is bilinear, with the conjugation absorbed in $\Psi_h(g)=h^{*}(g^{-1}\eta)$.

\subsection{Spin-changing convolutions and steerable kernels}
\label{sec:reductions:spin}

\begin{proposition}[Spin-changing convolutions]
\label{prop:spin-change}
Every bounded multiplier $\sw{s'}{f}_{\ell m}=c_\ell\,\sw{s}{f}_{\ell m}$, for $\ell\ge\max(\lvert s\rvert,\lvert s'\rvert)$ and annihilating input degrees absent from the target, is $\Cop{s}{s'}[\Psi_T]$ in the Fourier-multiplier sense of Theorem~\ref{thm:classification}(b), with $\Psi_T{}^{\ell}_{-s',-s}=(-1)^{s+s'}c^{*}_{\ell}$.  The canonical kernel is generally a distribution and belongs to $L^2$ exactly when $\sum_\ell(2\ell+1)\lvert c_\ell\rvert^2<\infty$.  In particular:
\begin{enumerate}
  \item[(a)] the spin-weighted convolution of Esteves \etal\ \cite{esteves2020spin}, Eq.~(13), with a single input spin $s$ and output spin $s'$ has $c_\ell=\sw{s'}{k}{}^{\ell}_{s}$ in their notation, and their multi-spin convolution sums over input spins and stacks output spins.  Their cross-correlation, Eq.~(15), $g^t_{\ell m}=\sum_i f^i_{\ell m}k^{i*}_{\ell t}$, instead corresponds to $\Psi^\ell_{-t,-i}=(-1)^{i+t}k^i_{\ell t}$;
  \item[(b)] the spin needlets of Geller \& Marinucci \cite{geller2010spin,geller2011mixed} are $\Cop{s}{s}$ with $c_\ell = b\big(\sqrt{e_{\ell s}}/B^{j}\big)$, $e_{\ell s} = (\ell-\lvert s\rvert)(\ell+\lvert s\rvert+1)$, and their mixed needlets are the corresponding $\Cop{s}{0}$ and $\Cop{0}{s}$; mixed needlets connect spin $s$ with spin $0$ only;
  \item[(c)] write $s(x)^{-1}s(y)=R_z(a)R_y(b)R_z(c)$.  In the translated canonical angles $(a,b,-c)$, the radial-with-phases construction of Tessore \etal\ \cite{tessore2025euclid} has kernel $e^{\ii Sa}K(b)e^{\ii sc}$, where $K(b)=\sum_\ell(2\ell+1)K_\ell d^\ell_{Ss}(b)/(4\pi)$, and is $\Cop{s}{S}[\Psi_T]$ with $c_\ell=K_\ell$ by \eqref{eq:app-addition}.  The first angle differs by $\pi$ from the source's printed quadrant convention: its literal angles are $\alpha=a+\pi$, $\alpha'=-c$, giving $c_\ell=(-1)^S K_\ell$ in our conventions with the same radial coefficients;
  \item[(d)] the directional spin convolution of \cite{mcewen2015directional}, Definition~\ref{def:spin-directional}, is not itself of this form, its output being a function on $\SOthree$, but for a steerable kernel, as the directional spin wavelets are, it is by Proposition~\ref{prop:band-limited} below the tuple of the spin-changing convolutions $\Cop{s}{s'}[\lift{s}\psi]$ over $\lvert s'\rvert < N$, with $c_\ell = (-1)^{s'}\sqrt{8\pi^2/(2\ell+1)}\;\sw{s}{\psi}^{*}_{\ell,-s'}$ as in \eqref{eq:harmonic-ss-lifted}: the azimuthal order $-s'$ of the kernel is read to produce the output of spin $s'$, so that the directional spin convolution computes every spin-changing convolution of its kernel at once, and its axisymmetric case, the scaling coefficients of Definition~\ref{def:spin-axisymmetric}, is $(2\pi)^{-1/2}\Cop{s}{0}[\lift{s}\Phi]$, with multiplier $c_\ell=\sqrt{4\pi/(2\ell+1)}\;\sw{s}{\Phi}^{*}_{\ell0}$ as in Proposition~\ref{prop:axisymmetric}.
  \item[(e)] the spin spherical Fourier neural operator layer of the forthcoming article \cite{mulder2026spin}, which uses the same convention as here, learns the multiplier $c_\ell$ directly, as one harmonic weight per degree for each pair of spins $(s,s')$.
\end{enumerate}
\end{proposition}

Proposition~\ref{prop:spin-change} makes precise the sense in which a single layer may map a spin-$s$ field to a spin-$s'$ field without forming a feature map on $\SOthree$: the operator is a multiplier, evaluated with one forward spin-$s$ and one inverse spin-$s'$ spherical harmonic transform, and by Theorem~\ref{thm:classification} the family of such multipliers is the complete family of bounded equivariant linear maps between the fixed input and output spins.  Forming the feature map on $\SOthree$, as the directional spin convolution does, is conversely the way to obtain all of the spin-changing convolutions of a kernel in one operation.

\begin{proposition}[Azimuthally band-limited and steerable kernels]
\label{prop:band-limited}
If the column $m' = -s$ of $\Psi$ is supported on $\lvert n \rvert < N$ then
\begin{equation}
  \Cop{s}{G}[\Psi] = \sum_{\lvert s'\rvert < N} \lift{s'}\, \Cop{s}{s'}[\Psi] ,
\end{equation}
\ie\ the output is represented by at most $2N-1$ spin-changing components $\Cop{s}{s'}[\Psi]$, $\lvert s'\rvert<N$, with $2\min(\ell,N-1)+1$ admissible output slots at degree $\ell$.  Conversely, any finite tuple of bounded equivariant spin-changing maps from the same input spin assembles into a bounded map $\Cop{s}{G}$ with a generally distributional kernel.  A square integrable representative exists exactly when the sum of their Hilbert--Schmidt norms squared is finite.  At finite band-limit the assembled kernel is always square integrable.
\end{proposition}
\begin{proof}
By \eqref{eq:harmonic-sG} the Wigner coefficients of $\Cop{s}{G}[\Psi]f$ vanish for $\lvert n\rvert \ge N$; apply Corollary~\ref{cor:fibre-decomposition} to the output and note that $\proj{s'}\Cop{s}{G}[\Psi] = \Cop{s}{s'}[\Psi]$ by definition.
\end{proof}

Proposition~\ref{prop:band-limited} expresses the finite-spin content behind the steerability result of Wiaux \etal\ \cite{wiaux2006fast}: a kernel with azimuthal band-limit $N$ produces at most $2N-1$ spin components, each obtained by a spin-changing convolution on the sphere.  The directional spin wavelets of \cite{mcewen2015directional} are steerable by construction and their coefficients at each scale form such a tuple.  This also relates the $\mathcal{O}(NL^3)$ group layers of Cobb \etal\ \cite{cobb2021efficient} to multi-spin layers in the sense of Esteves \etal\ \cite{esteves2020spin}, after translating the kernel conventions.  In particular, their cross-correlation is obtained at finite band-limit from projected directional convolutions by replacing their spherical kernels with $k'{}^i_{\ell n}=(-1)^n\sqrt{(2\ell+1)/(8\pi^2)}\,k^i_{\ell,-n}$, rather than using those kernels unchanged.

\subsection{Convolution on the rotation group}
\label{sec:reductions:group}

\begin{proposition}[Convolution on the rotation group]
\label{prop:group}
$\Cop{G}{G}[\Psi]$ is the group convolution of Definition~\ref{def:group-conv}, and by Theorem~\ref{thm:classification}(c) every bounded equivariant map $\Lrot \to \Lrot$ is of this form with a distributional kernel.
\end{proposition}

Proposition~\ref{prop:group} covers the $\SOthree$ layer of Cohen \etal\ \cite{cohen2018spherical}, written there for real signals with the measure normalised to unity and the convolution theorem in the form $\hat f\hat\psi^{\dagger}$, the correlation of Kostelec \& Rockmore \cite{kostelec2008ffts}, the Clebsch--Gordan networks of Kondor \etal\ \cite{kondor2018clebsch}, whose linear maps are the per-degree blocks of \eqref{eq:harmonic-GG} acting across fragments, and the generalised convolution of Cobb \etal\ \cite{cobb2021efficient}, whose channel mixing is the multi-channel version of \eqref{eq:harmonic-GG}.  A lifted spin-$s$ input reduces $\Cop{G}{G}$ to $\Cop{s}{G}$; it gives the traditional directional convolution of Proposition~\ref{prop:directional} after replacing the kernel by the lifted spherical kernel with the same relevant input column.  At infinite band-limit this replacement may require a distributional spherical kernel.  For an input supported on $\lvert m'\rvert<N_{\rm in}$ and an effective kernel supported on $\lvert n\rvert<N_{\rm out}$ in those input columns, Corollary~\ref{cor:fibre-decomposition} gives an array with at most $(2N_{\rm out}-1)\times(2N_{\rm in}-1)$ spin-changing entries.  Setting $N_{\rm in}=N_{\rm out}=N$ gives the square array used in Sec.~\ref{sec:computation}.

\subsection{Spin-changing operators of mathematical physics}
\label{sec:reductions:physics}

The rotation-equivariant operators considered here are harmonic multipliers, with generally distributional kernels.  Bounded examples fall under Theorem~\ref{thm:classification}; differential examples use the domain \eqref{eq:multiplier-domain} instead.

\begin{proposition}[Spin-changing operators of mathematical physics]
\label{prop:physics}
The following classical operators have the multiplier form \eqref{eq:classification}, with \eqref{eq:multiplier-domain} understood for unbounded examples:
\begin{enumerate}
  \item[(a)] the spin raising operator $\eth : \mathcal{D}(\eth)\to\Lsph{s+1}$, with $c_\ell=\sqrt{(\ell-s)(\ell+s+1)}$, and the spin lowering operator $\bar\eth : \mathcal{D}(\bar\eth)\to\Lsph{s-1}$, with $c_\ell=-\sqrt{(\ell+s)(\ell-s+1)}$, have the first-order spin Sobolev domains $\mathcal{D}(\eth),\mathcal{D}(\bar\eth)\subset\Lsph{s}$ specified by \eqref{eq:multiplier-domain} \cite{newman1966note,goldberg1967spin};
  \item[(b)] the $E$ and $B$ components of a spin-$\pm 2$ polarisation field, $a_{E,\ell m} = -\tfrac12(a_{2,\ell m}+a_{-2,\ell m})$ and $a_{B,\ell m} = \tfrac{\ii}{2}(a_{2,\ell m}-a_{-2,\ell m})$ \cite{zaldarriaga1997all,kamionkowski1997statistics}, each the sum of a $\Cop{2}{0}$ and a $\Cop{-2}{0}$ operator with constant $c_\ell$; the scalars $\bar\eth^2(Q+\ii U)$ and $\eth^2(Q-\ii U)$ of \cite{zaldarriaga1997all} are the constituent $\Cop{2}{0}$ and $\Cop{-2}{0}$ maps with $c_\ell = \sqrt{(\ell+2)!/(\ell-2)!}$;
  \item[(c)] the spherical Kaiser--Squires relation of weak gravitational lensing between convergence and shear \cite{wallis2022mapping}: $\Cop{0}{2}$ with $c_\ell = \mathcal{D}_\ell = -\sqrt{(\ell+2)(\ell-1)/(\ell(\ell+1))}$ and its inverse $\Cop{2}{0}$ with $c_\ell = \mathcal{D}_\ell^{-1}$, for $\ell \ge 2$; the real-space forward map of \cite{wallis2022mapping} is synthesis with an axisymmetric spin-$2$ kernel, the adjoint of spin-$2$ to scalar analysis, rather than that forward analysis itself.  Its two-point kernel is \eqref{eq:nogo-kernel} with $(s,s')=(0,2)$;
  \item[(d)] the all-Stokes beam response of Challinor \etal\ \cite{challinor2000all}, \ie\ their Eq.~(39), corresponds to a sum of $\Cop{0}{G}$ maps on intensity $I$ and circular polarisation $V$, and $\Cop{\pm2}{G}$ maps on linear polarisation.  For an axisymmetric co-polar beam, the intensity kernel has azimuthal order $0$ and the linear-polarisation kernels have the prescribed orders $\mp2$ for output spins $\pm2$.  Their Eqs.~(47)--(50) give scalar and spin-$2$ smoothing multipliers $W_\ell$ and $\sw{2}{W}_\ell$; the response in Eq.~(46) depends on the third Euler angle through $\{1,\cos2\psi,\sin2\psi\}$, corresponding to the retained spins $0,\pm2$ and band-limit $N=3$ in Proposition~\ref{prop:band-limited}.
\end{enumerate}
\end{proposition}

The identity, with $c_\ell\equiv1$, illustrates the distributional-kernel clause of Theorem~\ref{thm:classification}.  Spin raising instead has an unbounded polynomially growing multiplier and uses the separate domain \eqref{eq:multiplier-domain}; the second-order differential maps in part (b) have second-order spin Sobolev domains.  None of these operators has a square integrable canonical kernel.  Finally, the discrete--continuous convolution of Ocampo \etal\ \cite{ocampo2023scalable} replaces the integral in Definition~\ref{def:directional} by a quadrature on a sampling of the sphere; it is Proposition~\ref{prop:directional} up to quadrature error and inherits its equivariance to that accuracy, while its restricted form, with output on the sphere, is a section evaluation, to which we now turn.

\section{Beyond the framework}
\label{sec:beyond}

The families of Sec.~\ref{sec:review:anisotropic} are not fully $\SOthree$-equivariant in general, although they contain equivariant special cases.  Their non-equivariant members are not generalised lifted spherical convolutions, by Theorem~\ref{thm:harmonic}, but are described by the same tools: they either evaluate a lifted convolution on a section without the constraint of Corollary~\ref{cor:constraint}, or use an $m$-dependent diagonal harmonic multiplier.  We characterise full equivariance and establish guaranteed residual symmetries, without classifying the exact stabiliser of every kernel.

\subsection{Evaluation on a section}
\label{sec:beyond:section}

\begin{proposition}[Evaluation on a section]
\label{prop:section-evaluation}
Let $\sigma(x)=s(x)R_z(\chi_\sigma(x))$ be a measurable section and let $\Psi\in\Lrot$ satisfy $\Psi^\ell_{n,-s}=0$ for $\lvert n\rvert\ge N$.  Define $E_\sigma[\Psi]\,\sw{s}{f}=(\Psi\conv\lift{s}\,\sw{s}{f})\circ\sigma$, using the continuous representative of the $L^2*L^2$ correlation.  Then
\begin{equation}
  E_\sigma[\Psi]\, \sw{s}{f} = \sum_{\lvert s'\rvert < N} \frac{e^{-\ii s'\chi_\sigma}}{\sqrt{2\pi}}\; \Cop{s}{s'}[\Psi]\, \sw{s}{f} ,
  \label{eq:section-evaluation}
\end{equation}
where $N$ bounds the first index within the effective input column.  With the fixed coordinate-frame output action, the zero map is equivariant into every spin space for every section.  Otherwise $E_\sigma[\Psi]$ is equivariant into spin-$t$ fields if and only if $\Psi^\ell_{n,-s}=0$ for $n\ne-t$ and $e^{-\ii t\chi_\sigma}$ is constant almost everywhere.  The latter condition is automatic for $t=0$.

For $\sigma_0=s(\cdot)$, all polar rotations are symmetries.  For $\chi_\sigma=\pi-\varphi$, a polar rotation $R_z(a)$ is a symmetry exactly when $e^{\ii ta}=1$ for every active output spin $t$, where active means $\Psi^\ell_{-t,-s}\ne0$ for some $\ell$.  This criterion does not classify non-polar symmetries.
\end{proposition}
\begin{proof}
Corollary~\ref{cor:fibre-decomposition} and Definition~\ref{def:lift} give the finite sum, which is bounded since each component is bounded and the phase has unit modulus.  If it is equivariant into spin $t$, Theorem~\ref{thm:classification} applies.  Testing each input harmonic and comparing the independent Wigner columns representing the different output spins forces all components except $t$ to vanish and, for a nonzero map, forces its phase to be constant almost everywhere.  The converse follows from equivariance of $\Cop{s}{t}$.  Polar rotations act by coordinate rotation for every spin; substituting $\chi_\sigma=\pi-\varphi$ gives the stated phase criterion.  The degreewise argument and the zero-map case are detailed in Appendix~\ref{app:proofs:beyond}.
\end{proof}

When only spin $t$ is present, $E_\sigma=e^{-\ii t\chi_\sigma}E_{\sigma_0}$ already expresses the output in the $\sigma$ frame.  Returning it to the coordinate frame multiplies by $e^{+\ii t\chi_\sigma}$; reading it in the $\sigma$ frame does not remove the phase.

The standard correlation of Definition~\ref{def:standard-correlation} is $E_{\sigma_0}[\lift{0}\psi]$.  The restricted discrete--continuous convolution of Ocampo \etal\ \cite{ocampo2023scalable} has this form up to quadrature.  Wiaux \etal\ \cite{wiaux2006fast} expand the standard correlation as $\sum_n\sum_{\ell m}(-1)^n\sqrt{4\pi/(2\ell+1)}\,\psi^{*}_{\ell,-n}f_{\ell m}\,\sw{n}{Y}_{\ell m}$, which is \eqref{eq:section-evaluation} with \eqref{eq:harmonic-ss-lifted}, constants included: it sums the spin components in the coordinate frame.  Polar equivariance is guaranteed, but full rotation equivariance is not present in general.  Ocampo \etal\ derive the limiting rotation behaviour as $\beta\to0$, corresponding to the vanishing twist for polar rotations, away from the coordinate singularities.

The commutative anisotropic convolution of Definition~\ref{def:commutative-anisotropic} uses $\chi_\sigma=\pi-\varphi$ and a bilinear pairing.  It is commutative \cite{sadeghi2012commutative} but not fully rotation-equivariant in general.  Its polar symmetries depend on the active orders: even orders retain a half-turn, and orders divisible by $q$ retain the cyclic rotations through multiples of $2\pi/q$.  A zonal kernel retains full scalar equivariance, while non-polar symmetries of other kernels require separate analysis.  Proposition~\ref{prop:section-evaluation} gives the precise obstruction to combining nonzero-spin components into one equivariant scalar output.

\subsection{Operators diagonal in the spherical harmonic basis}
\label{sec:beyond:diagonal}

\begin{proposition}[Harmonic multiplication and sifting]
\label{prop:diagonal}
Let $T$ act on $L^2(\sphere)$ by $(Tf)_{\ell m}=g_{\ell m}f_{\ell m}$ with a bounded coefficient sequence $g$.  This is harmonic multiplication, or sifting with the coefficientwise-conjugated filter; the filter may be a distribution.  Then $T$ commutes with every polar rotation, and is $\SOthree$-equivariant if and only if $g_{\ell m}$ is independent of $m$ for each $\ell$.  In that case $T=\Cop{0}{0}[\Psi_T]$ with $c_\ell=g_{\ell0}$ in the Fourier-multiplier sense of Theorem~\ref{thm:classification}, with a generally distributional zonal representative.  If $g_{\ell m}=g_{\ell,-m}$, $T$ also commutes with $R_y(\pi)$, and thus with the embedded $O(2)$ normaliser of polar rotations.
\end{proposition}
\begin{proof}
The full-equivariance criterion is Theorem~\ref{thm:classification} with $s=s'=0$.  Polar rotations act diagonally by $e^{-\ii ma}$ and commute with any diagonal multiplier, while $R_y(\pi)$ exchanges $m$ and $-m$ up to a sign, so commutes when the corresponding multipliers agree.
\end{proof}

Proposition~\ref{prop:diagonal} is the characterisation of isotropic operators of Kennedy \etal\ \cite{kennedy2011azimuthally}.  An $L^2$ sifting filter gives a Hilbert--Schmidt diagonal operator, with norm squared $\sum_{\ell m}\lvert g_{\ell m}\rvert^2$; general bounded diagonal operators allow distributional filters.  The sifting filter and the equivalent zonal convolution kernel are different representations of the operator.

\subsection{The desiderata revisited}
\label{sec:beyond:desiderata}

Roddy \& McEwen \cite{roddy2021sifting} ask for directional kernels, spherical output and the $\mathcal{O}(L^3)$ cost of spherical harmonic transforms.  Their comparison includes isotropic, left, directional, commutative anisotropic and sifting constructions, and motivates sifting by the limitations of the preceding constructions considered there.  Sifting meets the three criteria, as does coordinate-section standard correlation, whose $\mathcal{O}(L^3)$ grid evaluation is derived in Appendix~\ref{app:proofs:standard-correlation}.  The cited commutative anisotropic algorithm costs $\mathcal{O}(L^3\log L)$ \cite{sadeghi2012commutative} and does not meet the strict third criterion.

Full rotation equivariance adds a different restriction: from one fixed input spin to one fixed output spin, Corollary~\ref{cor:nogo} permits only a radial kernel with the prescribed frame phases, not independent directionality.  Standard correlation and sifting are therefore not fully equivariant in general, although each has equivariant special cases and guaranteed polar symmetry.  A tuple of output spins retains directional information equivariantly, at an $\mathcal{O}(NL^3)$ upper bound for at most $2N-1$ components.

\section{Computation}
\label{sec:computation}

In this section we describe how the operators of the article are computed exactly and efficiently for band-limited signals.  We first present harmonic-space algorithms for the four generalised lifted spherical convolutions and their computational complexity, which follow from the part of the kernel that each operator reads and from the representation retained at its output.  We then describe \texttt{s2conv}, the open-source code in which the operators are implemented.  We finally present numerical experiments that verify the reductions of Sec.~\ref{sec:reductions} and the equivariance of every operator to machine precision, and that confirm the predicted scaling of the computation time with the band-limit.

\subsection{Harmonic-space algorithms and their complexity}
\label{sec:computation:algorithms}

The lifted input need not be formed explicitly: for band-limited signals every operator of Table~\ref{tab:taxonomy} is computed in harmonic space by spin spherical harmonic and Wigner transforms, by appealing to the sampling theorems on the sphere \cite{mcewen2011novel} and on the rotation group \cite{mcewen2015novel}, which are exact for band-limited signals.  Let $L$ be the band-limit and $1\le N\le L$ an azimuthal band-limit.  We assume that the Wigner coefficients of the kernel are available.  We also consider the cost of the transforms of the sampled input and output.  Working memory is given by the arrays of harmonic coefficients held during the computation, with the Wigner $d$-functions of the transforms generated on the fly; the kernel parameters are counted separately, and fully precomputed Wigner $d$-functions for a transform of fixed spin require $\mathcal{O}(L^3)$ further storage \cite{price2024differentiable}.  The four operators of Definition~\ref{def:lifted-conv} are computed as follows.
\begin{enumerate}
  \item[$\Cop{s}{s'}$:] a forward spin-$s$ spherical harmonic transform, multiplication of the coefficients by $c_\ell=(-1)^{s+s'}\Psi^{\ell *}_{-s',-s}$, and an inverse spin-$s'$ transform, at a cost of $\mathcal{O}(L^3)$ \cite{mcewen2011novel}, with $\mathcal{O}(L^2)$ working memory and at most $L$ kernel parameters.
  \item[$\Cop{s}{G}$:] a forward spin-$s$ transform, multiplication by $\Psi^{\ell *}_{n,-s}$ for $\lvert n\rvert<N$, and an inverse Wigner transform, which comprises at most $2N-1$ inverse spin transforms and a Fourier series in $\gamma$ \cite{mcewen2015novel,price2024differentiable}, at a cost of $\mathcal{O}(NL^3)$, with $\mathcal{O}(NL^2)$ storage for the output and $\mathcal{O}(NL)$ kernel parameters; $N=L$ allows an unrestricted column.
  \item[$\Cop{G}{s'}$:] for an input supported on $\lvert m'\rvert<N$, a forward Wigner transform, contraction with the entries $\Psi^{\ell *}_{-s',m'}$ of the row $-s'$, and an inverse spin-$s'$ transform, at a cost of $\mathcal{O}(NL^3)$, with $\mathcal{O}(NL^2)$ storage for the input and $\mathcal{O}(NL)$ kernel parameters.
  \item[$\Cop{G}{G}$:] a forward Wigner transform, a matrix product for each degree, and an inverse Wigner transform.  With input azimuthal band-limit $N_{\rm in}$ and kernel row bound $N_{\rm out}$ as in Sec.~\ref{sec:reductions:group}, the transforms cost $\mathcal{O}((N_{\rm in}+N_{\rm out})L^3)$ and the products $\mathcal{O}(N_{\rm in}N_{\rm out}L^2)$, with $\mathcal{O}(N_{\rm in}N_{\rm out}L)$ kernel parameters and $\mathcal{O}((N_{\rm in}+N_{\rm out})L^2)$ storage; taking both bounds equal to $N$ gives a cost of $\mathcal{O}(NL^3+N^2L^2)$, and $N=L$ the unrestricted $\mathcal{O}(L^4)$.
\end{enumerate}
These costs are collected in Table~\ref{tab:complexity}, together with those of the constructions of Sec.~\ref{sec:beyond}, which are formed from the same transforms: the evaluation on a general section is computed component by component in $\mathcal{O}(NL^3)$, each spin component being accumulated into the single $\mathcal{O}(L^2)$ output array as it is computed, so that the tuple of $\mathcal{O}(NL^2)$ coefficients need not be held at once, while the standard correlation on grid values, harmonic multiplication and the sifting convolution cost $\mathcal{O}(L^3)$ and the commutative anisotropic convolution $\mathcal{O}(L^3\log L)$ \cite{sadeghi2012commutative}.  Note that a single spin-changing convolution costs $\mathcal{O}(L^3)$, as does any finite tuple of spins, and that further structure of the kernel may permit other algorithms.  Furthermore, by Proposition~\ref{prop:composition} consecutive operators require no intermediate transforms: a chain of sphere-to-sphere convolutions is a single multiplier, and a tuple of spins is stored as $2N-1$ arrays of harmonic coefficients between which no transform is required.  All of the harmonic transforms required are available with exact quadrature on the sphere \cite{mcewen2011novel} and on the rotation group \cite{mcewen2015novel}, and are implemented differentiably on hardware accelerators, with precomputed or on-the-fly Wigner $d$-functions, in the \texttt{s2fft} code \cite{price2024differentiable} (\url{https://github.com/astro-informatics/s2fft}).

\begin{table}[t]
  \caption{Algorithmic complexity for band-limit $L$ and azimuthal band-limit $N\le L$: the scaling of the computation time, the working memory and the number of kernel parameters of each operator, with the part of the kernel that it reads.  Kernel coefficients are supplied and the transforms are those of \cite{mcewen2011novel,mcewen2015novel}, implemented in \cite{price2024differentiable}.  Memory is the working memory of Sec.~\ref{sec:computation:algorithms}, the arrays of input and output coefficients held during the computation with Wigner $d$-functions generated on the fly; parameter storage is listed separately, and fully precomputed Wigner $d$-functions increase the memory.  Counts per degree are further limited by $2\ell+1$, and ``--'' means not applicable.}
  \label{tab:complexity}
  \scriptsize\setlength{\tabcolsep}{3pt}
  \centering
  \resizebox{\textwidth}{!}{%
  \begin{tabular}{lllll}
    \toprule
    Operator & Effective kernel support & Compute & Memory & Parameters \\
    \midrule
    $\Cop{s}{s'}$ & one entry per $\ell$ & $L^3$ & $L^2$ & $L$ \\
    $\Cop{s}{G}$, $\lvert n\rvert<N$ & at most $2N-1$ entries per $\ell$ & $NL^3$ & $NL^2$ & $NL$ \\
    $\Cop{s}{G}$, $N=L$ & one column per $\ell$ & $L^4$ & $L^3$ & $L^2$ \\
    $\Cop{G}{s'}$, input $\lvert m'\rvert<N$ & row $-s'$, $\lvert m'\rvert<N$ & $NL^3$ & $NL^2$ & $NL$ \\
    $\Cop{G}{G}$, input $\lvert m'\rvert<N$ & $\lvert n\rvert,\lvert m'\rvert<N$ & $NL^3+N^2L^2$ & $NL^2$ & $N^2L$ \\
    $\Cop{G}{G}$, $N=L$ & full block per $\ell$ & $L^4$ & $L^3$ & $L^3$ \\
    General-section evaluation, accumulated & col.\ $-s$, $\lvert n\rvert<N$ & $NL^3$ & $L^2$ & $NL$ \\
    Standard correlation, grid values & one column per $\ell$ & $L^3$ & $L^2$ & $L^2$ \\
    Commutative anisotropic \cite{sadeghi2012commutative} & one column per $\ell$ & $L^3\log L$ & -- & $L^2$ \\
    Harmonic multiplication / sifting & diagonal in $(\ell,m)$ & $L^3$ & $L^2$ & $L^2$ \\
    \bottomrule
  \end{tabular}}
\end{table}

\subsection{Implementation: the \texttt{s2conv} code}
\label{sec:computation:software}

We have implemented the algorithms described above in \texttt{s2conv}, an open-source code that accompanies this article and is publicly available at \url{https://github.com/astro-informatics/s2conv} under the MIT licence.  The code is written in \texttt{JAX} and relies on the \texttt{s2fft} code \cite{price2024differentiable} to compute spin spherical harmonic and Wigner transforms, so that every operator is differentiable, may be compiled, runs on hardware accelerators (\eg\ GPUs, TPUs) and is exact for band-limited signals.  Since the operators are differentiable, their kernel coefficients may also be learned, as in the forthcoming article \cite{mulder2026spin}.  The code follows the article closely.  A single function implements the four operators of Definition~\ref{def:lifted-conv} for a kernel given by its Wigner coefficients, in the harmonic forms of Theorem~\ref{thm:harmonic}, reading only the entry, column, row or block of the kernel that the operator uses, so that the complexity of Table~\ref{tab:complexity} is realised.  The lift, the adjoint projection and the decomposition of a function on the rotation group into its spin components (Corollary~\ref{cor:fibre-decomposition}) are also provided, together with the kernels of Sec.~\ref{sec:reductions}, namely the kernel lifted from a spin field, the kernel $\Psi_h$ of the left convolution, the canonical kernel of a multiplier (Theorem~\ref{thm:classification}) and the flipped kernel $\widetilde\Psi$ of an adjoint.  Furthermore, the classical convolutions of Sec.~\ref{sec:review} are provided in the closed forms of Sec.~\ref{sec:reductions}, the operators of mathematical physics of Proposition~\ref{prop:physics} as multipliers, and the constructions of Sec.~\ref{sec:beyond}, namely the standard correlation and the commutative anisotropic convolution as evaluations on a section, and harmonic multiplication and the sifting convolution as diagonal operators.  The harmonic-space operators act on arrays of coefficients in the layout of \texttt{s2fft} and broadcast over leading axes, so that batches, channels or spins may be carried as leading dimensions or mapped over.  The transforms may be computed with Wigner $d$-functions generated on the fly, with the $\mathcal{O}(L^2)$ working memory of Table~\ref{tab:complexity}, or precomputed and stored, which trades $\mathcal{O}(L^3)$ memory for speed, as \texttt{s2fft} allows.

\subsection{Numerical validation}
\label{sec:computation:validation}

In order to verify the reductions of Sec.~\ref{sec:reductions} and the equivariance of the operators, and to assess their computation time, we perform the following numerical experiments with the benchmark scripts distributed with the code.  We measure the rotational equivariance of an operator $T$ by the protocol of Cobb \etal\ \cite{cobb2021efficient} and Ocampo \etal\ \cite{ocampo2023scalable}.  We simulate band-limited test signals $f_i$ whose harmonic or Wigner coefficients are drawn from the standard normal distribution, draw rotations $g_j$ from the Haar measure of $\SOthree$, and compute the relative errors
\begin{equation}
  \epsilon_{ij} = \frac{\lVert T\,\pi(g_j)f_i - \pi'(g_j)\,Tf_i\rVert}{\lVert \pi'(g_j)\,Tf_i\rVert},
  \label{eq:equivariance-error}
\end{equation}
where $\pi$ and $\pi'$ are the representations acting on the input and on the output, $\pis{s}$ on a spin-$s$ field and $\piG$ on a function on the rotation group, and where the norms are those of $\Lsph{s}$ or $\Lrot$, which are computed by Parseval's identity and are exact for band-limited outputs.  The numerator is the equivariance error, the difference between rotating the input and rotating the output, and the denominator is the norm of the output, which a rotation does not change.  For an evaluation on a section, whose output is not band-limited, the comparison is made through the spin components of Proposition~\ref{prop:section-evaluation}, which is exact and considerably faster than a comparison point by point.  The errors are averaged over 100 random signals and 100 random rotations for each operator.

For each operator we also compute a reduction error, the largest relative difference, over 100 random signals, between its closed form and its form as a generalised lifted spherical convolution, \ie\ Definition~\ref{def:lifted-conv} composed literally from the lift, the group convolution in Wigner space and the adjoint projection.  This verifies the reductions of Sec.~\ref{sec:reductions} together with their constants, for example the $\sqrt{4\pi/(2\ell+1)}$ of the axisymmetric convolution and the $2\pi$ of the left convolution, and, for the multipliers of Proposition~\ref{prop:physics}, the canonical kernel of Theorem~\ref{thm:classification}.

The operators tested are the four operators of Definition~\ref{def:lifted-conv} with random kernels, for the spins $0 \to 0$, $2 \to 1$ and $-1 \to 2$ between spin fields, $0$ and $2$ to the rotation group, the rotation group to spins $0$ and $2$, and the rotation group to itself, with functions on the rotation group azimuthally band-limited at $N = 5$; the classical convolutions of Sec.~\ref{sec:review}, namely the axisymmetric convolution, the left convolution, the directional convolution of a spin-$2$ field with a steerable kernel and the synthesis; the operators of mathematical physics, namely $\eth$ from spin $1$ to spin $2$, $\bar\eth$ from spin $1$ to spin $0$, the $E$ mode of a spin-$2$ field and the Kaiser--Squires relation; and, for contrast, the standard correlation, the commutative anisotropic convolution, the sifting convolution and harmonic multiplication, which Sec.~\ref{sec:beyond} shows are not equivariant under general rotations.

\begin{table}[t]
  \caption{Reduction errors and equivariance errors of the operators of \texttt{s2conv} at band-limit $L = 128$, in double precision, from the benchmark script of the code.  The equivariance error is the mean of \eqref{eq:equivariance-error} over 100 random signals and 100 random rotations, and the reduction error the largest relative difference between the closed form of the operator and Definition~\ref{def:lifted-conv} composed from the lift, the group convolution and the projection, over 100 random signals (``--'' where the operator is the definition itself or has no lifted form).  The random kernels of the lifted operators, the directional kernels and the functions on the rotation group are azimuthally band-limited at $N = 5$.  The rule separates the equivariant operators from the constructions of Sec.~\ref{sec:beyond}.}
  \label{tab:equivariance}
  \small
  \centering
\begin{tabular}{lcc}
\toprule
Operator & Reduction error & Equivariance error \\
\midrule
$\mathcal{C}_{0\to0}[\Psi]$ & $1.4\times10^{-16}$ & $9.2\times10^{-16}$ \\
$\mathcal{C}_{2\to1}[\Psi]$ & $1.4\times10^{-16}$ & $9.0\times10^{-16}$ \\
$\mathcal{C}_{-1\to2}[\Psi]$ & $1.4\times10^{-16}$ & $8.7\times10^{-16}$ \\
$\mathcal{C}_{0\to G}[\Psi]$ & $1.3\times10^{-16}$ & $1.2\times10^{-15}$ \\
$\mathcal{C}_{2\to G}[\Psi]$ & $1.3\times10^{-16}$ & $1.2\times10^{-15}$ \\
$\mathcal{C}_{G\to0}[\Psi]$ & $2.0\times10^{-16}$ & $1.3\times10^{-15}$ \\
$\mathcal{C}_{G\to2}[\Psi]$ & $2.1\times10^{-16}$ & $1.3\times10^{-15}$ \\
$\mathcal{C}_{G\to G}[\Psi]$ & $0$ & $1.6\times10^{-15}$ \\
axisymmetric convolution & $1.7\times10^{-16}$ & $8.3\times10^{-16}$ \\
left convolution & $0$ & $7.6\times10^{-16}$ \\
directional convolution & $1.5\times10^{-16}$ & $1.1\times10^{-15}$ \\
directional synthesis & $1.6\times10^{-16}$ & $1.2\times10^{-15}$ \\
$\eth$ & $0$ & $9.6\times10^{-16}$ \\
$\bar\eth$ & $0$ & $9.6\times10^{-16}$ \\
$E$ mode from $Q + \ii U$ & $0$ & $9.7\times10^{-16}$ \\
Kaiser--Squires & $0$ & $8.8\times10^{-16}$ \\
\midrule
standard correlation & $3.8\times10^{-16}$ & $1.3\times10^{0}$ \\
commutative anisotropic convolution & -- & $1.3\times10^{0}$ \\
sifting convolution & -- & $1.4\times10^{0}$ \\
harmonic multiplication & -- & $1.4\times10^{0}$ \\
\bottomrule
\end{tabular}

\end{table}

Results at $L = 128$ are presented in Table~\ref{tab:equivariance}; the results at $L = 32$ and $64$ are distributed with the code.  All numerical experiments are performed in double precision on an NVIDIA GB10 graphics processing unit.  The equivariance error of every equivariant operator is of order $10^{-15}$, \ie\ close to machine precision.  The reduction errors are of order $10^{-16}$, or vanish identically where the closed form and the lifted form read the same coefficients, which verifies the reductions of Sec.~\ref{sec:reductions} with their constants.  In contrast, the standard correlation, the commutative anisotropic convolution, the sifting convolution and harmonic multiplication exhibit equivariance errors of order unity, as expected, since Sec.~\ref{sec:beyond} shows that these constructions are equivariant only under rotations about the polar axis.

In order to assess the scaling of computation time with band-limit, we time one operator of each of the four kinds, $\Cop{2}{1}$ and, with $N = 5$, $\Cop{2}{G}$, $\Cop{G}{2}$ and $\Cop{G}{G}$, with random kernels, for band-limits from $L = 16$ to $L = 256$.  Each operator is run end to end, including its transforms, after compilation, and the computation time is measured by the median of five runs in double precision.  The timings are performed on a single processor core, one Cortex-X925 core of the GB10 with the multithreading of the compiler disabled, so that they count operations rather than the occupancy of a graphics processing unit.  Results are plotted in Fig.~\ref{fig:scaling}.  With $N$ fixed, the complexity of each of the four operators is the $\mathcal{O}(L^3)$ of its transforms (Table~\ref{tab:complexity}), and, as expected, the measured computation time is indeed found to scale as $\mathcal{O}(L^3)$.  The operators on the rotation group are more costly than the operator between spin fields since each of their Wigner transforms runs over $2N-1 = 9$ azimuthal orders.

\begin{figure}[t]
  \centering
  \includegraphics{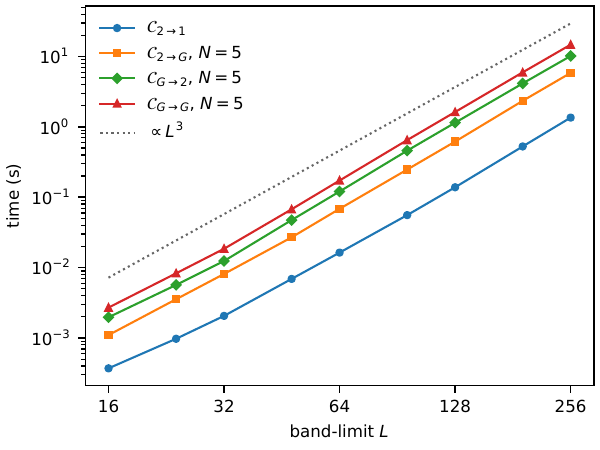}
  \caption{Computation time of the four generalised lifted spherical convolutions against band-limit, for $\Cop{2}{1}$ and, with $N = 5$, $\Cop{2}{G}$, $\Cop{G}{2}$ and $\Cop{G}{G}$, with random kernels, each run end to end with its transforms and measured by the median of five compiled runs in double precision on a single Cortex-X925 core of an NVIDIA GB10 with multithreading disabled.  As expected, computation time scales as $\mathcal{O}(L^3)$, the complexity of the transforms (Table~\ref{tab:complexity}), which is shown by the dotted line.}
  \label{fig:scaling}
\end{figure}

\section{Physical illustrations}
\label{sec:illustrations}

In this section we apply the generalised lifted spherical convolution to fields measured on the Earth and on the sky, in order to illustrate in practice the operators that it provides, which map between fields of different spin exactly and efficiently and which respect the geometry of each field.  Many of the fields of geophysics and cosmology are not scalars.  Wind, an ocean current, the deflection of the vertical and the horizontal magnetic field are tangent vector fields, of spin $\pm1$, while the polarisation of the cosmic microwave background, the shear of weak gravitational lensing and the horizontal part of a gravity-gradient tensor are fields of spin $\pm2$, each describing a headless direction at every point.  The components of such a field are defined only with respect to the local frame $(e_\theta, e_\varphi)$ of Sec.~\ref{sec:prelim:spin}, which turns from point to point and is undefined at the poles of the coordinates.  Treating the components as independent scalar fields, for example by smoothing $v_\theta$ and $v_\varphi$, or $Q$ and $U$, separately, combines values expressed in different frames, and the result then depends on the choice of coordinates and is incorrect near their poles.  In each of the examples we adopt the lifted convolutions whose input and output spins are those of the physical quantities involved.  A tangent vector field with northward and eastward components $(v,u)$ is the spin-$1$ field $v_\theta + \ii v_\varphi = -v + \ii u$, since $e_\theta$ points south, and a polarisation with Stokes parameters $(Q,U)$ is the spin-$2$ field $Q + \ii U$.  The fields are truncated at the band-limit $L = 64$ for these illustrations, and the fields of the Earth are multiplied by the taper $e^{-\ell(\ell+1)/L^2}$, which limits the ringing of the truncation at coastlines and jets.  The examples are computed with \texttt{s2conv}, in a notebook distributed with the code, from which the figures of this section are taken.

We restrict our attention to the lifted convolution itself, which, since its kernels may be distributions, also includes differential operators such as $\eth$ and $\bar\eth$ and their compositions.  Note that convolution is often one component of other analyses, such as wavelet transforms on the sphere, convolutional neural networks and Fourier neural operators, which build on the operators illustrated here but which we do not consider further.

The fields considered are shown in Fig.~\ref{fig:fields-overview} of Sec.~\ref{sec:prelim:spin}.  Throughout this section, maps of the whole sphere use the Mollweide projection, centred on the meridian of longitude zero with north at the top, the fields of the Earth in geographic coordinates and the Planck maps in Galactic coordinates, and views of one hemisphere use the orthographic projection, centred on the point given in each caption.  A signed quantity is shown as colour on a diverging scale, red for positive and blue for negative values and white for zero, and a non-negative quantity, such as the modulus of a spin field, on a sequential scale from light to dark.  Spin fields are drawn in the tangent plane with the glyphs of Fig.~\ref{fig:spin-fields}: a spin-$1$ field $\sw{1}{f}$ as an arrow at the angle $\arg(\sw{1}{f})$ from $e_\theta$ towards $e_\varphi$, and a spin-$2$ field $\sw{2}{f}$ as a headless stick at the angle $\arg(\sw{2}{f})/2$, at quasi-uniform points, with lengths that grow with the modulus of the field.  Coastlines are drawn as black lines.

\subsection{Topography: slopes and ridges}
\label{sec:illustrations:topography}

The elevation of the land and the depth of the ocean floor, the topography and bathymetry $f$ of ETOPO 2022 \cite{noaa2022etopo}, form a scalar field, but the two quantities that describe the shape of the relief are not scalars: the slope is a vector and the orientation of a mountain range is a headless direction.  We obtain the slope by the lifted convolution $\Cop{0}{1}$ whose kernel is the multiplier $\sqrt{\ell(\ell+1)}$ of the spin raising operator, Proposition~\ref{prop:physics}(a).  Since $\eth f = -(\partial_\theta + \ii\sin^{-1}\theta\,\partial_\varphi)f$ for a scalar field, by \eqref{eq:eth-def}, the spin-$1$ field $\eth f/R_\oplus$, where $R_\oplus$ is the radius of the Earth, is the vector $-\nabla f$, which points downhill and whose modulus is the slope.

To locate mountain ranges we consider a steerable ridge detector $\psi$, a kernel elongated along one direction with a zero-mean profile across it, like a second derivative, truncated to the azimuthal orders $\lvert m\rvert < 5$.  Its directional convolution with the field, $f \circledast \psi$, is a function on the rotation group that gives the response of the detector at every position and every orientation $\gamma$.  Since the kernel is steerable, this function is the tuple of the nine spin fields $\sw{s'}{g} = \Cop{0}{s'}[\lift{0}\psi]f$, $\lvert s'\rvert < 5$, by Proposition~\ref{prop:band-limited}, from which the response at any orientation follows by Proposition~\ref{prop:section-evaluation} as $\sum_{s'}(2\pi)^{-1/2}e^{-\ii s'\gamma}\,\sw{s'}{g}$.  The largest response over $\gamma$ measures the strength of a ridge, and the spin-$2$ component $\sw{2}{g}$, the part of the response that varies with orientation as $\cos 2\gamma$, like that of a line, is itself the orientation of the ridges as a spin-$2$ field: the response is largest when the detector lies at the angle $\arg(\sw{2}{g})/2$, along the ridge.  No sampling of orientations is required, and since both outputs are rotation-equivariant, ranges are detected alike whatever their orientation and position, near the equator or near the poles.  We apply the detector to the land elevation $f_+ = \max(f,0)$, since the continental margins, a step of several kilometres, would otherwise dominate its response.  Figure~\ref{fig:topography} shows the results: the arrows run down the flanks of the Tibetan Plateau, the ridge strength picks out the Himalaya, the Andes, the Rocky Mountains, the highlands of East Africa and the ice sheets of Greenland and Antarctica, and where the ridge strength is large the sticks lie along the ranges, from east to west along the Himalaya and from north to south along the Andes.

\begin{figure}[tp]
  \centering
  \begin{subfigure}[t]{0.32\textwidth}
    \centering
    \includegraphics[width=\linewidth]{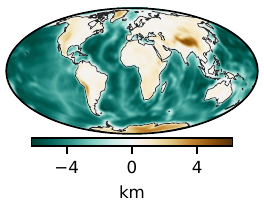}
    \caption{Topography and bathymetry $f$, in km.}
  \end{subfigure}\hfill
  \begin{subfigure}[t]{0.32\textwidth}
    \centering
    \includegraphics[width=\linewidth]{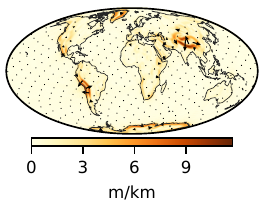}
    \caption{Slope $\lvert\eth f_+\rvert/R_\oplus$, in m\,km$^{-1}$, with $\eth f_+$ drawn as arrows pointing downhill.}
  \end{subfigure}\hfill
  \begin{subfigure}[t]{0.32\textwidth}
    \centering
    \includegraphics[width=\linewidth]{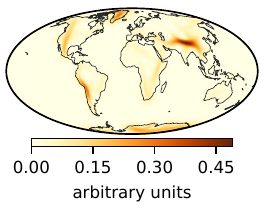}
    \caption{Ridge strength $\max_\gamma (f_+ \circledast \psi)(\cdot,\gamma)$.}
  \end{subfigure}\\[1ex]
  \begin{subfigure}[t]{0.32\textwidth}
    \centering
    \includegraphics[width=\linewidth]{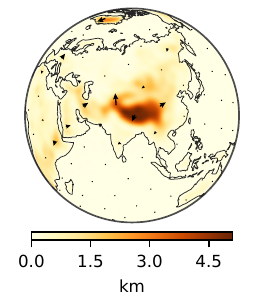}
    \caption{Himalaya and Tibetan Plateau: $f_+$, in km, with the arrows of (b).}
  \end{subfigure}\hfill
  \begin{subfigure}[t]{0.32\textwidth}
    \centering
    \includegraphics[width=\linewidth]{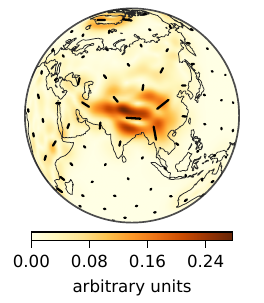}
    \caption{The same view: $\lvert\sw{2}{g}\rvert$ with the ridge orientation drawn as sticks.}
  \end{subfigure}\hfill
  \begin{subfigure}[t]{0.32\textwidth}
    \centering
    \includegraphics[width=\linewidth]{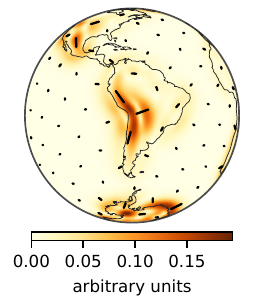}
    \caption{The Andes: as (e).}
  \end{subfigure}
  \caption{The topography of the Earth, its slope and its ridges, from lifted convolutions of spin $0$ to spins $1$ and $2$.  The ridge strength (c) is the largest response of the ridge detector over all orientations, and the ridge orientation (e), (f) is the spin-$2$ component $\sw{2}{g}$ of the response, drawn as headless sticks at the angle $\arg(\sw{2}{g})/2$ with its modulus shown as colour.  The orthographic views are centred on $32^\circ$N, $85^\circ$E (d), (e) and on $20^\circ$S, $68^\circ$W (f).}
  \label{fig:topography}
\end{figure}

\subsection{Gravity: one potential and its derivatives of spin 0, 1 and 2}
\label{sec:illustrations:gravity}

The geoid is the surface of constant gravity potential that coincides with mean sea level over the oceans, and its height $N$ above the reference ellipsoid is a scalar field, which we take from the Earth Gravitational Model EGM2008 \cite{pavlis2012egm2008}.  By Bruns' formula \cite{heiskanen1967physical}, $N$ is the disturbing potential $T$, the difference between the gravity potential of the Earth and that of the ellipsoid, divided by the normal gravity $\gamma$, and the quantities that geodesy measures are derivatives of $T$ of increasing spin.  The gravity anomaly, measured by gravimeters, is the scalar field with coefficients $\Delta g_{\ell m} = \frac{\gamma}{R_\oplus}(\ell-1)N_{\ell m}$ in the spherical approximation, \ie\ the lifted convolution $\Cop{0}{0}$ with the multiplier $\frac{\gamma}{R_\oplus}(\ell-1)$.  The deflection of the vertical, the angle between the direction of gravity and the normal to the ellipsoid, measured by astronomical positioning, is the vector $-\nabla N/R_\oplus$, \ie\ the spin-$1$ field $\varepsilon = \eth N/R_\oplus$ given by $\Cop{0}{1}$ with the multiplier $\sqrt{\ell(\ell+1)}/R_\oplus$.  The horizontal trace-free part of the tensor of second derivatives of $T$, measured by the gradiometer of the GOCE satellite, is the spin-$2$ field $\Gamma = \frac12\eth^2 T/R_\oplus^2$ given by $\Cop{0}{2}$.

The classical integral formulae of physical geodesy are lifted convolutions between these spins.  The formula of Vening Meinesz \cite{heiskanen1967physical}, which gives the deflection of the vertical from gravity anomalies, is the spin-changing convolution $\Cop{0}{1}$ with the multiplier $\sqrt{\ell(\ell+1)}/(\gamma(\ell-1))$.  The recovery of the geoid from deflections of the vertical, astrogeodetic levelling, is the lifted convolution $\Cop{1}{0}$ that inverts $\Cop{0}{1}$, and the recovery of the geoid from its gradients is the lifted convolution $\Cop{2}{0}$ that inverts $\frac12\eth^2$, the operator that also underlies the Kaiser--Squires relation of weak lensing.  Each of these relations is a single multiplier in harmonic space and holds to machine precision for band-limited fields: in Fig.~\ref{fig:gravity} the formula of Vening Meinesz reproduces $\varepsilon$ from $\Delta g$ to a relative error of $10^{-16}$, and the geoid is recovered from $\varepsilon$ and from $\Gamma$ to relative errors of $2\times10^{-16}$ and $10^{-16}$, apart from its monopole and dipole, which the derivatives do not determine.

\begin{figure}[tp]
  \centering
  \begin{subfigure}[t]{0.32\textwidth}
    \centering
    \includegraphics[width=\linewidth]{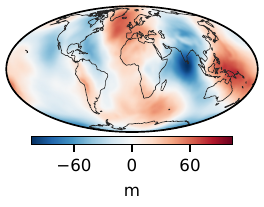}
    \caption{Geoid height $N$, in m.}
  \end{subfigure}\hfill
  \begin{subfigure}[t]{0.32\textwidth}
    \centering
    \includegraphics[width=\linewidth]{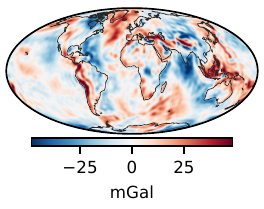}
    \caption{Gravity anomaly $\Delta g = \Cop{0}{0}N$, in mGal.}
  \end{subfigure}\hfill
  \begin{subfigure}[t]{0.32\textwidth}
    \centering
    \includegraphics[width=\linewidth]{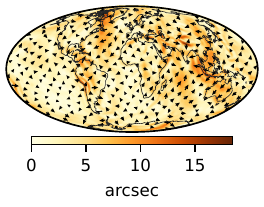}
    \caption{Deflection of the vertical $\varepsilon = \Cop{0}{1}N$, in arcseconds, drawn as arrows.}
  \end{subfigure}\\[1ex]
  \begin{subfigure}[t]{0.32\textwidth}
    \centering
    \includegraphics[width=\linewidth]{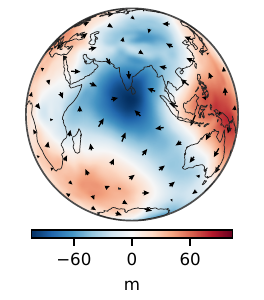}
    \caption{Indian Ocean: $N$ with $\varepsilon$ drawn as arrows.}
  \end{subfigure}\hfill
  \begin{subfigure}[t]{0.32\textwidth}
    \centering
    \includegraphics[width=\linewidth]{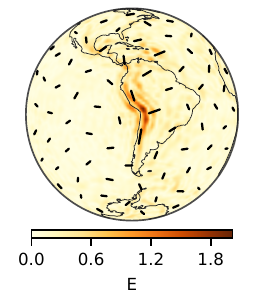}
    \caption{The Andes: gravity gradients $\Gamma = \Cop{0}{2}N$, in E\"otv\"os, drawn as sticks.}
  \end{subfigure}\hfill
  \begin{subfigure}[t]{0.32\textwidth}
    \centering
    \includegraphics[width=\linewidth]{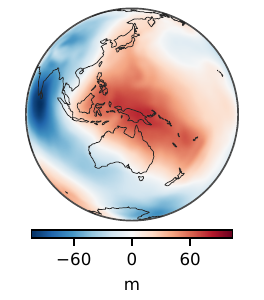}
    \caption{Indonesia and New Guinea: $N$ recovered from $\varepsilon$ by $\Cop{1}{0}$.}
  \end{subfigure}
  \caption{The gravity field of the Earth from EGM2008, as fields of spin $0$, $1$ and $2$ related by lifted convolutions.  The gravity anomaly, in mGal ($10^{-5}$ m s$^{-2}$), is shown on a scale saturated at the 99.5th percentile of its modulus; the arrows of the deflection point down the slope of the geoid; and the sticks of the gradients, in E\"otv\"os ($10^{-9}$ s$^{-2}$), lie along the direction in which the second derivative of the disturbing potential is largest.  The orthographic views are centred on $5^\circ$S, $80^\circ$E (d), $20^\circ$S, $75^\circ$W (e) and $5^\circ$S, $140^\circ$E (f).  The deflections point into the geoid low of the Indian Ocean and away from the geoid high over New Guinea, and the gradients are largest along the Andes and the adjacent trench.}
  \label{fig:gravity}
\end{figure}

\subsection{The geomagnetic field: a potential field of spin 0 and 1, continued to the core}
\label{sec:illustrations:geomagnetic}

The main magnetic field of the Earth, generated by the dynamo in the fluid outer core, is the gradient of a potential, $\bm{B} = -\nabla V$, which the International Geomagnetic Reference Field provides as spherical harmonic coefficients up to degree 13 \cite{alken2021igrf}.  With $v = V/R_\oplus$ on the surface of the Earth, the radial component of the field is the scalar field $B_r = \Cop{0}{0}[\ell+1]\,v$, and the horizontal component is the spin-$1$ field $B_\theta + \ii B_\varphi = \eth v = \Cop{0}{1}[\sqrt{\ell(\ell+1)}]\,v$, the field that aligns a compass needle.  The declination, the angle by which the needle deviates east of geographic north, is the direction of the horizontal field measured from north, and the inclination, the angle by which the field dips below the horizontal, follows from the ratio of $-B_r$ to the modulus of the horizontal field.  Each degree $\ell$ of a potential field of internal origin decays with radius as $r^{-(\ell+2)}$, so that the field at the boundary between the core and the mantle, at radius $r_c = 3485$ km, follows from that at the surface by the multiplier $(R_\oplus/r_c)^{\ell+2}$, the downward continuation.  The multiplier preserves spin, and we apply it to the radial field as $\Cop{0}{0}$ and to the horizontal vector field as $\Cop{1}{1}$, directly and without reference to its components.  Downward continuation amplifies the coefficients of degree 13 by $(R_\oplus/r_c)^{15} \approx 8500$ and reveals the structure of the field at the top of the core (Fig.~\ref{fig:geomagnetic}): intense patches of flux at high latitudes, and patches of reversed flux beneath the South Atlantic and southern Africa, which weaken the surface field in the South Atlantic Anomaly.

\begin{figure}[tp]
  \centering
  \begin{subfigure}[t]{0.32\textwidth}
    \centering
    \includegraphics[width=\linewidth]{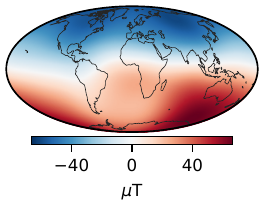}
    \caption{Radial field $B_r = \Cop{0}{0}v$ at the surface, in $\mu$T.}
  \end{subfigure}\hfill
  \begin{subfigure}[t]{0.32\textwidth}
    \centering
    \includegraphics[width=\linewidth]{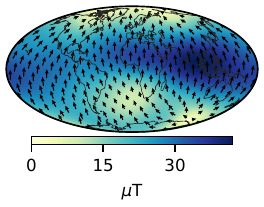}
    \caption{Horizontal field $\eth v = \Cop{0}{1}v$, in $\mu$T, drawn as arrows.}
  \end{subfigure}\hfill
  \begin{subfigure}[t]{0.32\textwidth}
    \centering
    \includegraphics[width=\linewidth]{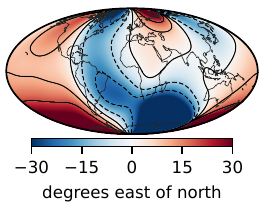}
    \caption{Declination, in degrees east of geographic north.}
  \end{subfigure}\\[1ex]
  \begin{subfigure}[t]{0.32\textwidth}
    \centering
    \includegraphics[width=\linewidth]{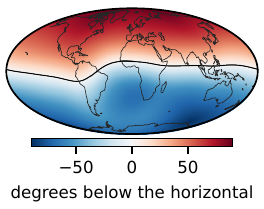}
    \caption{Inclination, in degrees below the horizontal.}
  \end{subfigure}\hfill
  \begin{subfigure}[t]{0.32\textwidth}
    \centering
    \includegraphics[width=\linewidth]{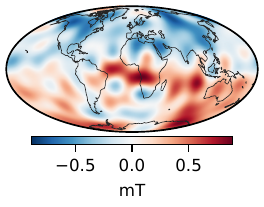}
    \caption{Radial field continued by $\Cop{0}{0}$ to the boundary between the core and the mantle, in mT.}
  \end{subfigure}\hfill
  \begin{subfigure}[t]{0.32\textwidth}
    \centering
    \includegraphics[width=\linewidth]{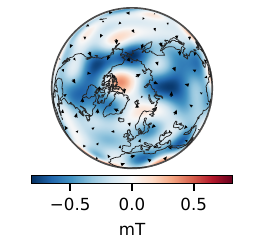}
    \caption{The same boundary from above the North pole, with the horizontal field continued by $\Cop{1}{1}$ drawn as arrows.}
  \end{subfigure}
  \caption{The geomagnetic field of IGRF-14 in 2025, as fields of spin $0$ and $1$, and its downward continuation to the core.  The radial field is blue where it points into the Earth and red where it points out, and the arrows of the horizontal field point in the direction of a compass needle.  The declination is saturated at $\pm30^\circ$, with contours every $10^\circ$, among which those of $0^\circ$ are the agonic lines, and the magnetic equator is the contour of zero inclination.  In (f) the coastlines are projected radially onto the boundary for orientation.}
  \label{fig:geomagnetic}
\end{figure}

\subsection{Wind: vorticity, divergence and smoothing}
\label{sec:illustrations:wind}

We consider the near-surface wind in January, averaged over 1991--2020, from the NCEP/NCAR reanalysis \cite{kalnay1996ncep}.  With northward and eastward components $(v,u)$, and since $e_\theta$ points south, the wind is the spin-$1$ field $\sw{1}{f} = -v + \ii u$.  For a spin-$1$ field the spin lowering operator gives the divergence and the vorticity, the radial component of the curl, together,
\begin{equation}
  \bar\eth\,\sw{1}{f} = -R_\oplus\big(\nabla\cdot\bm{v} + \ii\,\hat{\bm{r}}\cdot\nabla\times\bm{v}\big),
  \label{eq:helmholtz}
\end{equation}
so that the single lifted convolution $\Cop{1}{0}$ yields both as the real and imaginary parts of one scalar field, exactly for the band-limited field, without finite differences and without special treatment of the poles.  Since $\eth\bar\eth$ is the Laplacian on spin-$1$ fields, applying the lifted convolution $\Cop{0}{1}$ of $\eth\nabla^{-2}$ to the divergence and to the vorticity separately gives the Helmholtz decomposition of the wind into a curl-free part, the flow into and out of regions of convergence, and a divergence-free part, the rotational flow of the trade winds and the westerlies, which sum to the wind to a relative error of $10^{-16}$ (Fig.~\ref{fig:wind}).

Smoothing the wind by a Gaussian is the spin-preserving lifted convolution $\Cop{1}{1}$ whose kernel is the Gaussian multiplier $e^{-\ell(\ell+1)\sigma^2/2}$.  A common alternative smooths the components $v_\theta$ and $v_\varphi$ as though they were two scalar fields.  Since the frame $(e_\theta, e_\varphi)$ turns around the poles of the coordinates, this averages vectors expressed in different frames, and near a pole the components vary rapidly even where the wind itself is uniform.  The poles of the coordinates are a matter of convention, and to expose the dependence of each method on that convention we place one pole at the location of the strongest smoothed wind and smooth in both systems of coordinates (Fig.~\ref{fig:wind-smoothing}).  The spin-$1$ smoothing, which is rotation-equivariant, gives the same result in either system to a relative error of $6\times10^{-16}$.  The component-wise smoothing gives a plausible result in geographic coordinates, where the strongest winds lie far from the poles, but with the pole placed in the strong flow it combines the components across the turning frame and destroys the flow near the pole, with a relative error of ten per cent over the sphere: its result depends on the choice of coordinates, which is the error that the spin formalism removes.

\begin{figure}[tp]
  \centering
  \begin{subfigure}[t]{0.32\textwidth}
    \centering
    \includegraphics[width=\linewidth]{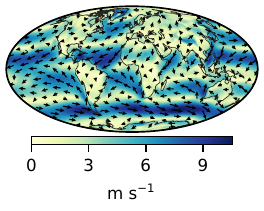}
    \caption{Wind speed $\lvert\sw{1}{f}\rvert$, in m\,s$^{-1}$, with the wind drawn as arrows.}
  \end{subfigure}\hfill
  \begin{subfigure}[t]{0.32\textwidth}
    \centering
    \includegraphics[width=\linewidth]{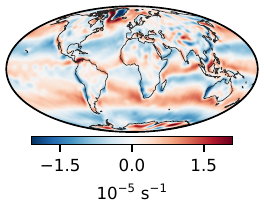}
    \caption{Vorticity $-\mathrm{Im}(\bar\eth\,\sw{1}{f})/R_\oplus$, in $10^{-5}$ s$^{-1}$.}
  \end{subfigure}\hfill
  \begin{subfigure}[t]{0.32\textwidth}
    \centering
    \includegraphics[width=\linewidth]{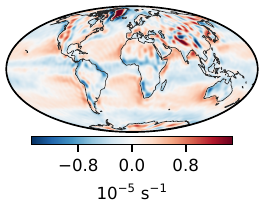}
    \caption{Divergence $-\mathrm{Re}(\bar\eth\,\sw{1}{f})/R_\oplus$, in $10^{-5}$ s$^{-1}$.}
  \end{subfigure}\\[1ex]
  \begin{subfigure}[t]{0.32\textwidth}
    \centering
    \includegraphics[width=\linewidth]{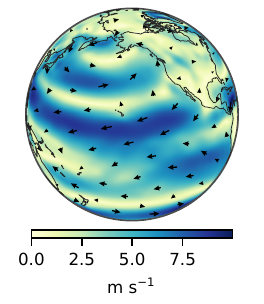}
    \caption{North Pacific: the wind.}
  \end{subfigure}\hfill
  \begin{subfigure}[t]{0.32\textwidth}
    \centering
    \includegraphics[width=\linewidth]{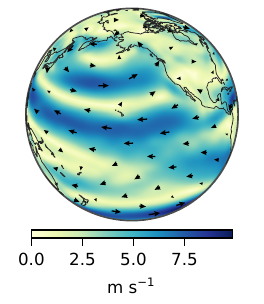}
    \caption{Its divergence-free part.}
  \end{subfigure}\hfill
  \begin{subfigure}[t]{0.32\textwidth}
    \centering
    \includegraphics[width=\linewidth]{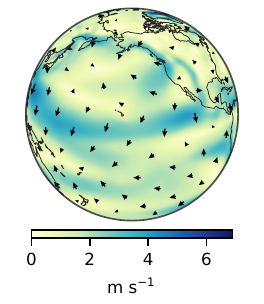}
    \caption{Its curl-free part.}
  \end{subfigure}
  \caption{The January near-surface wind and its vorticity, divergence and Helmholtz decomposition, from the lifted convolutions $\Cop{1}{0}$ and $\Cop{0}{1}$.  The vorticity and the divergence are shown on scales saturated at the 99.5th percentile of their moduli, red for anticlockwise rotation and for outflow and blue for clockwise rotation and for inflow, and the orthographic views are centred on $20^\circ$N, $150^\circ$W.  The vorticity is positive in the cyclonic flow around the Aleutian and Icelandic lows and negative in the anticyclones of the northern subtropics, and the divergence is negative along the intertropical convergence zone, into which the trade winds of both hemispheres flow.}
  \label{fig:wind}
\end{figure}

\begin{figure}[tp]
  \centering
  \begin{subfigure}[t]{0.48\textwidth}
    \centering
    \includegraphics[width=\linewidth]{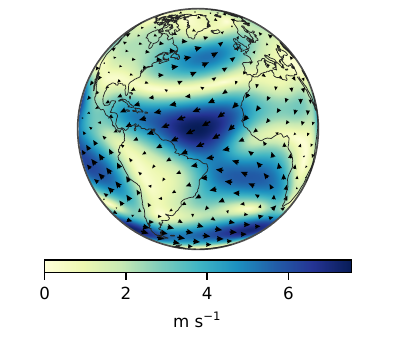}
    \caption{Smoothing by the spin-$1$ convolution $\Cop{1}{1}$, which is the same in either system of coordinates.}
    \label{fig:wind-smoothing:a}
  \end{subfigure}\hfill
  \begin{subfigure}[t]{0.48\textwidth}
    \centering
    \includegraphics[width=\linewidth]{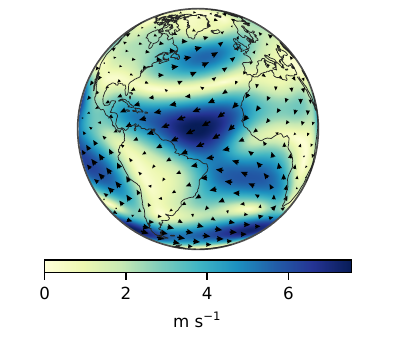}
    \caption{Smoothing of the components $v_\theta$ and $v_\varphi$ as scalar fields in geographic coordinates.}
  \end{subfigure}\\[1ex]
  \begin{subfigure}[t]{0.48\textwidth}
    \centering
    \includegraphics[width=\linewidth]{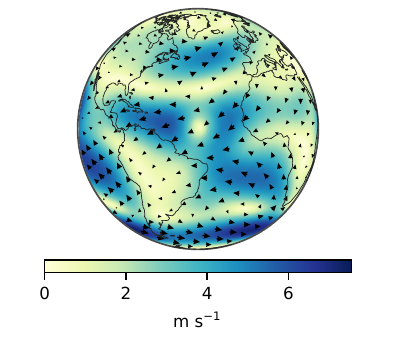}
    \caption{Smoothing of the components with the pole of the coordinates at the centre of the view.}
    \label{fig:wind-smoothing:c}
  \end{subfigure}\hfill
  \begin{subfigure}[t]{0.48\textwidth}
    \centering
    \includegraphics[width=\linewidth]{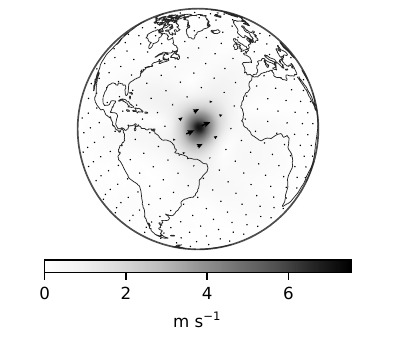}
    \caption{Modulus of the difference between (\subref{fig:wind-smoothing:c}) and (\subref{fig:wind-smoothing:a}), with the difference drawn as arrows.}
  \end{subfigure}
  \caption{Smoothing the wind as a spin-$1$ field and as two scalar components.  Orthographic views centred on the strongest wind after smoothing, at $12^\circ$N, $40^\circ$W, of the wind smoothed by a Gaussian of width $8/L$ radians, with the speed shown as colour on a common scale, in m s$^{-1}$, and the wind drawn as arrows.  With the pole at the centre, the component-wise smoothing destroys the flow within about the width of the smoothing of the pole, where the frame turns fastest.}
  \label{fig:wind-smoothing}
\end{figure}

\subsection{Ocean currents: gyres and upwelling}
\label{sec:illustrations:ocean}

The surface currents of the ocean in 2025, averaged over the year at a depth of 5 m, from the Global Ocean Data Assimilation System \cite{behringer2004godas}, form a spin-$1$ field like the wind, which we set to zero on land.  The lifted convolution $\Cop{1}{0}$ again gives their vorticity and divergence.  The stream function $\psi = R_\oplus^2\nabla^{-2}(\hat{\bm{r}}\cdot\nabla\times\bm{v})$ of the divergence-free part, obtained from the vorticity by the inverse Laplacian, has as its contours the streamlines of that flow, along which the rotational part of the current runs with higher values of $\psi$ to its right.  The streamlines trace the subtropical gyres, which circulate clockwise in the northern hemisphere and anticlockwise in the southern hemisphere, and the Antarctic Circumpolar Current (Fig.~\ref{fig:ocean}).  The divergence of the surface flow marks upwelling, where water rises from below to replace surface water that is carried away.  The trade winds drive a transport of the surface water away from the equator on both sides, so that the surface flow diverges along the equator, in the cold tongues of the eastern Pacific and Atlantic, and along the coasts of Peru and north-west Africa, where the winds blow along the coast.

\begin{figure}[tp]
  \centering
  \begin{subfigure}[t]{0.48\textwidth}
    \centering
    \includegraphics[width=\linewidth]{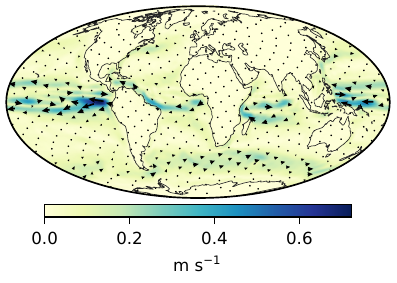}
    \caption{Current speed $\lvert\sw{1}{f}\rvert$, in m\,s$^{-1}$, with the currents drawn as arrows.}
  \end{subfigure}\hfill
  \begin{subfigure}[t]{0.48\textwidth}
    \centering
    \includegraphics[width=\linewidth]{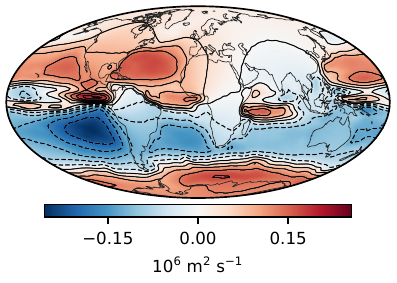}
    \caption{Stream function $\psi$ of the divergence-free part, in $10^6$ m$^2$ s$^{-1}$, with its contours, the streamlines.}
  \end{subfigure}\\[1ex]
  \begin{subfigure}[t]{0.48\textwidth}
    \centering
    \includegraphics[width=\linewidth]{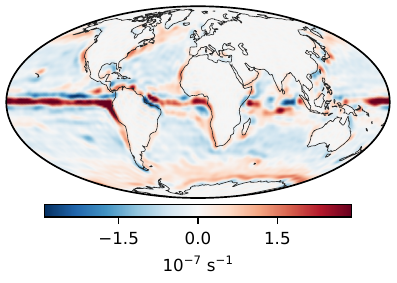}
    \caption{Divergence $-\mathrm{Re}(\bar\eth\,\sw{1}{f})/R_\oplus$, in $10^{-7}$ s$^{-1}$.}
  \end{subfigure}\hfill
  \begin{subfigure}[t]{0.48\textwidth}
    \centering
    \includegraphics[width=\linewidth]{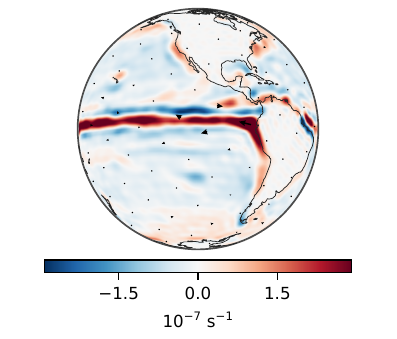}
    \caption{Eastern equatorial Pacific: the divergence, with the currents drawn as arrows.}
  \end{subfigure}
  \caption{The surface currents of the ocean in 2025, their stream function and their divergence, from the lifted convolution $\Cop{1}{0}$.  The divergence is shown on a scale saturated at the 99.5th percentile of its modulus, red where the surface flow diverges and water upwells and blue where it converges, and the orthographic view is centred on $5^\circ$S, $110^\circ$W.  The fastest currents are the westward equatorial currents; the subtropical gyres appear as maxima of $\psi$ in the northern hemisphere and minima in the southern hemisphere; and the divergence is concentrated along the equator in the Pacific and Atlantic and along the coast of Peru.}
  \label{fig:ocean}
\end{figure}

\subsection{Polarisation of the cosmic microwave background}
\label{sec:illustrations:cmb}

The linear polarisation of the cosmic microwave background, described by the Stokes parameters $Q$ and $U$, is the spin-$2$ field $Q + \ii U$.  Its physical content is carried by two scalar fields, the $E$ and $B$ modes, which do not depend on the frame and which separate the contributions of different physical processes: density perturbations in the early Universe generate only $E$ modes, while $B$ modes are generated by the gravitational lensing of $E$ modes and would be generated by primordial gravitational waves.  By Proposition~\ref{prop:physics}(b), $a_E = -\frac12(a_2 + a_{-2})$ and $a_B = \frac{\ii}{2}(a_2 - a_{-2})$, where $a_{\pm2}$ are the coefficients of $Q \pm \ii U$, so that each mode is a sum of the lifted convolutions $\Cop{2}{0}$ and $\Cop{-2}{0}$.  Conversely, the polarisation follows from the two modes by the lifted convolutions $\Cop{0}{\pm2}$, as $a_{\pm2} = -(a_E \pm \ii a_B)$, so that the decomposition may be inverted exactly.

We consider the Planck 2018 SMICA maps of the temperature and polarisation of the cosmic microwave background \cite{planck2020smica}, in Galactic coordinates, so that the plane of the Galaxy lies along the equator of the maps.  We decompose the polarisation into its $E$ and $B$ modes and synthesise the polarisation again from the two modes (Fig.~\ref{fig:cmb}).  On the large scales considered here, the polarisation measured by Planck combines the signal with noise and residual Galactic emission, most evident along the Galactic plane, which an analysis would mask; the $B$ mode, dominated on these scales by noise and by these residuals, has a root mean square of 79 per cent of that of the $E$ mode.  The decomposition and its inverse are single multipliers in harmonic space, $-\frac12$ and $\pm\frac{\ii}{2}$ one way and $-1$ and $\mp\ii$ the other, so that the polarisation synthesised from the two modes reproduces the measured polarisation to the precision of the arithmetic.  Starting from the map of the polarisation, computing the maps of the $E$ and $B$ modes and synthesising the map of the polarisation again, including all the harmonic transforms, recovers the map to a relative error of $1.3\times10^{-13}$, set by the machine precision of the transforms.

\begin{figure}[tp]
  \centering
  \begin{subfigure}[t]{0.32\textwidth}
    \centering
    \includegraphics[width=\linewidth]{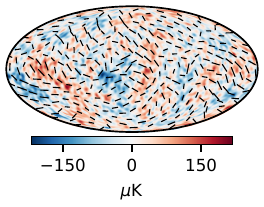}
    \caption{Temperature anisotropy, in $\mu$K, with the polarisation drawn as sticks.}
  \end{subfigure}\hspace{0.04\textwidth}
  \begin{subfigure}[t]{0.32\textwidth}
    \centering
    \includegraphics[width=\linewidth]{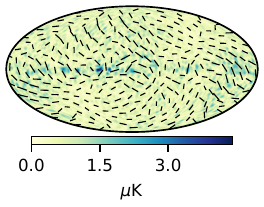}
    \caption{Polarised intensity $\lvert Q + \ii U\rvert$, in $\mu$K, with the same sticks.}
  \end{subfigure}\\[1ex]
  \begin{subfigure}[t]{0.32\textwidth}
    \centering
    \includegraphics[width=\linewidth]{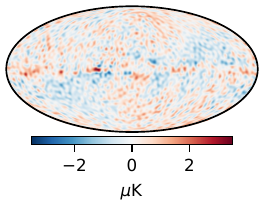}
    \caption{$E$ mode, in $\mu$K.}
  \end{subfigure}\hfill
  \begin{subfigure}[t]{0.32\textwidth}
    \centering
    \includegraphics[width=\linewidth]{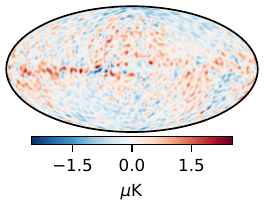}
    \caption{$B$ mode, in $\mu$K.}
  \end{subfigure}\hfill
  \begin{subfigure}[t]{0.32\textwidth}
    \centering
    \includegraphics[width=\linewidth]{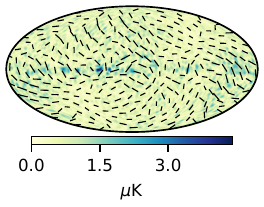}
    \caption{Polarised intensity synthesised from the $E$ and $B$ modes, with its polarisation drawn as sticks.}
  \end{subfigure}
  \caption{The temperature and polarisation of the cosmic microwave background measured by Planck (top), and the $E$ and $B$ modes of the polarisation, from the lifted convolutions $\Cop{2}{0} + \Cop{-2}{0}$, with the polarisation synthesised again from the two modes by $\Cop{0}{2}$ (bottom), in Galactic coordinates.  The sticks are drawn along the direction of polarisation, at the angle $\arg(Q + \ii U)/2$, after smoothing by the spin-$2$ convolution $\Cop{2}{2}$ with a Gaussian of width $4/L$ radians, so that they follow its large-scale pattern.  Both modes show residual Galactic emission along the equator of the maps.}
  \label{fig:cmb}
\end{figure}

\subsection{Weak gravitational lensing: shear and convergence}
\label{sec:illustrations:lensing}

The images of distant galaxies are distorted by the gravitational deflection of their light by the matter along the line of sight.  The convergence $\kappa$, the projected density of matter weighted by the geometry of the lensing, is a scalar field, and the shear $\gamma$, the anisotropic stretching of the images, is a spin-$2$ field.  The two are related by the spherical Kaiser--Squires relation \cite{kaiser1993mapping,wallis2022mapping}, which is the lifted convolution $\gamma = \Cop{0}{2}[\mathcal{D}]\kappa$ with $\mathcal{D}_\ell = -\sqrt{(\ell+2)(\ell-1)/(\ell(\ell+1))}$, Proposition~\ref{prop:physics}(c), and whose inverse $\kappa = \Cop{2}{0}[\mathcal{D}^{-1}]\gamma$ recovers the convergence, apart from its monopole and dipole, from the shear, on the full sphere and without the planar approximation of the original method.  Since the galaxies are stretched tangentially around concentrations of mass, the sticks of the shear circle the overdense regions and point radially away from the underdense regions.

We simulate the convergence as a shifted lognormal field \cite{hilbert2011lognormal}, which captures the skewness of the distribution of matter, with a realistic power spectrum: the Limber integral of the non-linear matter power spectrum computed by \texttt{CAMB} \cite{lewis2000camb} with \texttt{HMcode} \cite{mead2021hmcode} for the Planck 2018 cosmology, for a Euclid-like distribution of source galaxies with median redshift 0.9.  We compute the shear from the convergence and recover the convergence from the shear (Fig.~\ref{fig:lensing}); since the simulated convergence has no monopole or dipole, the recovery is exact, to a relative error of $10^{-16}$.

\begin{figure}[tp]
  \centering
  \begin{subfigure}[t]{0.48\textwidth}
    \centering
    \includegraphics[width=\linewidth]{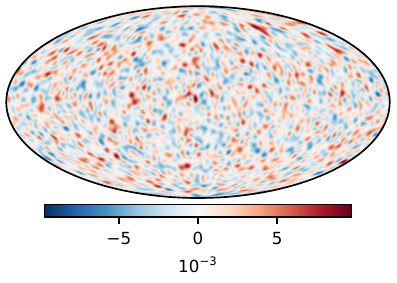}
    \caption{Simulated convergence $\kappa$.}
    \label{fig:lensing:a}
  \end{subfigure}\hfill
  \begin{subfigure}[t]{0.48\textwidth}
    \centering
    \includegraphics[width=\linewidth]{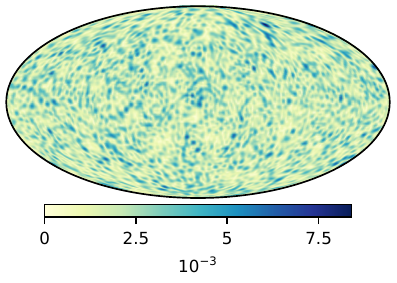}
    \caption{Modulus of the shear $\gamma = \Cop{0}{2}[\mathcal{D}]\kappa$.}
  \end{subfigure}\\[1ex]
  \begin{subfigure}[t]{0.48\textwidth}
    \centering
    \includegraphics[width=\linewidth]{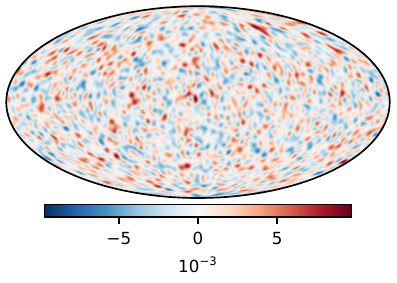}
    \caption{Convergence recovered from the shear by $\Cop{2}{0}[\mathcal{D}^{-1}]$, on the scale of (\subref{fig:lensing:a}).}
  \end{subfigure}\hfill
  \begin{subfigure}[t]{0.48\textwidth}
    \centering
    \includegraphics[width=\linewidth]{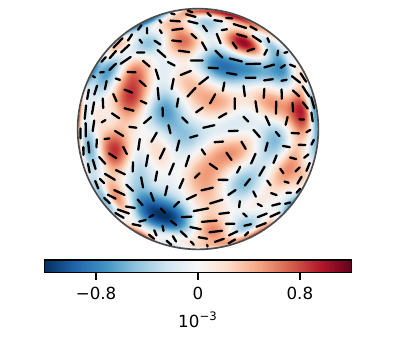}
    \caption{Smoothed convergence, with the smoothed shear drawn as sticks.}
  \end{subfigure}
  \caption{The convergence and shear of weak gravitational lensing, related by the lifted convolutions $\Cop{0}{2}$ and $\Cop{2}{0}$.  The convergence is red for overdense and blue for underdense regions.  The orthographic view (d), centred on $20^\circ$N, $34^\circ$E in the coordinates of the simulation, shows the convergence smoothed by a Gaussian of width $8/L$ radians, with the shear smoothed in the same way drawn as headless sticks along the direction in which the images are stretched, at the angle $\arg(\gamma)/2$; the smoothing is a multiplier and commutes with $\Cop{0}{2}[\mathcal{D}]$, so that the smoothed shear is the shear of the smoothed convergence.  The sticks circle the overdense regions and point away from the underdense regions.}
  \label{fig:lensing}
\end{figure}

In summary, in each example we have applied lifted convolutions $\Cop{s}{s'}$ whose input and output spins are those of the physical quantities concerned.  Derivatives raise or lower the spin, as do their compositions, such as the formula of Vening Meinesz and the Kaiser--Squires relation; smoothing and the continuation of potential fields preserve it; the decomposition of polarisation into $E$ and $B$ modes lowers it; and the steerable ridge detector maps a scalar field to the tuple of spin fields that makes up a function on the rotation group.  Each operator is a single multiplier of the harmonic coefficients, computed exactly and efficiently for band-limited fields, invertible where the physics allows, as the round trips of the geoid, the polarisation and the convergence confirm to the precision of the arithmetic, and rotation-equivariant, so that its result does not depend on the choice of coordinates.  The comparison with the smoothing of the components of the wind as scalar fields shows the error that the spin formalism removes.

\section{Conclusions}
\label{sec:discussion}

We have developed a unified description of convolution on the sphere by lifting to the rotation group.  A spin-$s$ field is lifted to $\SOthree$ through the irreducible representation of spin $s$ of the stabiliser $\SOtwo$, convolved there with an arbitrary kernel, and projected back to the sphere through the adjoint of the lift for a possibly different spin $s'$.  In spin spherical harmonic space the resulting operator multiplies each coefficient by a single conjugated Wigner coefficient of the kernel, up to a sign.  Every rotation-equivariant spherical convolution considered is recovered by constraining the kernel, choosing the input and output spins, and including or omitting the lift and the projection.  The constants of the literature follow from the different ways in which a kernel on the sphere is identified with a kernel on the rotation group.

The construction also settles what an equivariant spherical convolution can be.  Every bounded rotation-equivariant map between spin fields is a multiplier of this form, and the equivariant maps to and from the rotation group are read from a single column or row of the kernel.  A function on the rotation group azimuthally band-limited at $N$ is a tuple of $2N-1$ spin fields, so that a directional convolution with a steerable kernel is a finite tuple of spin-changing convolutions, computed in $\mathcal{O}(NL^3)$ operations.  A key result of the article is that, for maps between fields of fixed input and output spin, directional kernels, output on the sphere and equivariance cannot be achieved simultaneously.  The anisotropic constructions of the literature that keep their output on the sphere are evaluations of a lifted convolution on a section of the bundle or multipliers that depend on the azimuthal order, and they are equivariant only in special cases.

The geometric description yields two further key insights.  The spin-weighted functions of physics, defined by the phase that they acquire under a rotation of the local frame, are the sections of the line bundles associated to the quotient $\SOthree/\SOtwo$, or equivalently Mackey functions on the rotation group.  Making this identification explicit links the physics and the bundle literatures.  The phase that accompanies the rotation of a spin field is thereby expressed through the geometry of the bundle: it is the character of the stabiliser evaluated on the inverse twist of the section under the rotation, so that the rotation of a spin field is precisely the representation induced from that character.

All of the convolutions may be computed exactly and efficiently by spin spherical harmonic and Wigner transforms, at a cost that follows from the part of the kernel that each operator reads.  We have implemented every operator of the article in the open-source \texttt{s2conv} code, whose numerical tests confirm the reductions and the equivariance of the operators to machine precision, and whose computation time scales as expected.  We have applied the operators to the relief, gravity, magnetic field, winds and currents of the Earth and to the polarisation of the cosmic microwave background and weak gravitational lensing.

In addition to unifying existing spherical convolutions, the framework also suggests new ones.  In a forthcoming article \cite{mulder2026spin} we build on the operators presented to construct artificial intelligence (AI) architectures based on spin spherical Fourier neural operators, in which the multipliers are learned and in which a lift, a pointwise non-linearity on the rotation group and a projecting convolution yield equivariant layers for spin fields.

\begin{acknowledgments}
This work was funded under the embedded CSE programme of the ARCHER2 UK National Supercomputing Service (\url{https://www.archer2.ac.uk}), hosted by the University of Edinburgh and funded by EPSRC, through the GPU eCSE project \emph{Accelerated and differentiable spherical transforms} (GPU-eCSE01-41).
\end{acknowledgments}

\section*{Author declarations}
\subsection*{Conflict of interest}
The author has no conflicts to disclose.

\subsection*{Use of large language models}
Large language models (LLMs) were used in the writing of this manuscript.  The author takes full responsibility for the accuracy of the material presented.

\section*{Data availability}
No new observational data were generated for this article.  The measured fields of Sec.~\ref{sec:illustrations} are publicly available from the sources cited there, and the only data generated are the simulated convergence and shear of Sec.~\ref{sec:illustrations:lensing}.  The data of all of the examples, the simulation included, are distributed with the open-source \texttt{s2conv} code (\url{https://github.com/astro-informatics/s2conv}), together with the notebook that produces the figures and the benchmark scripts that produce the numerical results of Sec.~\ref{sec:computation:validation}.

\appendix

\section{Wigner functions and spin harmonics}
\label{app:conventions}

Every constant in the article depends on the choices fixed in Sec.~\ref{sec:prelim}.  This appendix collects the identities for Wigner functions and spin harmonics that the proofs use.

\paragraph{Wigner $d$-functions.}
The real functions $d^{\ell}_{mn}$ in \eqref{eq:wignerD} are those of Varshalovich \emph{et al.}~\cite{varshalovich1988quantum}, fixed by
\begin{equation}
  d^{\ell}_{\ell n}(\beta) = \sqrt{\frac{(2\ell)!}{(\ell+n)!\,(\ell-n)!}}\; \Bigl(-\sin\tfrac{\beta}{2}\Bigr)^{\ell-n} \Bigl(\cos\tfrac{\beta}{2}\Bigr)^{\ell+n},
  \qquad
  d^{\ell}_{m0}(\beta) = \sqrt{\frac{(\ell-m)!}{(\ell+m)!}}\; P^{m}_{\ell}(\cos\beta),
  \label{eq:app-d-normalisation}
\end{equation}
so that $d^{1}_{10}(\beta) = -\sin\beta/\sqrt{2}$ and $d^{1}_{11}(\beta) = \cos^2(\beta/2)$ \cite{price2024differentiable}.  They satisfy the symmetries
\begin{equation}
  d^{\ell}_{mn}(\beta) = (-1)^{m-n}\, d^{\ell}_{nm}(\beta) = d^{\ell}_{-n,-m}(\beta) = (-1)^{m-n}\, d^{\ell}_{mn}(-\beta) = (-1)^{\ell+m}\, d^{\ell}_{m,-n}(\pi-\beta),
  \label{eq:app-d-symmetries}
\end{equation}
for which the $\ell=1$ entries above provide simple sign checks, and take the values
\begin{equation}
  d^{\ell}_{mn}(0) = \delta_{mn},
  \qquad
  d^{\ell}_{mn}(\pi) = (-1)^{\ell-n}\, \delta_{m,-n} = (-1)^{\ell+m}\, \delta_{m,-n}
  \label{eq:app-d-poles}
\end{equation}
at the ends of the interval (for $\ell = 1$: $d^{1}_{00}(\pi) = \cos\pi = -1$ and $d^{1}_{1,-1}(\pi) = \sin^2(\pi/2) = 1$).  Inserting \eqref{eq:app-d-symmetries} in \eqref{eq:wignerD} gives the conjugation symmetry of the $D$-functions and, from the homomorphism property with $D^{\ell}_{mn}(R_z(\chi)) = e^{-\ii m\chi}\delta_{mn}$, their behaviour under the stabiliser,
\begin{equation}
  D^{\ell *}_{mn}(\rho) = (-1)^{m-n}\, D^{\ell}_{-m,-n}(\rho),
  \qquad
  D^{\ell}_{mn}\bigl(R_z(\chi_1)\,\rho\, R_z(\chi_2)\bigr) = e^{-\ii m\chi_1}\, D^{\ell}_{mn}(\rho)\, e^{-\ii n\chi_2},
  \label{eq:app-D-conj}
\end{equation}
the second of which is the identity behind Corollary~\ref{cor:constraint}.  At the poles, \eqref{eq:app-d-poles} gives $D^{\ell}_{mn}(\alpha,0,\gamma) = \delta_{mn}\,e^{-\ii m(\alpha+\gamma)}$ and $D^{\ell}_{mn}(\alpha,\pi,\gamma) = (-1)^{\ell-n}\,\delta_{m,-n}\,e^{-\ii m(\alpha-\gamma)}$.

\paragraph{Spin raising and lowering.}
Acting on a spin-$s$ field, the operators in \eqref{eq:eth} are \cite{newman1966note,goldberg1967spin,price2024differentiable}
\begin{equation}
  \eth = -\sin^{s}\theta\, \Bigl( \partial_\theta + \frac{\ii}{\sin\theta}\, \partial_\varphi \Bigr) \sin^{-s}\theta,
  \qquad
  \bar\eth = -\sin^{-s}\theta\, \Bigl( \partial_\theta - \frac{\ii}{\sin\theta}\, \partial_\varphi \Bigr) \sin^{s}\theta,
  \label{eq:app-eth}
\end{equation}
and the spin harmonics are generated from \eqref{eq:Ylm} by
\begin{equation}
  \sw{s}{Y}_{\ell m} =
  \begin{cases}
    \sqrt{\dfrac{(\ell-s)!}{(\ell+s)!}}\; \eth^{s}\, Y_{\ell m}, & 0 \le s \le \ell, \\[2.5ex]
    (-1)^{s} \sqrt{\dfrac{(\ell+s)!}{(\ell-s)!}}\; \bar\eth^{-s}\, Y_{\ell m}, & -\ell \le s \le 0,
  \end{cases}
  \label{eq:app-sY-closed}
\end{equation}
the sign $(-1)^s$ in the second line being what fixes the phase of the negative-spin harmonics.  Setting $\gamma = 0$ in \eqref{eq:sY-D} gives the explicit form
\begin{equation}
  \sw{s}{Y}_{\ell m}(\theta,\varphi) = (-1)^{s} \sqrt{\frac{2\ell+1}{4\pi}}\; d^{\ell}_{m,-s}(\theta)\; e^{\ii m\varphi},
  \label{eq:app-sY-d}
\end{equation}
which reduces to \eqref{eq:Ylm} for $s = 0$ by \eqref{eq:app-d-normalisation} and shows directly that $\sw{s}{Y}_{\ell m}$ exists only for $\ell \ge \lvert s\rvert$.

\paragraph{Values at the poles.}
By \eqref{eq:app-sY-d} and \eqref{eq:app-d-poles},
\begin{equation}
  \sw{s}{Y}_{\ell m}(0,\varphi) = (-1)^{s} \sqrt{\frac{2\ell+1}{4\pi}}\; \delta_{m,-s}\; e^{-\ii s\varphi},
  \qquad
  \sw{s}{Y}_{\ell m}(\pi,\varphi) = (-1)^{\ell} \sqrt{\frac{2\ell+1}{4\pi}}\; \delta_{m,s}\; e^{\ii s\varphi} :
  \label{eq:app-sY-poles}
\end{equation}
at the north pole only the order $m = -s$ survives and at the south pole only $m = s$, with the windings $e^{\mp\ii s\varphi}$ of Remark~\ref{rem:no-global-section}.  For the constrained kernel of Corollary~\ref{cor:constraint} the same values give $d^{\ell}_{-s',-s}(0) = \delta_{ss'}$ and $d^{\ell}_{-s',-s}(\pi) = (-1)^{\ell+s}\,\delta_{s',-s}$, the pole conditions of the constrained kernel.

\paragraph{Conjugation.}
From \eqref{eq:app-sY-d} and the first two symmetries in \eqref{eq:app-d-symmetries},
\begin{equation}
  \sw{s}{Y}^{*}_{\ell m} = (-1)^{s+m}\; \sw{-s}{Y}_{\ell,-m},
  \label{eq:app-sY-conj}
\end{equation}
which extends $Y^{*}_{\ell m} = (-1)^m Y_{\ell,-m}$.  Consequently the complex conjugate of a spin-$s$ field $\sw{s}{f}$ is the spin-$(-s)$ field with coefficients $(\sw{s}{f}^{*})_{\ell m} = (-1)^{s+m}\, \sw{s}{f}^{*}_{\ell,-m}$, the relation between the two members of a real field of spin $s$ (Sec.~\ref{sec:lifting:fields}).

\paragraph{Addition theorem.}
For $x, y \in \sphere$ and any spins $s, s'$,
\begin{equation}
  \sum_{m=-\ell}^{\ell} \sw{s'}{Y}_{\ell m}(x)\; \sw{s}{Y}^{*}_{\ell m}(y)
  = (-1)^{s+s'}\, \frac{2\ell+1}{4\pi}\; D^{\ell}_{-s',-s}\bigl(s(x)^{-1} s(y)\bigr),
  \label{eq:app-addition}
\end{equation}
with $s(\cdot)$ the section \eqref{eq:section}.  The proof in Appendix~\ref{app:proofs} inserts \eqref{eq:sY-D} for both harmonics and uses $D^\ell(s(x))^\dagger D^\ell(s(y))=D^\ell(s(x)^{-1}s(y))$.  For noncoincident, nonantipodal points, write the relative rotation as $R_z(a)R_y(b)R_z(c)$.  Then $b$ is the separation, and the bearings measured from $e_\theta$ towards $e_\varphi$ are $a$ at $x$ towards $y$ and $\pi-c$ at $y$ towards $x$.  Using $(-1)^{s+s'}d^\ell_{-s',-s}=d^\ell_{s's}$, the right-hand side of \eqref{eq:app-addition} becomes $(2\ell+1)e^{\ii s'a}d^\ell_{s's}(b)e^{\ii sc}/(4\pi)$.  The literal quadrant-resolved angles of Tessore \etal\ \cite{tessore2025euclid} are instead $\alpha=a+\pi$ and $\alpha'=-c$, the distinction described in Proposition~\ref{prop:spin-change}.  For $s=s'=0$ this is the Legendre addition theorem, and for equal spins and $y=x$ it gives $\sum_m\lvert\sw{s}{Y}_{\ell m}(x)\rvert^2=(2\ell+1)/(4\pi)$.

\section{Proofs}
\label{app:proofs}

We collect here the proofs deferred from the main text, in the order in which the results appear.  Throughout, fields on the sphere are square integrable, kernels are square integrable functions on $\SOthree$ unless they are explicitly said to be distributions, and identities between functions on $\SOthree$ or the sphere hold almost everywhere; since $\Psi \conv F$ is continuous for $\Psi, F \in \Lrot$, by the strong continuity of the regular representation, its evaluation on a section is nevertheless well defined.  We use freely that left translation acts on the first Wigner index, $\big(F(\rho^{-1}\,\cdot)\big)^{\ell}_{mn} = \sum_{m'} D^{\ell}_{mm'}(\rho)\,F^{\ell}_{m'n}$, and right translation on the second, $\big(F(\cdot\,\rho)\big)^{\ell}_{mn} = \sum_{n'} F^{\ell}_{mn'}\,D^{\ell}_{n'n}(\rho^{-1})$, both of which follow from \eqref{eq:wigner-transform}, the change of variables $g \mapsto \rho g$ or $g \mapsto g\rho^{-1}$ and the representation property $D^{\ell}(\rho g) = D^{\ell}(\rho)D^{\ell}(g)$.

\subsection{The twist}
\label{app:proofs:twist}

\begin{proof}[Proof of Lemma~\ref{lem:twist-composition}]
Inserting $s(g_2x)\,s(g_2x)^{-1}$ into the definition \eqref{eq:twist} and regrouping,
\begin{align}
  \twist(x,g_1g_2) &= s(g_1g_2x)^{-1}\,g_1g_2\,s(x) \notag\\
  &= \big[s(g_1g_2x)^{-1}\,g_1\,s(g_2x)\big]\big[s(g_2x)^{-1}\,g_2\,s(x)\big]
  = \twist(g_2x,g_1)\,\twist(x,g_2) ,
\end{align}
where the first bracket is the twist of $g_1$ at the point $g_2x$ and the second the twist of $g_2$ at $x$.  On the sphere, for $0<\beta<\pi$, the rotation $R(\alpha,\beta,\gamma)$ sends the North pole to the point $(\beta,\alpha)$, so that $s(R(\alpha,\beta,\gamma)\eta) = R(\alpha,\beta,0)$ and
\begin{equation}
  \twist(R(\alpha,\beta,\gamma)) = R(\alpha,\beta,0)^{-1}R(\alpha,\beta,\gamma) = R_y(\beta)^{-1}R_z(\alpha)^{-1}R_z(\alpha)R_y(\beta)R_z(\gamma) = R_z(\gamma) ,
\end{equation}
as claimed.  At $\beta=0$, the choice $s(\eta)=e$ instead gives $\twist(R(\alpha,0,\gamma))=R_z(\alpha+\gamma)$; at $\beta=\pi$, the choice $s(-\eta)=R_y(\pi)$ gives $R_z(\gamma-\alpha)$.  These endpoint choices do not affect any $L^2$ identity.
\end{proof}

\subsection{The lift, its adjoint and the induced representation}
\label{app:proofs:lift}

\begin{proof}[Proof of Lemma~\ref{lem:lift}]
(a) By Fubini's theorem, $\lVert\lift{s}f\rVert^2 = \frac{1}{2\pi}\int_0^{2\pi}\dd\gamma\int_{\sphere}\lvert f\rvert^2\dmu = \lVert f\rVert^2$, so that $\lift{s}$ is an isometry, and
\begin{align}
  \langle\lift{s}f, F\rangle &= \int_{\sphere}\dmu(\beta,\alpha)\, f(\beta,\alpha)\,\Big[\frac{1}{\sqrt{2\pi}}\int_0^{2\pi}e^{\ii s\gamma}F(\alpha,\beta,\gamma)\dd\gamma\Big]^{*} \notag\\
  &= \langle f, \proj{s}F\rangle
\end{align}
identifies \eqref{eq:lift-proj} as the adjoint.  Next, $\proj{s}\lift{s}f = f$ since $\frac{1}{2\pi}\int_0^{2\pi}e^{\ii s\gamma}e^{-\ii s\gamma}\dd\gamma = 1$.  Since $R(\alpha,\beta,\gamma)R_z(\gamma') = R(\alpha,\beta,\gamma+\gamma')$, the lift $\lift{s}f$ satisfies \eqref{eq:Ms}; conversely, $F \in \Mackey{s}$ has a fibrewise equivariant representative with $F(\alpha,\beta,\gamma) = e^{-\ii s\gamma}F(\alpha,\beta,0)$ for almost every $(\alpha,\beta,\gamma)$, so that $F = \lift{s}(\sqrt{2\pi}\,F\circ s)$ with $\sqrt{2\pi}\,F\circ s \in L^2(\sphere)$ by Fubini's theorem, and the range of $\lift{s}$ is all of $\Mackey{s}$.  Finally,
\begin{equation}
  (\lift{s}\proj{s}F)(\alpha,\beta,\gamma) = e^{-\ii s\gamma}\,\frac{1}{2\pi}\int_0^{2\pi}e^{\ii s\gamma'}F(\alpha,\beta,\gamma')\dd\gamma'
\end{equation}
is the term in $e^{-\ii s\gamma}$ of the Fourier series of $F(\alpha,\beta,\cdot)$ in $\gamma$, which for almost every $(\alpha,\beta)$ is a square integrable function on the circle.  Hence $\lift{s}\proj{s}$ is the orthogonal projector onto $\Mackey{s}$, the subspaces $\Mackey{s}$ are mutually orthogonal, and $\sum_s\lift{s}\proj{s}F = F$ in $\Lrot$ by Plancherel's theorem on the circle and Fubini's theorem, which is the decomposition $\Lrot = \widehat{\bigoplus}_s\Mackey{s}$.

(b) Using the fibrewise equivariant representative of $F \in \Mackey{s}$,
\begin{align*}
  (\proj{s}F)(\beta,\alpha)
  &= \frac{1}{\sqrt{2\pi}}\int_0^{2\pi}e^{\ii s\gamma}e^{-\ii s\gamma}F(\alpha,\beta,0)\dd\gamma \\
  &= \sqrt{2\pi}\,F(\alpha,\beta,0).
\end{align*}

(d) By \eqref{eq:sY-D}, $\lift{s}\sw{s}{Y}_{\ell m} = (2\pi)^{-1/2}e^{-\ii s\gamma}\sw{s}{Y}_{\ell m}(\beta,\alpha) = (-1)^s\sqrt{(2\ell+1)/8\pi^2}\;D^{\ell *}_{m,-s}$, which has unit norm by the orthogonality relation of the Wigner $D$-functions.

(c) By (d) and the orthogonality relation, the Wigner coefficients of $\lift{s}\sw{s}{Y}_{\ell m}$ are $\int\lift{s}\sw{s}{Y}_{\ell m}\,D^{\ell'}_{m'n'}\dmu = (-1)^s\sqrt{8\pi^2/(2\ell+1)}\,\delta_{\ell\ell'}\delta_{mm'}\delta_{n',-s}$, and the first relation of \eqref{eq:lift-harmonic} follows by linearity and the continuity of $\lift{s}$.  For the projection, $(\proj{s}F)_{\ell m} = \langle\proj{s}F,\sw{s}{Y}_{\ell m}\rangle = \langle F,\lift{s}\sw{s}{Y}_{\ell m}\rangle = (-1)^s\sqrt{(2\ell+1)/8\pi^2}\int F\,D^{\ell}_{m,-s}\dmu$, which is the second relation.

(e) This is Proposition~\ref{prop:intertwining}, proved next.
\end{proof}

\begin{proof}[Proof of Proposition~\ref{prop:intertwining}]
Since $(\piG(g)F)(g'h) = F(g^{-1}g'h) = \rho^{(s)}(h^{-1})(\piG(g)F)(g')$ for $F \in \Mackey{s}$, the left regular representation preserves $\Mackey{s}$, and $\pis{s}(g) = \proj{s}\piG(g)\lift{s}$ is therefore a representation: by Lemma~\ref{lem:lift}(a), $\pis{s}(g)\pis{s}(g') = \proj{s}\piG(g)\,\lift{s}\proj{s}\,\piG(g')\lift{s} = \proj{s}\piG(g)\piG(g')\lift{s} = \pis{s}(gg')$, since $\lift{s}\proj{s}$ is the identity on the range of $\piG(g')\lift{s}$, which lies in $\Mackey{s}$.  It is unitary since $\pis{s}(g)^{\dagger} = \proj{s}\piG(g)^{\dagger}\lift{s} = \proj{s}\piG(g^{-1})\lift{s} = \pis{s}(g^{-1})$, and it intertwines, $\lift{s}\pis{s}(g) = \lift{s}\proj{s}\piG(g)\lift{s} = \piG(g)\lift{s}$, again because $\piG(g)\lift{s}f \in \Mackey{s}$; taking adjoints gives $\pis{s}(g)\proj{s} = \proj{s}\piG(g)$.

For the explicit form, write $(\lift{s}f)(g) = (2\pi)^{-1/2}\rho^{(s)}(\twist(g)^{-1})f(gH)$ with $\twist(g) = s(gH)^{-1}g$, which is Definition~\ref{def:lift} by Lemma~\ref{lem:twist-composition}.  Since $\piG(g)\lift{s}f \in \Mackey{s}$, Lemma~\ref{lem:lift}(b) gives
\begin{equation}
  (\pis{s}(g)f)(x) = \sqrt{2\pi}\,(\lift{s}f)(g^{-1}s(x)) = \rho^{(s)}\big(\twist(g^{-1}s(x))^{-1}\big)\,f(g^{-1}x),
\end{equation}
and $\twist(g^{-1}s(x)) = s(g^{-1}x)^{-1}g^{-1}s(x) = \twist(x,g^{-1})$ by \eqref{eq:twist}, which is \eqref{eq:induced-explicit}.  The same conclusion follows from Lemma~\ref{lem:twist-composition} with $x = x_0$, $g_1 = g^{-1}$ and $g_2 = g'$, which splits $\twist(g^{-1}g')$ into $\twist(g'H,g^{-1})\,\twist(g')$.

On the sphere, away from the coordinate singularities, with $g = R_\rho$ and $x = \omega$, the rotation $R_\rho^{-1}R_\omega$ sends the North pole to $R_\rho^{-1}\omega$, so that its Euler angles are $(\alpha',\beta',\vartheta)$ with $(\beta',\alpha') = R_\rho^{-1}\omega$, \ie\ $R_\rho^{-1}R_\omega = R(\alpha',\beta',0)R_z(\vartheta) = s(R_\rho^{-1}\omega)R_z(\vartheta)$.  Hence $\twist(\omega,R_\rho^{-1}) = s(R_\rho^{-1}\omega)^{-1}R_\rho^{-1}s(\omega) = R_z(\vartheta)$, and \eqref{eq:induced-explicit} reads $(\pis{s}(R_\rho)f)(\omega) = e^{-\ii s\vartheta}f(R_\rho^{-1}\omega)$, which is \eqref{eq:spin-rotation}.

Finally, \eqref{eq:rotation-coeffs} follows by lifting: by Lemma~\ref{lem:lift}(c) and the action of left translation on the first Wigner index,
\begin{align}
  (\pis{s}(\rho)f)_{\ell m} &= (-1)^s\sqrt{\tfrac{2\ell+1}{8\pi^2}}\,\big(\piG(\rho)\lift{s}f\big)^{\ell}_{m,-s}
  = (-1)^s\sqrt{\tfrac{2\ell+1}{8\pi^2}}\sum_{m'}D^{\ell}_{mm'}(\rho)\,(\lift{s}f)^{\ell}_{m',-s} \notag\\
  &= \sum_{m'}D^{\ell}_{mm'}(\rho)\,\sw{s}{f}_{\ell m'} ,
\end{align}
the constants cancelling, which holds with the same matrices $D^{\ell}(\rho)$ for every spin.
\end{proof}

\begin{proof}[Proof of Corollary~\ref{cor:fibre-decomposition}]
The decomposition and its uniqueness are Lemma~\ref{lem:lift}(a).  By \eqref{eq:lift-harmonic}, $(\lift{s}\proj{s}F)^{\ell}_{mn} = (-1)^s\sqrt{8\pi^2/(2\ell+1)}\,(\proj{s}F)_{\ell m}\,\delta_{n,-s} = F^{\ell}_{m,-s}\,\delta_{n,-s}$, so that the component of spin $s$ carries exactly the Wigner column $n = -s$; hence $F^{\ell}_{mn} = 0$ for all $\lvert n\rvert \ge N$ if and only if the components with $\lvert s\rvert \ge N$ vanish.
\end{proof}

\subsection{Convolution on the rotation group}
\label{app:proofs:group}

\begin{proof}[Proof of \eqref{eq:group-conv-harmonic}, and of the equivariance, boundedness and adjoint of $\Psi\conv$]
Substituting $\rho = \rho' h^{-1}$, \ie\ $h = \rho^{-1}\rho'$, in the Wigner transform of \eqref{eq:group-conv}, which is permitted since the Haar measure is bi-invariant, and expanding $D^{\ell}_{mn}(\rho'h^{-1}) = \sum_k D^{\ell}_{mk}(\rho')\,D^{\ell *}_{nk}(h)$ by the representation property and unitarity,
\begin{align}
  (\Psi\conv F)^{\ell}_{mn}
  &= \int_{\SOthree}\!\!\int_{\SOthree} \Psi^{*}(h)\,F(\rho')\,D^{\ell}_{mn}(\rho'h^{-1})\dmu(h)\dmu(\rho') \notag\\
  &= \sum_k \Big[\int_{\SOthree} F(\rho')D^{\ell}_{mk}(\rho')\dmu(\rho')\Big]\Big[\int_{\SOthree}\Psi^{*}(h)\,D^{\ell *}_{nk}(h)\dmu(h)\Big]
  = \sum_k F^{\ell}_{mk}\,\Psi^{\ell *}_{nk} ,
\end{align}
where the interchange of integrals is justified by Fubini's theorem since $\Psi, F \in \Lrot \subset L^1(\SOthree)$.  For equivariance,
\begin{align}
  (\Psi\conv\piG(g)F)(\rho) &= \int\Psi^{*}(\rho^{-1}\rho')\,F(g^{-1}\rho')\dmu(\rho') \notag\\
  &= \int\Psi^{*}(\rho^{-1}g\rho'')\,F(\rho'')\dmu(\rho'') = (\Psi\conv F)(g^{-1}\rho) .
\end{align}
For boundedness, in the notation $F^{\ell}$ for the $(2\ell+1)\times(2\ell+1)$ block of coefficients, \eqref{eq:group-conv-harmonic} reads $(\Psi\conv F)^{\ell} = F^{\ell}\Psi^{\ell\dagger}$, hence by Parseval's identity $\lVert\Psi\conv F\rVert^2 = \sum_\ell\frac{2\ell+1}{8\pi^2}\lVert F^{\ell}\Psi^{\ell\dagger}\rVert_{\mathrm{F}}^2 \le \sup_\ell\lVert\Psi^{\ell}\rVert^2_{\mathrm{op}}\,\lVert F\rVert^2$, and $\lVert\Psi^{\ell}\rVert^2_{\mathrm{op}} \le \lVert\Psi^{\ell}\rVert^2_{\mathrm{F}} \le \frac{8\pi^2}{2\ell+1}\lVert\Psi\rVert^2$, where $\lVert\cdot\rVert_{\mathrm{F}}$ denotes the Frobenius norm.  For the adjoint, by Fubini's theorem
\begin{equation}
  \langle\Psi\conv F, G\rangle = \int_{\SOthree} F(g')\,\Big[\int_{\SOthree}\Psi(g^{-1}g')\,G(g)\dmu(g)\Big]^{*}\dmu(g') ,
\end{equation}
and $\int\Psi(g^{-1}g')G(g)\dmu(g) = \int\widetilde\Psi^{*}(g'^{-1}g)\,G(g)\dmu(g) = (\widetilde\Psi\conv G)(g')$ with $\widetilde\Psi(k) = \Psi^{*}(k^{-1})$, so that $(\Psi\conv)^{\dagger} = \widetilde\Psi\conv$.  As a check, the Wigner coefficients of the flipped kernel are $\widetilde\Psi^{\ell}_{mn} = \int\Psi^{*}(g^{-1})D^{\ell}_{mn}(g)\dmu(g) = \int\Psi^{*}(g)\,D^{\ell *}_{nm}(g)\dmu(g) = \Psi^{\ell *}_{nm}$, \ie\ $\widetilde\Psi^{\ell} = \Psi^{\ell\dagger}$, and $(\widetilde\Psi\conv G)^{\ell} = G^{\ell}\Psi^{\ell}$ is indeed the adjoint of $F^{\ell} \mapsto F^{\ell}\Psi^{\ell\dagger}$ with respect to the Frobenius inner product.
\end{proof}

\subsection{Harmonic form, kernel constraint and composition}
\label{app:proofs:harmonic}

\begin{proof}[Proof of Theorem~\ref{thm:harmonic}]
By \eqref{eq:lift-harmonic}, $(\lift{s}\sw{s}{f})^{\ell}_{mm'} = (-1)^s\sqrt{8\pi^2/(2\ell+1)}\,\sw{s}{f}_{\ell m}\delta_{m',-s}$, so that the sum over $m'$ in \eqref{eq:group-conv-harmonic} collapses to the single term $m' = -s$, which gives \eqref{eq:harmonic-sG}.  Applying $\proj{s'}$ by \eqref{eq:lift-harmonic} reads the column $n = -s'$ and multiplies by $(-1)^{s'}\sqrt{(2\ell+1)/8\pi^2}$, so that the degree-dependent constants cancel and \eqref{eq:harmonic-ss} follows; \eqref{eq:harmonic-Gs} is $\proj{s'}$ applied to \eqref{eq:group-conv-harmonic}, and \eqref{eq:harmonic-GG} is \eqref{eq:group-conv-harmonic} itself.  Degrees $\ell < \lvert s\rvert$ are absent from $\Lsph{s}$ and degrees $\ell < \lvert s'\rvert$ from $\Lsph{s'}$, whence the restriction on $\ell$.  Boundedness follows from that of $\lift{s}$, $\proj{s'}$ and $\Psi\conv$.  For equivariance,
\begin{align}
  \Cop{s}{s'}[\Psi]\,\pis{s}(g) &= \proj{s'}\big(\Psi\conv\lift{s}\pis{s}(g)\,\cdot\big) = \proj{s'}\big(\Psi\conv\piG(g)\lift{s}\,\cdot\big) \notag\\
  &= \proj{s'}\piG(g)\big(\Psi\conv\lift{s}\,\cdot\big) = \pis{s'}(g)\,\Cop{s}{s'}[\Psi]
\end{align}
by Lemma~\ref{lem:lift}(e) and the equivariance of $\Psi\conv$; the other three cases drop the lift, the projection, or both.  The adjoint identities follow from $(\Psi\conv)^{\dagger} = \widetilde\Psi\conv$ and $(\lift{s})^{\dagger} = \proj{s}$: for instance $(\Cop{s}{G}[\Psi])^{\dagger} = (\Psi\conv\lift{s})^{\dagger} = \proj{s}\,\widetilde\Psi\conv = \Cop{G}{s}[\widetilde\Psi]$.
\end{proof}

\begin{proof}[Proof of Corollary~\ref{cor:constraint}]
By Lemma~\ref{lem:lift}(a) and (c), a function on $\SOthree$ lies in $\Mackey{s'}$ if and only if its Wigner coefficients vanish outside the column $n = -s'$.  By \eqref{eq:harmonic-sG}, the coefficients of $\Psi\conv\lift{s}f$ are $(-1)^s\sqrt{8\pi^2/(2\ell+1)}\,\sw{s}{f}_{\ell m}\,\Psi^{\ell *}_{n,-s}$, which vanish for all $n \ne -s'$ and all $f$ if and only if $\Psi^{\ell}_{n,-s} = 0$ for all $n \ne -s'$ and all $\ell \ge \lvert s\rvert$, as is seen by taking $f = \sw{s}{Y}_{\ell m}$; this is the equivalence of (i) and (ii).  For (iii), the Wigner coefficients of the right-isotypic part are
\begin{align}
  \big(\Psi^{(s)}\big)^{\ell}_{nm'} &= \frac{1}{2\pi}\int_0^{2\pi}e^{\ii s\gamma}\int_{\SOthree}\Psi(gR_z(\gamma))\,D^{\ell}_{nm'}(g)\dmu(g)\dd\gamma \notag\\
  &= \frac{1}{2\pi}\int_0^{2\pi}e^{\ii(s+m')\gamma}\dd\gamma\;\Psi^{\ell}_{nm'} = \delta_{m',-s}\,\Psi^{\ell}_{nm'} ,
\end{align}
by the substitution $g' = gR_z(\gamma)$ and $D^{\ell}_{nm'}(g'R_z(-\gamma)) = D^{\ell}_{nm'}(g')e^{\ii m'\gamma}$: the right-isotypic part is the column $m' = -s$ of $\Psi$, which is exactly the part read by \eqref{eq:harmonic-sG}.  Since $D^{\ell *}_{nm'}(R_z(\chi_1)gR_z(\chi_2)) = e^{\ii n\chi_1}\,D^{\ell *}_{nm'}(g)\,e^{\ii m'\chi_2}$, the Wigner expansion of $\Psi^{(s)}$ satisfies \eqref{eq:two-sided} if and only if its coefficients vanish unless $n = -s'$ and $m' = -s$, which is (ii).  In that case $\Psi\conv\lift{s}f \in \Mackey{s'}$ and Lemma~\ref{lem:lift}(b) gives $\Cop{s}{s'}[\Psi]f = \sqrt{2\pi}\,(\Psi\conv\lift{s}f)\circ s$; the Wigner expansion \eqref{eq:wigner-transform} restricted to the entry $(-s',-s)$, with $D^{\ell *}_{-s',-s}(\alpha,\beta,\gamma) = e^{-\ii s'\alpha}d^{\ell}_{-s',-s}(\beta)e^{-\ii s\gamma}$, is \eqref{eq:constrained-kernel}.  Finally, for any $\Psi=\lift{s}\psi$, \eqref{eq:lift-harmonic} gives $\Psi^{\ell}_{-s',-s} = (-1)^s\sqrt{8\pi^2/(2\ell+1)}\,\sw{s}{\psi}_{\ell,-s'}$, and substituting into \eqref{eq:harmonic-ss} yields \eqref{eq:harmonic-ss-lifted}.
\end{proof}

\begin{proof}[Proof of Proposition~\ref{prop:composition}]
In the block notation of the proof of \eqref{eq:group-conv-harmonic}, $(\Psi_2\conv(\Psi_1\conv F))^{\ell} = F^{\ell}\Psi_1^{\ell\dagger}\Psi_2^{\ell\dagger} = F^{\ell}(\Psi_2^{\ell}\Psi_1^{\ell})^{\dagger}$, which is \eqref{eq:composition-GG} provided a kernel $\Psi_3 \in \Lrot$ with blocks $\Psi_2^{\ell}\Psi_1^{\ell}$ exists; it does, namely $\Psi_3 = \widetilde\Psi_1\conv\Psi_2$, since $(\widetilde\Psi_1\conv\Psi_2)^{\ell} = \Psi_2^{\ell}\widetilde\Psi_1^{\ell\dagger} = \Psi_2^{\ell}\Psi_1^{\ell}$ by $\widetilde\Psi_1^{\ell} = \Psi_1^{\ell\dagger}$, and $\Psi_3 \in \Lrot$ by the boundedness of $\widetilde\Psi_1\conv$.  The cases with an outer lift or projection, or both, follow by composing with $\lift{s}$ on the right or $\proj{s''}$ on the left.  When the intermediate stage is a spin-$s'$ field, the operator $\lift{s'}\proj{s'}$ is inserted between the two convolutions; by Corollary~\ref{cor:fibre-decomposition} its action on Wigner coefficients is $F^{\ell} \mapsto F^{\ell}E_{-s'}$, with $E_{-s'}$ the diagonal matrix with a single unit entry at $-s'$, so that the composed coefficients are $F^{\ell}\Psi_1^{\ell\dagger}E_{-s'}\Psi_2^{\ell\dagger} = F^{\ell}(\Psi_2^{\ell}E_{-s'}\Psi_1^{\ell})^{\dagger}$, and the block $\Psi_2^{\ell}E_{-s'}\Psi_1^{\ell}$ has entries $(\Psi_2)^{\ell}_{n,-s'}(\Psi_1)^{\ell}_{-s',m'}$.  Reading the entry $(-s'',-s)$ gives the spin-to-spin case, with $c^{(3)}_\ell = (-1)^{s+s''}\big[(\Psi_2)^{\ell}_{-s'',-s'}(\Psi_1)^{\ell}_{-s',-s}\big]^{*} = c^{(2)}_\ell c^{(1)}_\ell$ by \eqref{eq:harmonic-ss}, and reading the column $m' = -s$ gives the remaining case.
\end{proof}

\subsection{Classification}
\label{app:proofs:classification}

\begin{proof}[Proof of Theorem~\ref{thm:classification}]
By Lemma~\ref{lem:lift}(a) and (d), $\lift{s}$ is a unitary equivalence between $\Lsph{s}$ and $\Mackey{s} = \widehat{\bigoplus}_{\ell\ge\lvert s\rvert}V^{(s)}_\ell$, where $V^{(s)}_\ell = \mathrm{span}\{D^{\ell *}_{m,-s} : \lvert m\rvert\le\ell\}$ is the image of the degree-$\ell$ spin harmonics, and by the action of left translation on the first Wigner index, $\piG(\rho)\,D^{\ell *}_{m,-s} = \sum_k D^{\ell}_{km}(\rho)\,D^{\ell *}_{k,-s}$, so that each $V^{(s)}_\ell$ is invariant and carries the irreducible representation of $\SOthree$ of degree $\ell$, in the same matrices for every $s$; this is \eqref{eq:rotation-coeffs}.  Thus $\Lsph{s}$ decomposes into irreducible representations without multiplicity, each degree $\ell \ge \lvert s\rvert$ occurring exactly once.  This is Frobenius reciprocity \cite{mackey1952induced,folland2016course} in the present case: the multiplicity of the degree-$\ell$ representation in $\Ind_{\SOtwo}^{\SOthree}\rho^{(s)}$ equals the multiplicity of $\rho^{(s)}$ in the restriction of the degree-$\ell$ representation to $\SOtwo$, which is one for $\lvert s\rvert \le \ell$ and zero otherwise since $D^{\ell}_{mn}(R_z(\gamma)) = e^{-\ii n\gamma}\delta_{mn}$ contains each character with $\lvert n\rvert \le \ell$ exactly once.

Let $T$ be bounded and equivariant.  Since irreducible representations of different degree are inequivalent, $T$ maps the degree-$\ell$ component of $\Lsph{s}$ into the degree-$\ell$ component of $\Lsph{s'}$, and annihilates it if $\ell < \lvert s'\rvert$.  For $\ell \ge \max(\lvert s\rvert,\lvert s'\rvert)$ the two components are equivalent irreducible representations, realised in the bases $\sw{s}{Y}_{\ell m}$ and $\sw{s'}{Y}_{\ell m}$ by the same matrices $D^{\ell}(\rho)$, so that by Schur's lemma the restriction of $T$ is a scalar $c_\ell$ times the identity in these bases, which is \eqref{eq:classification}; the sequence is unique, and $\lVert T\rVert = \sup_\ell\lvert c_\ell\rvert$ by Parseval's identity, so that $T$ is bounded if and only if the sequence is.  Conversely, a bounded sequence defines a bounded $T$ by \eqref{eq:classification}, which commutes with $\pis{s}$ by \eqref{eq:rotation-coeffs}.

(a) If $T = \Cop{s}{s'}[\Psi]$ with $\Psi \in \Lrot$, then $c_\ell = (-1)^{s+s'}\Psi^{\ell *}_{-s',-s}$ by \eqref{eq:harmonic-ss}, and $\sum_\ell\frac{2\ell+1}{8\pi^2}\lvert c_\ell\rvert^2 \le \lVert\Psi\rVert^2 < \infty$ by Parseval's identity.  Conversely, if $\sum_\ell(2\ell+1)\lvert c_\ell\rvert^2 < \infty$, the function $\Psi_T$ with Wigner coefficients $\Psi_T{}^{\ell}_{-s',-s} = (-1)^{s+s'}c^{*}_{\ell}$ and all others zero lies in $\Lrot$, with $\lVert\Psi_T\rVert^2 = \sum_\ell\frac{2\ell+1}{8\pi^2}\lvert c_\ell\rvert^2$, and $\Cop{s}{s'}[\Psi_T] = T$ by \eqref{eq:harmonic-ss}.  Since the Hilbert--Schmidt norm of $T$ is $\lVert T\rVert^2_{\mathrm{HS}} = \sum_{\ell m}\lvert c_\ell\rvert^2 = \sum_\ell(2\ell+1)\lvert c_\ell\rvert^2$, the condition is that $T$ be Hilbert--Schmidt.  Any $\Psi$ with $\Cop{s}{s'}[\Psi] = T$ has the same entries $\Psi^{\ell}_{-s',-s}$, all other entries being unread, and $\Psi_T$ is the unique such kernel whose entire Wigner block is supported on that entry, equivalently whose entire kernel satisfies the two-sided law of \eqref{eq:two-sided}.  Constraining only $\Psi^{(s)}$ would still permit arbitrary unread input columns and would not give uniqueness.

(b) For a bounded sequence, $\Psi_T$ is a distribution on $\SOthree$ with bounded Wigner coefficients, and $\Cop{s}{s'}[\Psi_T]$ is defined by \eqref{eq:harmonic-ss} on the dense subspace of band-limited fields, on which it coincides with $T$; since $\lvert c_\ell\rvert$ is bounded it extends by continuity to $\Lsph{s}$, and the extension is $T$.  This is the sense in which every bounded equivariant map is a lifted convolution with a distributional kernel.

(c) The left regular representation decomposes as $\Lrot = \widehat{\bigoplus}_\ell\bigoplus_{\lvert n\rvert\le\ell}V_\ell^{(n)}$, with $V^{(n)}_\ell = \mathrm{span}\{D^{\ell *}_{mn} : \lvert m\rvert\le\ell\}$ the $n$-th Wigner column, each column carrying the degree-$\ell$ representation in the same matrices; the degree-$\ell$ representation thus occurs with multiplicity $2\ell+1$, one copy per column.  Schur's lemma gives a vector of coefficients $u^\ell$ from a spin field into the Wigner columns, a vector $v^\ell$ from those columns into a spin field, and an arbitrary matrix $B^\ell$ acting on the multiplicity index for a group-to-group map, as in \eqref{eq:classification-blocks}.  Conversely these coefficient maps commute degree by degree with the action on the first index.  For the first case Parseval's identity gives
\begin{equation}
  \lVert T_{s\to G}f\rVert^2
  = \sum_{\ell\ge\lvert s\rvert} a_\ell\lVert u^\ell\rVert_2^2\sum_m\lvert\sw{s}{f}_{\ell m}\rvert^2 .
\end{equation}
For the second, Cauchy--Schwarz gives
\begin{equation}
  \lVert T_{G\to s'}F\rVert^2
  \le \sup_{\ell\ge\lvert s'\rvert}\frac{\lVert v^\ell\rVert_2^2}{a_\ell}
  \sum_\ell a_\ell\lVert F^\ell\rVert_{\mathrm{F}}^2 ,
\end{equation}
with equality on a single row of a single block proportional to $(v^\ell)^*$; taking degrees approaching the supremum gives the exact norm.  Similarly, $\lVert F^\ell B^\ell\rVert_{\mathrm{F}}\le\lVert F^\ell\rVert_{\mathrm{F}}\lVert B^\ell\rVert_{\mathrm{op}}$, with equality attained on a suitable row, gives the third norm in \eqref{eq:classification-norms}.  Thus those uniform bounds are necessary and sufficient for extension from finite harmonic sums to the Hilbert direct sums; entrywise bounds alone do not suffice.

Comparing with Theorem~\ref{thm:harmonic} gives \eqref{eq:classification-kernels}.  The norm bounds make the canonical kernel blocks uniformly bounded in operator norm, and hence define distributions.  Summing the squared images of orthonormal bases, with the factor $a_\ell^{-1/2}$ for a unit-norm Wigner coefficient, gives
\begin{align}
  \lVert T_{s\to s'}\rVert_{\mathrm{HS}}^2 &= \sum_\ell d_\ell\lvert c_\ell\rvert^2, &
  \lVert T_{s\to G}\rVert_{\mathrm{HS}}^2 &= \sum_\ell d_\ell a_\ell\lVert u^\ell\rVert_2^2, \notag\\
  \lVert T_{G\to s'}\rVert_{\mathrm{HS}}^2 &= \sum_\ell \frac{d_\ell}{a_\ell}\lVert v^\ell\rVert_2^2, &
  \lVert T_{G\to G}\rVert_{\mathrm{HS}}^2 &= \sum_\ell d_\ell\lVert B^\ell\rVert_{\mathrm{F}}^2 .
  \label{eq:classification-hs}
\end{align}
Each sum is over the degrees present in both spaces.  By Parseval's identity each is $8\pi^2\lVert\Psi_T\rVert^2$ for the canonical kernel with unread entries set to zero, so a square integrable representative exists exactly when the map is Hilbert--Schmidt.  Adding unread coefficients cannot remove a divergence in the canonical part.
\end{proof}

For completeness, the unbounded extension \eqref{eq:multiplier-domain} is dense because it contains all finite harmonic sums.  If $f_j\to f$ and $Tf_j\to h$ in the respective $L^2$ spaces, convergence of every coefficient gives $h_{\ell m}=c_\ell f_{\ell m}$; since $h\in L^2$, $f\in\mathcal D(T)$ and $Tf=h$, proving closedness.  Rotations preserve each degree norm and hence the domain, and commute with its scalar multiplier.  Polynomial growth gives a distributional kernel by the rapid decay of smooth test-function coefficients; arbitrary growth need not do so.  The first- and second-order Sobolev assertions follow from the respective growth rates $\lvert c_\ell\rvert\asymp\ell$ and $\ell^2$, with finitely many low modes immaterial to the domain.

\begin{proof}[Proof of Corollary~\ref{cor:projection-general}]
This is Theorem~\ref{thm:classification}(c), the row and column cases.  By \eqref{eq:lift-harmonic}, $\proj{s'}$ reads the column $n = -s'$ with the constant $(-1)^{s'}\sqrt{(2\ell+1)/8\pi^2}$, which is \eqref{eq:harmonic-Gs} with $\Psi^{\ell *}_{-s',m'} = \delta_{m',-s'}$, and $\lift{s}$ is \eqref{eq:harmonic-sG} with $\Psi^{\ell *}_{n,-s} = \delta_{n,-s}$.
\end{proof}

\begin{proof}[Proof of Corollary~\ref{cor:nogo}]
By \eqref{eq:classification},
\begin{align}
  (Tf)(x) &= \sum_\ell c_\ell\sum_m\sw{s'}{Y}_{\ell m}(x)\,\langle f,\sw{s}{Y}_{\ell m}\rangle = \int_{\sphere}K(x,y)\,f(y)\dmu(y), \notag\\
  K(x,y) &= \sum_\ell c_\ell\sum_m\sw{s'}{Y}_{\ell m}(x)\,\sw{s}{Y}^{*}_{\ell m}(y) ,
\end{align}
the series converging in the sense of distributions.  By \eqref{eq:sY-D} at $\gamma = 0$, $\sw{s'}{Y}_{\ell m}(x) = (-1)^{s'}\sqrt{(2\ell+1)/4\pi}\,D^{\ell *}_{m,-s'}(s(x))$ and $\sw{s}{Y}^{*}_{\ell m}(y) = (-1)^{s}\sqrt{(2\ell+1)/4\pi}\,D^{\ell}_{m,-s}(s(y))$, so that, by $D^{\ell *}_{mn}(g) = D^{\ell}_{nm}(g^{-1})$ and the representation property,
\begin{align}
  \sum_m\sw{s'}{Y}_{\ell m}(x)\,\sw{s}{Y}^{*}_{\ell m}(y) &= (-1)^{s+s'}\frac{2\ell+1}{4\pi}\sum_m D^{\ell}_{-s',m}\big(s(x)^{-1}\big)\,D^{\ell}_{m,-s}\big(s(y)\big) \notag\\
  &= (-1)^{s+s'}\frac{2\ell+1}{4\pi}\,D^{\ell}_{-s',-s}\big(s(x)^{-1}s(y)\big) ,
\end{align}
which is the addition theorem for spin spherical harmonics and gives \eqref{eq:nogo-kernel}.  The rotation $s(x)^{-1}s(y)$ sends the North pole to $s(x)^{-1}y$, the point $y$ expressed in the frame of $x$, whose colatitude is the geodesic distance $\beta_{xy}$ between $x$ and $y$; writing its Euler angles as $(\alpha_{xy},\beta_{xy},\gamma_{xy})$, the kernel reads $D^{\ell}_{-s',-s}(\alpha_{xy},\beta_{xy},\gamma_{xy}) = e^{\ii s'\alpha_{xy}}\,d^{\ell}_{-s',-s}(\beta_{xy})\,e^{\ii s\gamma_{xy}}$, in which only the $d$-function depends on the distance and the two exponentials are the phases of a kernel with output spin $s'$ in $x$ and dual input spin $-s$ in $y$: under a frame change, $K'(x,y)=e^{-\ii s'\chi(x)}K(x,y)e^{+\ii s\chi(y)}$ by \eqref{eq:section-change}.  For $s = s' = 0$, $d^{\ell}_{00}(\beta) = P_\ell(\cos\beta)$ and $\cos\beta_{xy} = x\cdot y$.
\end{proof}

\subsection{Reductions}
\label{app:proofs:reductions}

\begin{proof}[Proof of Proposition~\ref{prop:directional}]
By Definition~\ref{def:group-conv}, $(\lift{s}\psi\conv\lift{s}f)(g) = \langle\lift{s}f, \piG(g)\lift{s}\psi\rangle = \langle\lift{s}f, \lift{s}\pis{s}(g)\psi\rangle = \langle f, \pis{s}(g)\psi\rangle$, by Lemma~\ref{lem:lift}(e) and the isometry of the lift, and $\pis{s}(g) = \Rot_g$ by Proposition~\ref{prop:intertwining}, so that this is \eqref{eq:spin-directional}.  In harmonic space, \eqref{eq:harmonic-sG} with $\Psi = \lift{s}\psi$ and \eqref{eq:lift-harmonic} give $(-1)^s\sqrt{8\pi^2/(2\ell+1)}\;\sw{s}{f}_{\ell m}\;\big[(-1)^s\sqrt{8\pi^2/(2\ell+1)}\;\sw{s}{\psi}_{\ell n}\big]^{*} = \frac{8\pi^2}{2\ell+1}\,\sw{s}{f}_{\ell m}\,\sw{s}{\psi}^{*}_{\ell n}$.
\end{proof}

\begin{proof}[Proof of Proposition~\ref{prop:synthesis}]
Since $\pis{s}(\rho)\psi = \proj{s}\piG(\rho)\lift{s}\psi$ and $\proj{s}$ is bounded,
\begin{equation}
  \int_{\SOthree} W(\rho)\,\pis{s}(\rho)\psi\dmu(\rho) = \proj{s}\int_{\SOthree} W(\rho)\,\piG(\rho)\Phi\dmu(\rho),
\end{equation}
where $\Phi = \lift{s}\psi$ and
the integrals are Bochner integrals in $\Lrot$.  Evaluating at $g$ and noting that $\widetilde\Phi^{*}(g^{-1}\rho) = \Phi(\rho^{-1}g)$,
\begin{equation}
  \int_{\SOthree} W(\rho)\,\Phi(\rho^{-1}g)\dmu(\rho) = \int_{\SOthree}\widetilde\Phi^{*}(g^{-1}\rho)\,W(\rho)\dmu(\rho) = (\widetilde\Phi\conv W)(g) ,
\end{equation}
so that the synthesis is $\Cop{G}{s}[\widetilde\Phi]W$, which is $(\Cop{s}{G}[\Phi])^{\dagger}W$ by Theorem~\ref{thm:harmonic}.  In harmonic space, $\widetilde\Phi^{\ell}_{-s,m'} = \Phi^{\ell *}_{m',-s} = (-1)^s\sqrt{8\pi^2/(2\ell+1)}\,\sw{s}{\psi}^{*}_{\ell m'}$ by \eqref{eq:lift-harmonic}, so that \eqref{eq:harmonic-Gs} gives
\begin{equation}
  (\Cop{G}{s}[\widetilde\Phi]W)_{\ell m} = (-1)^s\sqrt{\tfrac{2\ell+1}{8\pi^2}}\sum_{m'}W^{\ell}_{mm'}\,(-1)^s\sqrt{\tfrac{8\pi^2}{2\ell+1}}\,\sw{s}{\psi}_{\ell m'} = \sum_{m'}W^{\ell}_{mm'}\,\sw{s}{\psi}_{\ell m'} ,
\end{equation}
the constants cancelling.
\end{proof}

\begin{proof}[Proof of Proposition~\ref{prop:axisymmetric}]
By Proposition~\ref{prop:directional} evaluated at $g = s(\omega) = R_\omega$, $\langle f,\Rot_\omega\psi\rangle = (\Cop{0}{G}[\lift{0}\psi]f)(s(\omega))$.  For an axisymmetric kernel, $\psi_{\ell m} = \psi_{\ell 0}\delta_{m0}$, the kernel $\lift{0}\psi$ has, by \eqref{eq:lift-harmonic}, Wigner support at the single entry $(0,0)$ of each degree, so that it satisfies condition (ii) of Corollary~\ref{cor:constraint} with $s = s' = 0$, and hence $\Cop{0}{G}[\lift{0}\psi]f \in \Mackey{0}$; Lemma~\ref{lem:lift}(b) then gives $(\Cop{0}{G}[\lift{0}\psi]f)\circ s = (2\pi)^{-1/2}\proj{0}\Cop{0}{G}[\lift{0}\psi]f = (2\pi)^{-1/2}\Cop{0}{0}[\lift{0}\psi]f$.  The harmonic form follows from \eqref{eq:harmonic-ss-lifted} with $s = s' = 0$, divided by $\sqrt{2\pi}$: $(2\pi)^{-1/2}\sqrt{8\pi^2/(2\ell+1)} = \sqrt{4\pi/(2\ell+1)}$.  The spin case is identical with $s' = 0$: the kernel $\lift{s}\Phi$ with $\sw{s}{\Phi}_{\ell m} = \sw{s}{\Phi}_{\ell 0}\delta_{m0}$ satisfies (ii) with target spin $0$, and \eqref{eq:harmonic-ss-lifted} gives \eqref{eq:spin-axisymmetric-harmonic}.  If $\psi$ is not axisymmetric the constraint fails and the section evaluation is the standard correlation, treated in Proposition~\ref{prop:section-evaluation}.
\end{proof}

\begin{proof}[Proof of Proposition~\ref{prop:driscoll-healy}]
With $\Psi_h(g) = h^{*}(g^{-1}\eta)$ and $k = g^{-1}g'$, $\Psi^{*}_h(g^{-1}g') = h(k^{-1}\eta) = h(g'^{-1}g\eta)$, so that by Definition~\ref{def:group-conv}
\begin{equation}
  (\Psi_h\conv\lift{0}f)(g) = \frac{1}{\sqrt{2\pi}}\int_{\SOthree}h(g'^{-1}g\eta)\,f(g'\eta)\dmu(g') .
\end{equation}
Applying $\proj{0}$ at the point $\omega = (\beta,\alpha)$ integrates over $\gamma$ with $g = R(\alpha,\beta,\gamma)$, and since $g\eta = \omega$ does not depend on $\gamma$ the integrand is constant, whence $(\Cop{0}{0}[\Psi_h]f)(\omega) = \frac{1}{2\pi}\cdot 2\pi\int h(g'^{-1}\omega)f(g'\eta)\dmu(g') = (f\star_{\mathrm{L}}h)(\omega)$, exactly.  For the constant, only the entry $(\Psi_h)^{\ell}_{00}$ is read by \eqref{eq:harmonic-ss}.  Writing $g = R(\alpha,\beta,\gamma)$, the point $g^{-1}\eta = R_z(-\gamma)R_y(-\beta)\eta$ has colatitude $\beta$ and longitude $\pi - \gamma$, and $D^{\ell}_{00}(g) = d^{\ell}_{00}(\beta) = P_\ell(\cos\beta) = \sqrt{4\pi/(2\ell+1)}\,Y_{\ell 0}(\beta,\cdot)$, so that
\begin{align}
  (\Psi_h)^{\ell}_{00} &= \int_{\SOthree}h^{*}(\beta,\pi-\gamma)\,\sqrt{\tfrac{4\pi}{2\ell+1}}\,Y_{\ell 0}(\beta)\,\sin\beta\dd\alpha\dd\beta\dd\gamma \notag\\
  &= 2\pi\sqrt{\tfrac{4\pi}{2\ell+1}}\int_{\sphere}h^{*}(\omega')\,Y_{\ell 0}(\omega')\dmu(\omega')
  = 2\pi\sqrt{\tfrac{4\pi}{2\ell+1}}\;h^{*}_{\ell 0} ,
\end{align}
where the integral over $\alpha$ gives $2\pi$ and the substitution $\varphi' = \pi - \gamma$ runs over a full period.  Hence \eqref{eq:harmonic-ss} with $s = s' = 0$ gives $(f\star_{\mathrm{L}}h)_{\ell m} = f_{\ell m}\,(\Psi_h)^{\ell *}_{00} = 2\pi\sqrt{4\pi/(2\ell+1)}\,f_{\ell m}h_{\ell 0}$, which is \eqref{eq:driscoll-healy-harmonic}, with no conjugate on $h$ since the pairing \eqref{eq:driscoll-healy} is bilinear.  Finally, the marginalisation over $\gamma$ of the bilinear directional convolution $\int f(\omega)(\Rot_\rho h)(\omega)\dmu(\omega)$, which is a function on $\SOthree$, is by definition $\sqrt{2\pi}$ times its projection $\proj{0}$.
\end{proof}

\begin{proof}[Proof of Proposition~\ref{prop:band-limited}]
By \eqref{eq:harmonic-sG} the Wigner coefficients of $\Cop{s}{G}[\Psi]f$ vanish for $\lvert n\rvert \ge N$ when the column $m' = -s$ of $\Psi$ does.  Corollary~\ref{cor:fibre-decomposition} applied to the output gives $\Cop{s}{G}[\Psi]f = \sum_{\lvert s'\rvert<N}\lift{s'}\proj{s'}\Cop{s}{G}[\Psi]f$, and $\proj{s'}\Cop{s}{G}[\Psi] = \Cop{s}{s'}[\Psi]$ by Definition~\ref{def:lifted-conv}.  Conversely, given at most $2N-1$ bounded multipliers $c^{(s')}_\ell$, the canonical column $\Psi^{\ell}_{-s',-s}=(-1)^{s+s'}c^{(s')*}_{\ell}$ defines a distribution and assembles the tuple by the same identity.  The norm bound is $\sup_\ell(\sum_{s'}\lvert c^{(s')}_\ell\rvert^2)^{1/2}<\infty$.  By Parseval, the kernel belongs to $L^2$ exactly when $\sum_{\ell,s'}(2\ell+1)\lvert c^{(s')}_\ell\rvert^2<\infty$, the sum of the component Hilbert--Schmidt norms squared.  This condition holds automatically at finite band-limit.
\end{proof}

\begin{proof}[Proof of Proposition~\ref{prop:physics}(c)]
Wallis \etal\ \cite{wallis2022mapping} write the shear as
\begin{equation}
  \sw{2}{\gamma}(\omega) = \int_{\sphere}(\Rot_{\omega'}\sw{2}{K})(\omega)\,\kappa(\omega')\dmu(\omega'),
\end{equation}
with an axisymmetric spin-$2$ kernel, $\sw{2}{K}_{\ell m} = \sw{2}{K}_{\ell 0}\delta_{m0}$.  By \eqref{eq:rotation-coeffs} and \eqref{eq:D-sY},
\begin{equation}
  (\Rot_{\omega'}\sw{2}{K})_{\ell m} = D^{\ell}_{m0}(\varphi',\theta',0)\,\sw{2}{K}_{\ell 0} = \sqrt{\tfrac{4\pi}{2\ell+1}}\,Y^{*}_{\ell m}(\omega')\,\sw{2}{K}_{\ell 0} ,
\end{equation}
and integrating against $\kappa$ gives $\sw{2}{\gamma}_{\ell m} = \sqrt{4\pi/(2\ell+1)}\,\sw{2}{K}_{\ell 0}\,\kappa_{\ell m}$, a multiplier from spin $0$ to spin $2$; it is $\Cop{0}{2}[\Psi_T]$ with $c_\ell = \sqrt{4\pi/(2\ell+1)}\,\sw{2}{K}_{\ell 0}$ by Theorem~\ref{thm:classification}, and its kernel is \eqref{eq:nogo-kernel} with $(s,s') = (0,2)$: the integral is the adjoint synthesis of the spin-$2$ to scalar axisymmetric analysis of Proposition~\ref{prop:axisymmetric}, not that forward analysis.  The choice $\sw{2}{K}_{\ell 0} = \sqrt{(2\ell+1)/4\pi}\,\mathcal{D}_\ell$ gives $c_\ell = \mathcal{D}_\ell$, which is the spherical Kaiser--Squires relation.
\end{proof}

\subsection{Beyond the framework}
\label{app:proofs:beyond}

\begin{proof}[Proof of Proposition~\ref{prop:section-evaluation}]
Let $F=\Psi\conv\lift{s}f$.  It is continuous, and $\lVert F\rVert_\infty\le\lVert\Psi\rVert_2\lVert f\rVert_2$ by Cauchy--Schwarz and the isometry of the lift, so section evaluation is a bounded map into the coordinate $L^2$ space.  The explicit support hypothesis $\Psi^\ell_{n,-s}=0$ for $\lvert n\rvert\ge N$ makes its fibre decomposition a finite sum.  Evaluating each term at $\sigma(x)=s(x)R_z(\chi_\sigma(x))$ gives \eqref{eq:section-evaluation}, since $\proj{q}F=\Cop{s}{q}[\Psi]f$.

Suppose $E_\sigma$ is equivariant into spin $t$ with the fixed coordinate-frame action.  By Theorem~\ref{thm:classification} it acts on degree $\ell$ by a scalar $\lambda_\ell$, or annihilates that degree if $\ell<\lvert t\rvert$.  Put $a_{\ell q}=(-1)^{s+q}\Psi^{\ell *}_{-q,-s}/\sqrt{2\pi}$.  Testing on every $\sw{s}{Y}_{\ell m}$ gives, almost everywhere,
\begin{equation}
  \sum_{\substack{\lvert q\rvert<N\\\lvert q\rvert\le\ell}}
  a_{\ell q}e^{-\ii q\chi_\sigma(x)}\sw{q}{Y}_{\ell m}(x)
  =\lambda_\ell\sw{t}{Y}_{\ell m}(x),
  \label{eq:section-evaluation-test}
\end{equation}
with the right-hand side interpreted as zero when $\ell<\lvert t\rvert$.  For a fixed $x$, the vectors $(\sw{q}{Y}_{\ell m}(x))_{m=-\ell}^{\ell}$ for distinct admissible $q$ are nonzero multiples of distinct columns of a unitary Wigner matrix, by \eqref{eq:app-sY-d}.  They are therefore independent.  Since there are countably many tests, a common full-measure set may be used.  Comparing columns in \eqref{eq:section-evaluation-test} forces $a_{\ell q}=0$ for $q\ne t$ and all components to vanish when the target degree is absent.  If some $a_{\ell t}\ne0$, its column gives $e^{-\ii t\chi_\sigma(x)}=\lambda_\ell/a_{\ell t}$ almost everywhere.  Thus the phase is constant.  If every coefficient vanishes, $E_\sigma=0$, which is equivariant into every target spin independently of the section.  Conversely, one active output spin and a constant phase give a scalar multiple of $\Cop{s}{t}$ and hence equivariance.

For polar rotations, $R_z(a)s(x)=s(R_z(a)x)$ away from the poles, so all spins act by coordinate rotation.  The phase for $\sigma_0$ is constant.  For $\chi_\sigma=\pi-\varphi$, one has $\chi_\sigma(R_z(-a)x)=\chi_\sigma(x)+a$ modulo $2\pi$.  Testing harmonics and comparing the same independent columns shows that commutation is equivalent to $a_{\ell q}(1-e^{-\ii qa})=0$ for every $\ell,q$.  This is $e^{\ii qa}=1$ for every active spin $q$, including the vacuous condition for the zero map; it makes no assertion about non-polar symmetries.
\end{proof}

\begin{proof}[Proof of Proposition~\ref{prop:diagonal}]
For a rotation $R_z(a)$ about the polar axis the twist vanishes, so that $(\pis{0}(R_z(a))f)_{\ell m} = e^{-\ii ma}f_{\ell m}$ by \eqref{eq:rotation-coeffs} and \eqref{eq:wignerD}; a diagonal multiplier commutes with this action for every $g_{\ell m}$.  Equivariance under all of $\SOthree$ is, by Theorem~\ref{thm:classification} with $s = s' = 0$, the condition that $T$ act by a scalar on each degree, \ie\ that $g_{\ell m}$ be independent of $m$, in which case $T$ is the axisymmetric convolution $\Cop{0}{0}[\Psi_T]$ with $c_\ell = g_{\ell 0}$, understood in the distributional multiplier sense.  Moreover, $D^\ell(R_y(\pi))$ exchanges $m$ and $-m$ up to the nonzero sign in \eqref{eq:app-d-poles}, so it commutes with the diagonal multiplier when $g_{\ell m}=g_{\ell,-m}$.  Together with the polar rotations this gives the embedded $O(2)$ normaliser.
\end{proof}

\subsection{Coordinate-section standard correlation}
\label{app:proofs:standard-correlation}

The factorisation-of-rotations method \cite{wandelt2001fast,mcewen2007fast,mcewen2013directional} gives an $\mathcal{O}(L^3)$ evaluation of standard correlation on an equiangular coordinate grid.  Set $\Delta^\ell_{km}=d^\ell_{km}(\pi/2)$.  In the conventions of \eqref{eq:wignerD},
\begin{equation}
  d^\ell_{mn}(\theta)=\ii^{m-n}\sum_{k=-\ell}^{\ell}
  \Delta^\ell_{km}\Delta^\ell_{kn}e^{-\ii k\theta} .
  \label{eq:standard-factorisation}
\end{equation}
Substituting this identity in \eqref{eq:directional-harmonic} at $\gamma=0$ and contracting over $n$ before summing over $\ell$ gives
\begin{align}
  G^\ell_k &= \sum_{n=-\ell}^{\ell}\ii^{-n}\Delta^\ell_{kn}\psi^{*}_{\ell n}, \notag\\
  H_{mk} &= \sum_{\ell=\max(\lvert m\rvert,\lvert k\rvert)}^{L-1}
  \ii^m\Delta^\ell_{km}f_{\ell m}G^\ell_k,
  \label{eq:standard-contractions}\\
  \langle f,\Rot_{(\varphi,\theta,0)}\psi\rangle
  &=\sum_{m,k=-(L-1)}^{L-1}H_{mk}e^{\ii m\varphi-\ii k\theta} .
  \label{eq:standard-grid}
\end{align}
Each contraction costs $\mathcal{O}(L^3)$, and a two-dimensional FFT evaluates the last sum on a grid with $\mathcal{O}(L^2)$ points in $\mathcal{O}(L^2\log L)$.  The $\Delta$ blocks may be generated and discarded degree by degree, leaving $\mathcal{O}(L^2)$ working storage; retaining every block instead costs $\mathcal{O}(L^3)$.  These are values of a finite bivariate Fourier sum on the coordinate section, not scalar spherical harmonic coefficients of band-limit $L$.  Indeed $D^1_{01}(\varphi,\theta,0)\propto\sin\theta$ has an infinite scalar harmonic expansion, since an axisymmetric finite harmonic sum is a polynomial in $\cos\theta$.

\bibliography{refs}

\end{document}